\documentclass[mnsc,nonblindrev]{informs3} 

\RequirePackage{bm}
\RequirePackage{endnotes}

\OneAndAHalfSpacedXI 

\usepackage{amsmath,amssymb,amsfonts}

{\end{list}}

\def\b1{{\bf 1}}

\def\blot{\quad {$\vcenter{\vbox{\hrule height.4pt
             \hbox{\vrule width.4pt height.9ex \kern.9ex \vrule
width.4pt}
             \hrule height.4pt}}$}}

\usepackage{graphics}
\usepackage{graphicx}
\usepackage{epstopdf}
\usepackage{mathrsfs}
\usepackage{enumerate}
\usepackage{algorithm}
\usepackage{algpseudocodex}
\usepackage{booktabs}
\usepackage{float}
\usepackage{color}
\usepackage{xcolor}
\usepackage{multirow}
\usepackage{soul}
\usepackage{bibunits}
\usepackage{makecell}
\usepackage{pifont}

\usepackage{tikz}
\usetikzlibrary{calc,fit}
\usepackage[colorlinks=true, linkcolor=blue, citecolor=blue, urlcolor=blue]{hyperref}


\makeatletter
\newcommand{\StatexIndent}{\Statex\hspace{\ALG@thistlm}}
\makeatother

\usepackage{cases}
\usepackage{subfig}

\usepackage{natbib}
 \bibpunct[, ]{(}{)}{,}{a}{}{,}%
 \def\bibfont{\small}%

\algdef{SE}[DOWHILE]{Do}{doWhile}{\algorithmicdo}[1]{\algorithmicwhile\ #1}

\EquationsNumberedThrough    

\TheoremsNumberedThrough     
\ECRepeatTheorems  %

\renewcommand{\theARTICLETOP}{}
\begin{document}
\graphicspath{{figures/}}


\RUNAUTHOR{Sun et al.}

\RUNTITLE{Graph Generation Methods under Partial Information}

\TITLE{Graph Generation Methods under Partial Information}

\ARTICLEAUTHORS{%
\AUTHOR{Tong Sun}
\AFF{School of Management, Harbin Institute of Technology, Harbin, Heilongjiang 150001, China,\\ \EMAIL{24b910002@stu.hit.edu.cn}}

\AUTHOR{Jianshu Hao}
\AFF{Corporate Banking Department, China Construction Bank, Beijing 100033, China,\\ \EMAIL{haojianshu1.zh@ccb.com}}

\AUTHOR{Michael C. Fu}
\AFF{The Robert H. Smith School of Business, Institute for Systems Research, University of Maryland, College Park, MD 20742, USA, \EMAIL{mfu@umd.edu}}

\AUTHOR{Guangxin Jiang\footnotemark}  \footnotetext{Corresponding author}
\AFF{School of Management, Harbin Institute of Technology, Harbin, Heilongjiang 150001, China,\\ \EMAIL{gxjiang@hit.edu.cn}}
} 

\ABSTRACT{%
We study the problem of generating graphs with prescribed degree sequences for bipartite, directed, and undirected networks. We first propose a sequential method for bipartite graph generation and establish a necessary and sufficient interval condition that characterizes the admissible number of connections at each step, thereby guaranteeing global feasibility. Based on this result, we develop bipartite graph enumeration and sampling algorithms suitable for different problem sizes. We then extend these bipartite graph algorithms to the directed and undirected cases by incorporating additional connection constraints, as well as feasibility verification and symmetric connection steps, while preserving the same algorithmic principles. Finally, numerical experiments demonstrate the performance of the proposed algorithms, particularly their scalability to large instances where existing methods become computationally prohibitive.
}%



\maketitle

%

\section{Introduction}\label{sec:Intro}

In many risk management problems, the network structure of interconnected entities is a key determinant of how shocks propagate and whether localized disturbances evolve into systemic events. Certain network structures may magnify contagion, allowing small shocks to cascade broadly across the system, whereas other structures may dampen propagation and support systemic resilience \citep{Acemoglu2012, FU2025}. These mechanisms are central not only in financial systems, where interconnections among institutions can channel distress across markets \citep{Eisenberg2001systemic, Adrian2016CoVaR, Huang2024monte}, but also in supply chains, where production bottlenecks may spread disruptions downstream \citep{Barrot2016, Carvalho2021, Birge2023disruption}. In both settings, the network structure critically shapes systemic risk, making its accurate characterization essential for risk assessment and the design of effective mitigation strategies.

In practice, however, the complete structure of a network is rarely observable. Typically, only partial information is available, from which the possible network structure could be inferred. In this paper, we focus on one of the most common and practically tractable forms of such information, namely the number of connections associated with each node, while the specific counterparts are unobserved. This type of degree information arises naturally in many settings. In financial systems, regulators may observe the number of assets held by each investment firm and the number of firms holding each asset, but not the precise ownership relations \citep{glassman2018}. In supply chains, statistical information such as the number of suppliers and customers for firms may be available, while the identities of individual trading partners are unobserved \citep{Mungo2024}.
In social networks, survey-based data may provide information on the number of social contacts individuals have, while the identities of their specific partners are unobserved \citep{Marsden2014}. This motivates the problem of reconstructing plausible network structures consistent with such partial information, thereby enabling meaningful risk analysis in settings where complete data are unavailable.

When the available partial information is limited to the number of connections associated with each node, the network reconstruction problem can be abstracted as graph generation with prescribed degree sequences. A degree is the number of edges incident to a node, and the degree sequence is the collection of degrees for all nodes in the network. The goal is to generate graphs that exactly satisfy a given degree sequence, subject to the restriction that the graphs are simple, that is, without self-loops or multiple edges. This abstraction directly captures the examples discussed above: asset-firm relationships in financial systems correspond to bipartite graphs in which assets and firms form two disjoint node sets; supply chains can be represented as directed graphs, with edges oriented from suppliers to customers to describe input--output flows; and social networks are commonly modeled as undirected graphs, where edges represent reciprocal social ties between individuals. In this way, generating graphs consistent with prescribed degree sequences provides a unifying framework for analyzing network reconstruction across diverse application domains. 

Previous studies have proposed a variety of methods to address the challenge of network reconstruction under partial information, which can be broadly divided into bipartite, directed, and undirected graphs according to the type of network. For bipartite graphs, \cite{verhelst2008efficient} developed a Markov chain Monte Carlo (MCMC) algorithm for sampling binary matrices with fixed marginals, while \cite{Barvinok2010} established an asymptotic link between uniformly sampled binary matrices and a corresponding maximum entropy solution. \citet{Miller2013} further proposed an exact sampling algorithm for binary contingency tables, achieving provably uniform sampling. For directed graphs,  \cite{gandy2017bayesian} developed a Bayesian MCMC framework to infer and sample weighted directed financial networks from aggregate balance-sheet data, while \cite{squartini2017network} proposed a maximum entropy based method to reconstruct weighted directed networks from partial information, inferring the most likely topology and link weights consistent with known node strengths and sampled link densities. For the undirected graph, \cite{wormald1999models} introduced the switching method for generating random $d$-regular undirected graphs, which preserve degree constraints and uniformly explore the space of feasible configurations, and \cite{fosdick2018configuring} investigated MCMC approaches for generating simple graphs with prescribed degree sequences. Recently, \cite{Glasserman2023} proposed a sequential algorithm using maximum entropy and applied the method to reconstruct bipartite, directed, and undirected graphs.


Existing methods are generally sampling-based and typically rely on MCMC procedures or repeated optimization, which face computational scalability challenges as the problem size increases. In small-scale settings, an efficient exact enumeration method may be preferred, since it generates the full set of feasible graphs and enables exact evaluation of graph statistics and systemic-risk measures. As the problem scale grows, the number of feasible graphs can grow exponentially or even super-exponentially, making exhaustive enumeration  computationally infeasible and motivating the need for scalable sampling algorithms.

In this paper, we develop both enumeration and sampling algorithms of bipartite, directed, and undirected graphs with prescribed degree sequences across different problem sizes. Our framework unifies the treatment of different network types by exploiting the fact that directed and undirected graphs admit equivalent bipartite representations. We therefore begin by focusing on bipartite graphs. Specifically, we first propose a sequential method and derive, using the Gale--Ryser theorem \citep{gale1957theorem}, the necessary and sufficient interval condition for the admissible number of connections at each step in a way that guarantees the feasibility of completing the graph. Based on this method, we design an exact enumeration algorithm that systematically generates all feasible graphs without duplication, suitable for small-scale networks. 

For larger-scale problems, in which the number of feasible graphs may exceed $10^{20}$, exhaustive enumeration is infeasible, and sampling methods are required. In these settings, efficiency entails not only fast generation but also approximate uniformity, ensuring that samples are not unduly concentrated on a small subset of graphs. In the absence of additional prior information beyond the degree sequences, uniform sampling over the feasible set provides a natural criterion. Building on the sequential method, we design two sampling algorithms. The first, the uniform bipartite graph sampling (UBGS) algorithm, determines exact weights at each step by computing the number of distinct completions associated with each feasible value. Although this method achieves uniform sampling, it incurs additional computation and relatively high memory usage due to the need to obtain exact weights. To address this, we further propose the efficient bipartite graph sampling (EBGS) algorithm, which replaces exact weights with approximate ones, thereby improving scalability while maintaining approximately uniform performance.

Having developed enumeration and sampling algorithms for bipartite graphs, we extend our methods to directed and undirected graphs. The key idea is to construct bipartite representations that incorporate additional self-loop and single-degree connection constraints, as well as feasibility verification
and symmetric connection steps, thereby capturing directionality or symmetry and allowing the same algorithmic principles to be applied across different graph types. For directed graphs, each node is split into an “in” and “out” copy, reducing the problem to a bipartite matching between sources and targets subject to prescribed in-degrees and out-degrees. For undirected graphs, we employ a symmetric bipartite representation in which feasibility requires that edges be mirrored across the partition. Within these representations, our approach yields both enumeration and sampling procedures. For small-scale directed or undirected networks, the enumeration algorithm exhaustively generates all feasible graphs without duplication, guaranteeing complete coverage of the solution space. For larger-scale instances of such networks, we adapt the sampling algorithms, using exact or approximate counts of feasible completions to guide node connections. In this way, our framework provides parallel solutions for bipartite, directed, and undirected graphs: exact enumeration for small instances, and either uniform sampling or efficient, albeit slightly biased, sampling for larger-scale instances.

The work most closely related to ours is \cite{Glasserman2023}, which proposes a sequential algorithm for graph reconstruction with prescribed degree sequences based on a maximum entropy distribution. Their approximately uniform sampling method applies to bipartite, directed, and undirected graphs, but it requires repeatedly solving complex optimization problems, thereby causing the computational cost to grow rapidly with problem size. 
Another closely related work is \citet{Miller2013}, which develops an exact sampling algorithm for two-way contingency tables with binary or nonnegative integer entries. Their uniform sampling method does not address either directed or undirected graphs and does not scale to large instances. Moreover, neither of these works considers enumeration algorithms.
In sum, our work makes the following contributions:
\begin{enumerate}[(i)]
    \item We propose a sequential method for bipartite graph generation and establish a necessary and sufficient interval condition that characterizes the admissible number of connections at each step, thereby guaranteeing global feasibility. 

    \item We develop a duplication-free bipartite graph enumeration algorithm that exhaustively generates all feasible graphs, which is practical for small-scale problems, and bipartite graph sampling algorithms that enable uniform sampling or efficient, albeit slightly biased, sampling for larger-scale instances. To the best of our knowledge, this work is the first to develop sampling algorithms that can handle large-scale problems.

    \item By incorporating additional self-loop and single-degree connection constraints, as well as feasibility verification and symmetric connection steps, we extend the bipartite graph enumeration and sampling algorithms to the directed and undirected cases, while preserving the same algorithmic principles and applicability across different problem sizes.

    \item We demonstrate the performance of the proposed algorithms through numerical experiments, particularly their scalability to large instances where existing methods become computationally prohibitive.
\end{enumerate}

The rest of the paper is organized as follows. Section~\ref{sec:problem}
defines the three types of graphs and graph space generation problem. In Section~\ref{sec:BGA}, we present the sequential method for generating bipartite graphs with prescribed degree sequences and establish the necessary and sufficient interval condition required for this method, based on which we develop the bipartite graph enumeration and sampling algorithms. In Section~\ref{sec:Extend}, we extend these algorithms to the directed and undirected cases. Section~\ref{sec:numerical-experiemnts} evaluates the proposed algorithms through numerical experiments. Finally, Section~\ref{sec:conclsion} concludes this paper.

\section{Problem Formulation}{\label{sec:problem}}

Let $\mathcal{G} \!\left(\mathbf{a}, \mathbf{b}\right)$ denote the {\em graph space}, defined as the set of all bipartite, directed, or undirected graphs with prescribed degree sequences $\mathbf{a}$ and $\mathbf{b}$, where the interpretation of $\mathbf{a}$ and $\mathbf{b}$ depends on the type of graph:

{\bf Bipartite graph}. In this case, $\mathbf{a} \in \mathbb{N}_{+}^{m}$ and $\mathbf{b} \in \mathbb{N}_{+}^{n}$ are the degree sequences of nodes in the two disjoint vertex sets, respectively. For instance, consider a financial network consisting of $m$ financial assets and $n$ investment firms, where each firm holds a subset of the available assets. Such a network can be modeled as a bipartite graph, in which one vertex nodes corresponds to assets and the other to firms, with each edge indicating that a given firm holds a particular asset. Here, the degree of an asset node equals the number of firms that hold the asset, while the degree of a firm node corresponds to the number of distinct assets owned by that firm.

{\bf Directed graph}. In this case, $\mathbf{a} \in \mathbb{N}_{0}^{m}$ and $\mathbf{b} \in \mathbb{N}_{0}^{n}$ correspond to the out-degree and in-degree sequences of nodes, respectively, and hence $m = n$. For example, consider a 
supply chain network consisting of $m$ firms, where each firm may act as a supplier or a customer to a subset of other firms. Such a network can be represented as a directed graph, in which each node corresponds to a firm and each directed edge points from a supplier to a customer. In this representation, the out-degree of a firm node equals the number of distinct firms that it supplies to, while the in-degree of a firm node corresponds to the number of distinct suppliers from which it purchases.

{\bf Undirected graph}. In this case, $\mathbf{a} = \mathbf{b} \in \mathbb{N}_{+}^{m}$ is the degree sequence of nodes. For instance, consider a social network consisting of $m$ individuals, where each individual interacts with a subset of other individuals. Such a network can be represented as an undirected graph, in which each node corresponds to an individual and each undirected edge indicates a mutual social connection between two individuals. Here, the degree of an individual node equals the number of distinct connections that the individual maintains within the network.

The following example illustrates how a pair of degree sequences maps to the corresponding bipartite, directed, and undirected graph spaces.

\begin{example}\label{example1}
    For the degree sequences $\mathbf{a} = \mathbf{b} = (2, 1, 1)$, Figure~\ref{fig:all_bigraphs} depicts the five bipartite graphs, one directed graph, and one undirected graph consistent with these degree sequences, i.e., the bipartite graph space has five elements, whereas the directed and undirected spaces each have a single element.

    \begin{figure}[h!]
    (1) Five bipartite graphs:
    \vskip 1em
    \centering
    \begin{minipage}{0.3\textwidth}
        \centering
        \begin{tikzpicture}[baseline=(current bounding box.center), scale=0.7]
            \node (a1) at (0,0) [circle,draw,minimum size=0.6cm] {2};
            \node (a2) at (2,0) [circle,draw,minimum size=0.6cm] {1};
            \node (a3) at (4,0) [circle,draw,minimum size=0.6cm] {1};
            
            \node (b1) at (0,-2) [circle,draw,minimum size=0.6cm] {2};
            \node (b2) at (2,-2) [circle,draw,minimum size=0.6cm] {1};
            \node (b3) at (4,-2) [circle,draw,minimum size=0.6cm] {1};
            
            \draw[thick] (a1) -- (b1);
            \draw[thick] (a1) -- (b2);
            \draw[thick] (a2) -- (b1);
            \draw[thick] (a3) -- (b3);
        \end{tikzpicture}
    \end{minipage}%
    \begin{minipage}{0.3\textwidth}
        \centering
        \begin{tikzpicture}[baseline=(current bounding box.center), scale=0.7]
            \node (a1) at (0,0) [circle,draw,minimum size=0.6cm] {2};
            \node (a2) at (2,0) [circle,draw,minimum size=0.6cm] {1};
            \node (a3) at (4,0) [circle,draw,minimum size=0.6cm] {1};
            
            \node (b1) at (0,-2) [circle,draw,minimum size=0.6cm] {2};
            \node (b2) at (2,-2) [circle,draw,minimum size=0.6cm] {1};
            \node (b3) at (4,-2) [circle,draw,minimum size=0.6cm] {1};
            
            \draw[thick] (a1) -- (b1);
            \draw[thick] (a1) -- (b3);
            \draw[thick] (a2) -- (b1);
            \draw[thick] (a3) -- (b2);
        \end{tikzpicture}
    \end{minipage}
    \begin{minipage}{0.3\textwidth}
        \centering
        \begin{tikzpicture}[baseline=(current bounding box.center), scale=0.7]
            \node (a1) at (0,0) [circle,draw,minimum size=0.6cm] {2};
            \node (a2) at (2,0) [circle,draw,minimum size=0.6cm] {1};
            \node (a3) at (4,0) [circle,draw,minimum size=0.6cm] {1};
            
            \node (b1) at (0,-2) [circle,draw,minimum size=0.6cm] {2};
            \node (b2) at (2,-2) [circle,draw,minimum size=0.6cm] {1};
            \node (b3) at (4,-2) [circle,draw,minimum size=0.6cm] {1};
            
            \draw[thick] (a1) -- (b1);
            \draw[thick] (a1) -- (b2);
            \draw[thick] (a2) -- (b3);
            \draw[thick] (a3) -- (b1);
        \end{tikzpicture}
    \end{minipage}
    
    \vskip 0.75em
    
    \begin{minipage}{0.3\textwidth}
        \centering
        \begin{tikzpicture}[baseline=(current bounding box.center), scale=0.7]
            \node (a1) at (0,0) [circle,draw,minimum size=0.6cm] {2};
            \node (a2) at (2,0) [circle,draw,minimum size=0.6cm] {1};
            \node (a3) at (4,0) [circle,draw,minimum size=0.6cm] {1};
            
            \node (b1) at (0,-2) [circle,draw,minimum size=0.6cm] {2};
            \node (b2) at (2,-2) [circle,draw,minimum size=0.6cm] {1};
            \node (b3) at (4,-2) [circle,draw,minimum size=0.6cm] {1};
            
            \draw[thick] (a1) -- (b1);
            \draw[thick] (a1) -- (b3);
            \draw[thick] (a2) -- (b2);
            \draw[thick] (a3) -- (b1);
        \end{tikzpicture}
    \end{minipage}%
    \begin{minipage}{0.3\textwidth}
        \centering
        \begin{tikzpicture}[baseline=(current bounding box.center), scale=0.7]
            \node (a1) at (0,0) [circle,draw,minimum size=0.6cm] {2};
            \node (a2) at (2,0) [circle,draw,minimum size=0.6cm] {1};
            \node (a3) at (4,0) [circle,draw,minimum size=0.6cm] {1};
            
            \node (b1) at (0,-2) [circle,draw,minimum size=0.6cm] {2};
            \node (b2) at (2,-2) [circle,draw,minimum size=0.6cm] {1};
            \node (b3) at (4,-2) [circle,draw,minimum size=0.6cm] {1};
            
            \draw[thick] (a1) -- (b2);
            \draw[thick] (a1) -- (b3);
            \draw[thick] (a2) -- (b1);
            \draw[thick] (a3) -- (b1);
        \end{tikzpicture}
    \end{minipage}

    \vskip 1.25em

    \raggedright
    (2) One directed graph:
    \vskip 1em
    \centering
    \begin{minipage}{0.35\textwidth}
        \centering
        \begin{tikzpicture}[baseline=(current bounding box.center), scale=0.7]
            \node (a2) at (2,0) [circle,draw,fill=green,minimum size=0.6cm] {};
            \node (b1) at (0,-2) [circle,draw,fill=blue!50,minimum size=0.6cm] {};
            \node (b3) at (4,-2) [circle,draw,fill=yellow,minimum size=0.6cm] {};
            
            \draw[->, thick] (a2) to[bend left=15] (b1);
            \draw[->, thick] (b1) to[bend left=15] (a2);
            \draw[->, thick] (b1) to[bend left=15] (b3);
            \draw[->, thick] (b3) to[bend left=15] (b1);
        \end{tikzpicture}
    \end{minipage}
    \begin{minipage}{0.1\textwidth}
        \centering
        \begin{tikzpicture}[scale=0.4]
            \draw[line width=0.5pt] 
                (0,0.3) -- (1.5,0.3) -- (1.5,1) -- (3,0) -- (1.5,-1) -- (1.5,-0.3) -- (0,-0.3) -- cycle;
        \end{tikzpicture}
    \end{minipage}
    \begin{minipage}{0.35\textwidth}
        \centering
        \begin{tikzpicture}[baseline=(current bounding box.center), scale=0.7]
            \node (a1) at (0,0) [circle,draw,fill=blue!50,minimum size=0.6cm] {2};
            \node (a2) at (2,0) [circle,draw,fill=green,minimum size=0.6cm] {1};
            \node (a3) at (4,0) [circle,draw,fill=yellow,minimum size=0.6cm] {1};
            
            \node (b1) at (0,-2) [circle,draw,fill=blue!50,minimum size=0.6cm] {2};
            \node (b2) at (2,-2) [circle,draw,fill=green,minimum size=0.6cm] {1};
            \node (b3) at (4,-2) [circle,draw,fill=yellow,minimum size=0.6cm] {1};
            
            \draw[thick] (b1) -- (a2);
            \draw[thick] (b1) -- (a3);
            \draw[thick] (b2) -- (a1);
            \draw[thick] (b3) -- (a1);
        \end{tikzpicture}
    \end{minipage}

    \vskip 1.5em

    \raggedright
    (3) One undirected graph:
    \vskip 1em
    \centering
    \begin{minipage}{0.35\textwidth}
        \centering
        \begin{tikzpicture}[baseline=(current bounding box.center), scale=0.7]
            \node (a2) at (2,0) [circle,draw,fill=green,minimum size=0.6cm] {1};
            \node (b1) at (0,-2) [circle,draw,fill=blue!50,minimum size=0.6cm] {2};
            \node (b3) at (4,-2) [circle,draw,fill=yellow,minimum size=0.6cm] {1};
            
            \draw[thick] (b1) -- (a2);
            \draw[thick] (b1) -- (b3);
        \end{tikzpicture}
    \end{minipage}
    \begin{minipage}{0.1\textwidth}
        \centering
        \begin{tikzpicture}[scale=0.4]
            \draw[line width=0.5pt] 
                (0,0.3) -- (1.5,0.3) -- (1.5,1) -- (3,0) -- (1.5,-1) -- (1.5,-0.3) -- (0,-0.3) -- cycle;
        \end{tikzpicture}
    \end{minipage}
    \begin{minipage}{0.35\textwidth}
        \centering
        \begin{tikzpicture}[baseline=(current bounding box.center), scale=0.7]
            \node (a1) at (0,0) [circle,draw,fill=blue!50,minimum size=0.6cm] {2};
            \node (a2) at (2,0) [circle,draw,fill=green,minimum size=0.6cm] {1};
            \node (a3) at (4,0) [circle,draw,fill=yellow,minimum size=0.6cm] {1};
            
            \node (b1) at (0,-2) [circle,draw,fill=blue!50,minimum size=0.6cm] {2};
            \node (b2) at (2,-2) [circle,draw,fill=green,minimum size=0.6cm] {1};
            \node (b3) at (4,-2) [circle,draw,fill=yellow,minimum size=0.6cm] {1};
            
            \draw[thick] (b1) -- (a2);
            \draw[thick] (b1) -- (a3);
            \draw[thick] (b2) -- (a1);
            \draw[thick] (b3) -- (a1);
        \end{tikzpicture}
    \end{minipage}

    \vskip 1.5em
    \caption{\baselineskip10pt Five bipartite graphs, one directed graph, and one undirected graph consistent with $\mathbf{a} = \mathbf{b} = (2, 1, 1)$.}
    \label{fig:all_bigraphs}
    \end{figure}
    
\end{example}

The graph space generation problem is to generate all elements or a substantial subset of $\mathcal{G}\!\left(\mathbf{a},\mathbf{b}\right)$, where $\mathcal{G}\!\left(\mathbf{a},\mathbf{b}\right)$ may represent a bipartite, directed, or undirected graph space. In this paper, we develop efficient type-specific algorithms: enumeration algorithms for exhaustively generating the entire graph space $\mathcal{G} \!\left(\mathbf{a}, \mathbf{b}\right)$, practical for small-scale problems, and sampling algorithms for randomly generating graph instances, which can be applied to larger-scale instances.

\section{Bipartite Graph Algorithms}\label{sec:BGA}


We begin by considering the problem of generating bipartite graphs that satisfy a prescribed pair of degree sequences with $\mathbf{a} \in \mathbb{N}_{+}^{m}$ and $\mathbf{b} \in \mathbb{N}_{+}^{n}$. Without loss of generality, we assume that the degree sequence $\mathbf{b}$ is sorted in non-increasing order $\left[\mathbf{b}\right]_1 \geq \left[\mathbf{b}\right]_2 \geq \dots \geq \left[\mathbf{b}\right]_n$, where $\left[\mathbf{x}\right]_i$ denotes the $i$-th entry of vector $\mathbf{x}$. With slight abuse of notations, $\left[\mathbf{b}\right]_i$ also refers to the node itself. The key steps of the generating mechanism are as follows. We first partition the nodes in $\mathbf{a}$ into $p$ groups according to their distinct degree values: group $k$ contains $m_k$ nodes of degree $\alpha_k$, where $\alpha_1 < \alpha_2 < \dots < \alpha_p$ and $\sum_{k=1}^{p} m_k = m$. The first node in $\mathbf{b}$, denoted $\left[\mathbf{b}\right]_1$, then selects the number of nodes in $\mathbf{a}$ to connect to, proceeding sequentially from group $1$ to group $p$. After establishing these connections, any nodes whose degrees are reduced to zero are removed. The updated degree sequences are denoted as $\mathbf{a}^{[2]}\in \mathbb{N}_{+}^{m^{[2]}}$ and $\mathbf{b}^{[2]}\in \mathbb{N}_{+}^{n-1}$, where $m^{[2]}$ is the number of remaining nodes in $\mathbf{a}$ (depends on the number of nodes removed), and the size of $\mathbf{b}$ decreases to $n-1$. More generally, we denote by $\mathbf{a}^{[j+1]}$ and $\mathbf{b}^{[j+1]}$ the updated degree sequences obtained after processing node $\left[\mathbf{b}\right]_j$, where $\mathbf{a}^{[1]} = \mathbf{a}$ and $\mathbf{b}^{[1]} = \mathbf{b}$. These grouping and connection steps are then applied iteratively until all nodes in $\mathbf{b}$ have been processed. Because each iteration can be regarded as solving a new instance of the same problem with updated prescribed degree sequences, it is sufficient to focus on the connection process for the first node $\left[\mathbf{b}\right]_1$.

Returning to Example~\ref{example1}, we use Figure~\ref{fig:process_of_b1} to illustrate the generating mechanism. First, the nodes in $\mathbf{a}$ are partitioned into two groups ($p=2$): group 1 consists of two nodes of degree 1 ($\alpha_1 = 1$, $m_1 = 2$), and group 2 consists of one node of degree 2 ($\alpha_2 = 2$, $m_2 = 1$). Next, the node $\left[\mathbf{b}\right]_1$ of degree 2 sequentially connects to nodes in $\mathbf{a}$, starting from group 1 and proceeding to group 2. For example, the node $\left[\mathbf{b}\right]_1$ connects to the two nodes in group 1, reducing their degree from 1 to 0. After these connections are made, $\left[\mathbf{b}\right]_1$ and the two nodes of degree 0 in group 1 are removed, resulting in the updated degree sequences $\mathbf{a}^{[2]} = (2)$ and $\mathbf{b}^{[2]} = (1, 1)$. This connection process is then repeated for the next node in $\mathbf{b}$ using the updated degree sequences, and continued until all nodes have been processed.

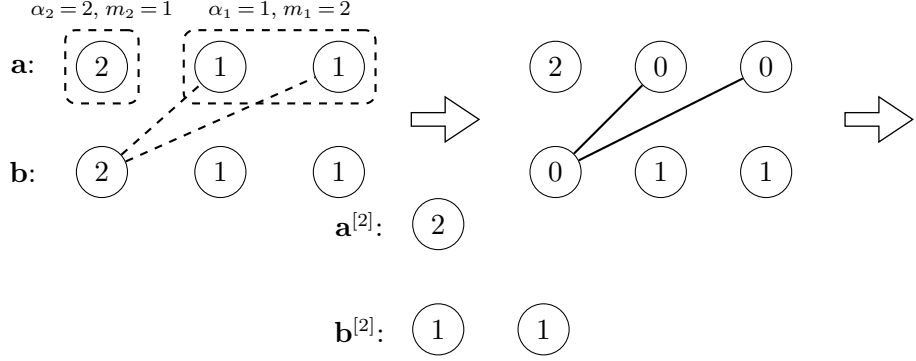
\begin{figure}[htpb!]
    \vspace{1em}
    \centering
    \begin{minipage}{0.28\textwidth}
        \centering
        \begin{tikzpicture}[scale=0.7, baseline=(a1.base)]
            \node (a_label) at (-1.5,0) {$\mathbf{a}$:};
            \node (a1) at (0,0) [circle,draw,minimum size=0.6cm] {2};
            \node (a2) at (2.25,0) [circle,draw,minimum size=0.6cm] {1};
            \node (a3) at (4.5,0) [circle,draw,minimum size=0.6cm] {1};

            \node (b_label) at (-1.5,-2) {$\mathbf{b}$:};
            \node (b1) at (0,-2) [circle,draw,minimum size=0.6cm] {2};
            \node (b2) at (2.25,-2) [circle,draw,minimum size=0.6cm] {1};
            \node (b3) at (4.5,-2) [circle,draw,minimum size=0.6cm] {1};
            
            \draw[dashed, thick] (b1) -- (a2);
            \draw[dashed, thick] (b1) -- (a3);

            \node[
                draw,
                dashed,
                thick,
                rounded corners,
                inner sep=4pt,
                fit={(a1)},
                overlay
            ] (box2) {};

            \node[
                overlay,
                align=center
            ] at ($(box2.north)+(0,0.35)$) {\scriptsize $\alpha_2 = 2,\,m_2 = 1$};

            \node[
                draw,
                dashed,
                thick,
                rounded corners,
                inner sep=4pt,
                fit={(a2) (a3)},
                overlay
            ] (box1) {};

            \node[
                overlay,
                align=center
            ] at ($(box1.north)+(0,0.35)$) {\scriptsize $\alpha_1 = 1,\,m_1 = 2$};
            
        \end{tikzpicture}
    \end{minipage}
    \hspace{1.5em}
    \begin{minipage}{0.04\textwidth}
        \centering
        \begin{tikzpicture}[scale=0.3]
            \draw[line width=0.5pt] 
                (0,0.3) -- (1.5,0.3) -- (1.5,1) -- (3,0) -- (1.5,-1) -- (1.5,-0.3) -- (0,-0.3) -- cycle;
        \end{tikzpicture}
    \end{minipage}
    \hspace{0.25em}
    \begin{minipage}{0.28\textwidth}
        \centering
        \begin{tikzpicture}[scale=0.7, baseline=(a1.base)]
            \node (a1) at (0,0) [circle,draw,minimum size=0.6cm] {2};
            \node (a2) at (2,0) [circle,draw,minimum size=0.6cm] {0};
            \node (a3) at (4,0) [circle,draw,minimum size=0.6cm] {0};
            
            \node (b1) at (0,-2) [circle,draw,minimum size=0.6cm] {0};
            \node (b2) at (2,-2) [circle,draw,minimum size=0.6cm] {1};
            \node (b3) at (4,-2) [circle,draw,minimum size=0.6cm] {1};

            \draw[thick] (b1) -- (a2);
            \draw[thick] (b1) -- (a3);
        \end{tikzpicture}
    \end{minipage}
    \begin{minipage}{0.04\textwidth}
        \centering
        \begin{tikzpicture}[scale=0.3]
            \draw[line width=0.5pt] 
                (0,0.3) -- (1.5,0.3) -- (1.5,1) -- (3,0) -- (1.5,-1) -- (1.5,-0.3) -- (0,-0.3) -- cycle;
        \end{tikzpicture}
    \end{minipage}
    \begin{minipage}{0.28\textwidth}
        \centering
        \begin{tikzpicture}[scale=0.7, baseline=(a1.base)]
            \node (a_label) at (-1.5,0) {$\mathbf{a}^{[2]}$:};
            \node (a1) at (0,0) [circle,draw,minimum size=0.6cm] {2};

            \node (b_label) at (-1.5,-2) {$\mathbf{b}^{[2]}$:};
            \node (b1) at (0,-2) [circle,draw,minimum size=0.6cm] {1};
            \node (b2) at (2,-2) [circle,draw,minimum size=0.6cm] {1};
        \end{tikzpicture}
    \end{minipage}

    \vskip 1em
    \caption{\baselineskip10pt Connection process for node $\left[\mathbf{b}\right]_1$ in Example~\ref{example1} with degree sequences $\mathbf{a} = \mathbf{b} = (2,\,1,\,1)$.}
    \label{fig:process_of_b1}
\end{figure}

The key challenge in this generating mechanism is to determine, at each step, how many nodes can be connected so as to guarantee the generation of a bipartite graph consistent with the prescribed degree sequences.

\subsection{Node Connection Interval}\label{ssec:interval}

In this subsection, we present a theorem that establishes a necessary and sufficient interval condition for the admissible number of connections at each  step, under which global feasibility is guaranteed.
We first introduce several essential mathematical definitions required for its formulation.

\begin{definition}
    $\mathscr{Z}: (\mathbb{N}^{p}, \mathbb{N}_{+})\to \left\{0,1\right\}^{p\times c}$ is a matrix operation that maps a vector $\mathbf{x}\in \mathbb{N}^{p}$ and a positive integer $c\in \mathbb{N}_{+}$ to a binary matrix of size $p\times c$. The operation is defined entry-wise as follows:
\begin{equation}\label{eq:Z matrix}
    \left[\mathscr{Z} \left(\mathbf{x}, c\right)\right]_{ij} = 
        \begin{cases}
        1, & \text{if } 1 \leq j \leq \left[\mathbf{x}\right]_i, \\
        0, & \text{otherwise}.
        \end{cases}
\end{equation}
Notice that the resulting binary matrix $\mathscr{Z} \!\left(\mathbf{x}, c\right)$ is such that for each $i\in \left\{1,2,\dots,p\right\}$, the first $\mathbf{x}_i$ columns in row $i$ are set to 1, while the remaining entries are set to 0.
\end{definition}

Let $\mathbf{z}\in \mathbb{N}^{|\mathbf{b}|}$ denote a vector derived from the matrix $\mathscr{Z} \!\left(\mathbf{a}, |\mathbf{b}|\right)$, defined as follows:
\begin{equation}\label{eq:z vector}
    \left[\mathbf{z}\right]_j = \sum_{i=1}^{m} \left[\mathscr{Z} \!\left(\mathbf{\mathbf{a}}, |\mathbf{b}|\right)\right]_{ij},
\end{equation}
representing the sum of the $j$-th column of $\mathscr{Z} \left(\mathbf{a}, |\mathbf{b}|\right)$.

To illustrate the above definitions, we revisit Example~\ref{example1}, where $\mathbf{a} = \mathbf{b} =(2,1,1)$. By \eqref{eq:Z matrix} and \eqref{eq:z vector}, we obtain:
\setlength{\arraycolsep}{6pt}
\begin{equation*}
    \mathscr{Z} \!\left(\mathbf{\mathbf{a}}, |\mathbf{b}|\right) =
    {\renewcommand{\arraystretch}{0.6}
        \begin{bmatrix}
        1 & 1 & 0 \\
        1 & 0 & 0 \\
        1 & 0 & 0 \\
        \end{bmatrix}}
    ,
\end{equation*}
and $\mathbf{z}=(3,1,0)$.

Using the definitions introduced above, we present the following lemma, which provides an equivalent formulation of the Gale-Ryser theorem \citep{gale1957theorem} characterizing the necessary and sufficient conditions for the existence of a bipartite graph with the prescribed degree sequences. The proof is included in Section \ref{appx:L1-proof} of the e-companion.

\begin{lemma}\label{lemma:necessary and sufficient condition}
    For the prescribed degree sequences $\mathbf{a}$ and $\mathbf{b}$, a bipartite graph exists if and only if $\sum_{i=1}^{m} \left [ \mathbf{a} \right ]_i = \sum_{j=1}^{n} \left [ \mathbf{b} \right ]_j$, and the following condition holds for all $r\in \left\{1, 2, \dots, n\right\}$:
    \begin{equation}\label{eq:ns-conditions}
        \sum_{j=1}^{r} \left [ \mathbf{z} \right ]_j \geq \sum_{j=1}^{r} \left [ \mathbf{b} \right ]_j.
    \end{equation}
\end{lemma}

Next, we introduce the notation for the number of nodes in each group connected to $\left[\mathbf{b}\right]_1$. Let $\mathscr{F} \!\left(k\right)\in \left[0,m_k\right] \cap \mathbb{N}$ (for $k=1,2,\dots,p$) denote the number of nodes from the $m_k$ nodes in group $k$ of $\mathbf{a}$ that are connected to $\left[\mathbf{b}\right]_1$. After connecting the nodes in group $k-1$ to $\left[\mathbf{b}\right]_1$, the following definitions are introduced:
\begin{enumerate}[(i)]
    \item The updated \(\mathbf{a}\), denoted by \(\mathbf{a}^{(k)}\), where $\mathbf{a}^{(1)} = \mathbf{a}$;
    \item The updated $\mathbf{z}$ derived from \(\mathscr{Z} \!\left(\mathbf{a}^{(k)}, |\mathbf{b}|\right)\), denoted by \(\mathbf{z}^{(k)}\), where $\mathbf{z}^{(1)} = \mathbf{z}$;
    \item The remaining degree of \(\left[\mathbf{b}\right]_1\), denoted by \(\left[\mathbf{b}\right]_{1}^{(k)}\), where $\left[\mathbf{b}\right]_{1}^{(k)} = \left[\mathbf{b}\right]_{1} - \sum_{k'=1}^{k-1} \mathscr{F} \!\left(k'\right)$ and $\left[\mathbf{b}\right]_{1}^{(1)} = \left[\mathbf{b}\right]_1$;
    \item The total number of nodes in groups $k+1$ to $p$ of $\mathbf{a}^{(k)}$, denoted by $M^{(k)}$, where $M^{(k)} = \sum_{k'=k+1}^{p} m_{k'}$ and $M^{(p)} = 0$.
\end{enumerate}

To illustrate the above definitions, we revisit Example~\ref{example1} with $\mathbf{a} = \mathbf{b} =(2,1,1)$. Consider the case $k = 2$ and suppose that $\mathscr{F}\!\left(1\right) = 1$, i.e., the node $\left[\mathbf{b}\right]_1$ is connected to one node in group $1$ of $\mathbf{a}$. In this case, the updated $\mathbf{a}$ is $\mathbf{a}^{(2)} = (2,0,1)$; The corresponding matrix is
    \begin{equation*}
    \mathscr{Z} \!\left(\mathbf{a}^{(2)}, |\mathbf{b}|\right) =
    {\renewcommand{\arraystretch}{0.6}
        \begin{bmatrix}
        1 & 1 & 0 \\
        0 & 0 & 0 \\
        1 & 0 & 0 \\
        \end{bmatrix}}
    ,
    \end{equation*}
    which derives $\mathbf{z}^{(2)} = (2,1,0)$; The remaining degree of $\left[\mathbf{b}\right]_1$ is $\left[\mathbf{b}\right]_{1}^{(2)} = \left[\mathbf{b}\right]_1 - \mathscr{F}\!\left(1\right) = 2 - 1 = 1$; The total number of nodes in groups $k+1$ to $p$ is $M^{(2)} = 0$, since $M^{(k)} = 0$ when $k = p = 2$.

\begin{remark}\label{F(K)-note}
    $\mathscr{F} \!\left(k\right)$ represents the number of entries in the $\alpha_k$-th column of $\mathscr{Z}\!\left(\mathbf{a},|\mathbf{b}|\right)$ that change from 1 to 0, i.e., the reduction in the $\alpha_k$-th component of $\mathbf{z}$. Consider again the case $k=2$ with $\mathscr{F}\!\left(1\right) = 1$ in Example~\ref{example1}, where $\mathbf{a}=\mathbf{b}=(2,1,1)$. Here, $\mathscr{F}\!\left(1\right) = 1$ means that one entry in the $\alpha_1 = 1$-st column of the corresponding matrix changes form 1 to 0, namely,
\[
\mathscr{Z}\!\left(\mathbf{a},|\mathbf{b}|\right)=
{\renewcommand{\arraystretch}{0.6}
    \begin{bmatrix}
    1 & 1 & 0 \\
    1 & 0 & 0 \\
    1 & 0 & 0 \\
    \end{bmatrix}}
\quad \longrightarrow \quad
\mathscr{Z}\!\left(\mathbf{a}^{(2)},|\mathbf{b}|\right)=
{\renewcommand{\arraystretch}{0.6}
    \begin{bmatrix}
    1 & 1 & 0 \\
    0 & 0 & 0 \\
    1 & 0 & 0 \\
    \end{bmatrix}} .
\]
Correspondingly, the $\alpha_1 = 1$-st component of $\mathbf{a}$ decreases from $\mathbf{z} = (3,1,0)$ to $\mathbf{z}^{(2)} = (2,1,0)$.
\end{remark}



We next introduce a generalized vector subtraction operation, which will be used to define several quantities appearing in the theorem.
\begin{definition}\label{def:D-vector}
    $\mathscr{D}:(\mathbb{N}^{p},\mathbb{N}^{q})\to \mathbb{Z}^{\max \left\{p,q\right\}}$ is a vector operation that maps two vectors $\mathbf{x}\in \mathbb{N}^{p}$ and $\mathbf{y}\in \mathbb{N}^{q}$ to another vector of length $\max \left\{p,q\right\}$. The operation is defined entry-wise as follows:
\begin{equation}\label{eq:D def}
    \left[\mathscr{D} \!\left(\mathbf{x}, \mathbf{y}\right)\right]_h =
        \begin{cases}
            \left[\mathbf{x}\right]_h - \left[\mathbf{y}\right]_h, & \text{if } h \leq \min \left\{p,q\right\}, \\
            \left[\mathbf{x}\right]_h, & \text{if } h > \min \left\{p,q\right\} \text{and } p > q, \\
            - \left[\mathbf{y}\right]_h, & \text{otherwise.}
        \end{cases}
\end{equation}
By this definition, the resulting vector $\mathscr{D} \!\left(\mathbf{x}, \mathbf{y}\right)$ is such that for each $h\in \left\{1,2,\dots,\max \left\{p,q\right\}\right\}$, the value is given by $\left[\mathbf{x}\right]_h - \left[\mathbf{y}\right]_h$, with zeros appended to either $\mathbf{x}$ or $\mathbf{y}$ if their lengths are less than $\max \left\{p,q\right\}$.
\end{definition}

By combining $\mathbf{z}^{(k)}$, $\mathbf{b}^{[2]}$, and \eqref{eq:D def}, we define the following quantities appearing in the theorem:
\begin{equation}\label{eq:definition of left and now}
    D_{\mathrm{left}}^{(k)} = \sum_{h=1}^{{\alpha_k}-1} \left[\mathscr{D} \!\left(\mathbf{z}^{(k)}, \mathbf{b}^{[2]}\right)\right]_h, \quad D_{\mathrm{now}}^{(k)}= \left[\mathscr{D} \!\left(\mathbf{z}^{(k)}, \mathbf{b}^{[2]}\right)\right]_{\alpha_k},
\end{equation}
and
\begin{equation}\label{eq:definition of right}
    D_{\mathrm{right}}^{(k)}= \min\left\{\min_{r\in \{\alpha_k+1,\, \alpha_k+2,\, \dots,\, n\}} \sum_{h=\alpha_k + 1}^{r} \left[\mathscr{D} \!\left(\mathbf{z}^{(k)}, \mathbf{b}^{[2]}\right)\right]_h, 0\right\},
\end{equation}
where $\mathscr{D} \!\left(\mathbf{z}^{(k)}, \mathbf{b}^{[2]}\right)\in \mathbb{Z}^n$, and recall that $\alpha_k$ is the degree of the group $k$ in $\mathbf{a}^{(k)}$. Notice that the vector $\mathbf{z}^{(k)}$, derived from the matrix $\mathscr{Z} \!\left(\mathbf{a}^{(k)}, |\mathbf{b}|\right)$, and the vector $\mathbf{b}^{[2]}$ are not in the same state. Nevertheless, we apply the $\mathscr{D}$-operation to both vectors for the theorem.

In this paper, we adopt the following notational conventions for summation operator. The summation operator $\sum$ is defined to be zero whenever its upper limit is less than its lower limit. Additionally, any term whose index falls outside the summation range is treated as zero.

Building on the preceding definitions and Lemma~\ref{lemma:necessary and sufficient condition}, we establish the following theorem, with the proof included in Section \ref{appx:T1-proof} of the e-companion. 

\begin{theorem}\label{thm:interval}
    Suppose that a bipartite graph exists for the prescribed degree sequences $\mathbf{a}$ and $\mathbf{b}$. Then a bipartite graph exists for the updated degree sequences $\mathbf{a}^{[2]}$ and $\mathbf{b}^{[2]}$ if and only if for each group $k$ of $\mathbf{a}$, the number of connected nodes $\mathscr{F} \!\left(k\right)$ satisfies
    \begin{equation}\label{eq:Fk}
        \mathscr{F} \!\left(k\right) \in 
        \left[\max \bigl\{ \bigl[\mathbf{b}\bigr]_1^{(k)} - M^{(k)},\, 0 \bigr\},\, \min \bigl\{m_k,\, D_{\mathrm{left}}^{(k)} + D_{\mathrm{now}}^{(k)} + D_{\mathrm{right}}^{(k)} \bigr\}\right].
    \end{equation}
\end{theorem}

Theorem~\ref{thm:interval} characterizes the necessary and sufficient interval condition for the admissible number of nodes from the $m_k$ nodes in group $k$ of $\mathbf{a}$ that are connected to $\left[\mathbf{b}\right]_1$. Specifically, for the prescribed degree sequences, the number of connected nodes $\mathscr{F} \!\left(k\right)$ must satisfy the interval condition given in \eqref{eq:Fk} in order for a bipartite graph to be generated. The first term $\left[\mathbf{b}\right]_{1}^{(k)} - M^{(k)}$ in the lower bound captures the situation where the remaining degree of $\left[\mathbf{b}\right]_1$ exceeds the total number of nodes in groups $k+1$ to $p$ of $\mathbf{a}^{(k)}$, in which case $\left[\mathbf{b}\right]_1$ must connect to at least $\left[\mathbf{b}\right]_{1}^{(k)} - M^{(k)}$ nodes from group~$k$ in order to exhaust its remaining degree. The inclusion of zero ensures the nonnegativity of $\mathscr{F} \!\left(k\right)$. The first term $m_k$ of the upper bound reflects the trivial constraint that no more than $m_k$ nodes from group~$k$ can be connected. The inclusion of $D_{\mathrm{left}}^{(k)} + D_{\mathrm{now}}^{(k)} + D_{\mathrm{right}}^{(k)}$ guarantees that the updated degree sequences $\mathbf{a}^{[2]}$ and $\mathbf{b}^{[2]}$ satisfy \eqref{eq:ns-conditions} in Lemma~\ref{lemma:necessary and sufficient condition}, thereby remain feasible for bipartite graph generation after the connections of $\left[\mathbf{b}\right]_1$ are determined.

\subsection{Bipartite Graph Enumeration Algorithm}\label{ssec:enumeration}

Building on Theorem \ref{thm:interval}, we develop a two-layer breadth-first search algorithm to enumerate all bipartite graphs that satisfy the prescribed degree sequences. The basic idea is to iterate over all connection configurations corresponding to each feasible value $\mathscr{F} \!\left(k\right)$ within the interval given by \eqref{eq:Fk}, thereby systematically exploring the entire state space. Here, a connection configuration refers to one distinct way of selecting $\mathscr{F} \!\left(k\right)$ from the $m_k$ nodes in group $k$ to connect to $\left[\mathbf{b}\right]_j$. The algorithm traverses the state space via two nested breadth-first search procedures. That is, the outer search iterates over the nodes of $\mathbf{b}$, while the inner search explores the degree groups associated with each node. This construction guarantees that all feasible bipartite graphs are enumerated without duplication. 

Before presenting the pseudocode, we briefly review the queue data structure and the associated operations \textit{enqueue} and \textit{dequeue}, which are fundamental to the breadth-first search procedures. A queue follows the first-in--first-out discipline: the \textit{enqueue} operation inserts an element at the end of the queue, whereas the \textit{dequeue} operation removes and returns the element at the front of the queue. These operations allow the algorithm to systematically explore all intermediate graph states in the exact order in which they are generated. Based on the queue data structure, the bipartite graph enumeration (BGE) algorithm is presented in Algorithm~\ref{alg:BGE}.
\begin{algorithm}[htpb]
\caption{Bipartite Graph Enumeration (BGE) Algorithm}
\label{alg:BGE}
{\renewcommand{\baselinestretch}{0.75}\selectfont
\begin{algorithmic}[1]
\Require Degree sequences $\mathbf{a} \in \mathbb{N}_{+}^{m}$ and $\mathbf{b} \in \mathbb{N}_{+}^{n}$
\Ensure The set $\mathcal{G} \!\left(\mathbf{a}, \mathbf{b}\right)$ of all bipartite graphs consistent with the degree sequences $\mathbf{a}$ and $\mathbf{b}$

\State Initialize $\mathcal{G}(\mathbf{a},\mathbf{b}) \gets \emptyset$, an empty bipartite graph $\mathbf{G}(\mathbf{a}, \mathbf{b})$, two empty queues $\mathcal{Q}_{\mathrm{outer}}$ and $\mathcal{Q}_{\mathrm{inner}}$
\State Enqueue the pair $\left(\mathbf{G}(\mathbf{a}, \mathbf{b}), 1\right)$ into $\mathcal{Q}_{\mathrm{outer}}$ \Comment{Node index $j=1$}

\While{$\mathcal{Q}_{\mathrm{outer}} \neq \emptyset$} \Comment{Outer breadth-first search over $n$ nodes}
  \State Dequeue the first pair $(\mathbf{G}(\mathbf{a}, \mathbf{b}), j)$ from $\mathcal{Q}_{\mathrm{outer}}$ and record the node index $j$

  \If{$j > n$}
    \State Add $\mathbf{G}\!\left(\mathbf{a}, \mathbf{b}\right)$ to $\mathcal{G} \!\left(\mathbf{a}, \mathbf{b}\right)$
    \State \textbf{continue}
  \EndIf
  \State Group $\mathbf{a}^{[j]}$ into $p$ groups $(\alpha_k, m_k)$ with $\alpha_1 < \cdots < \alpha_p$
  \State Enqueue the pair $\left(\mathbf{G}(\mathbf{a}, \mathbf{b}), 1\right)$ into $\mathcal{Q}_{\mathrm{inner}}$ \Comment{Group index $k=1$}

  \While{$\mathcal{Q}_{\mathrm{inner}} \neq \emptyset$} \Comment{Inner breadth-first search over $p$ groups for each node}
    \State Dequeue the first pair $(\mathbf{G}\left(\mathbf{a}, \mathbf{b}), k\right)$ from $\mathcal{Q}_{\mathrm{inner}}$ and record the group index $k$
    \If{$k > p$}
      \State Enqueue the pair $\left(\mathbf{G}\!\left(\mathbf{a}, \mathbf{b}\right), j + 1\right)$ into $\mathcal{Q}_{\mathrm{outer}}$
        \State \textbf{continue}
    \EndIf
    
    \State Compute the interval condition $[\ell_k, u_k]$ via \eqref{eq:Fk}
    \For{$\mathscr{F}\!\left(k\right) = \ell_k,\,\dots,u_k$}
      \State Generate all $\binom{m_k}{\mathscr{F}\!\left(k\right)}$ connection configurations
      \For{each connection configuration}
        \State Copy $\mathbf{G}\!\left(\mathbf{a}, \mathbf{b}\right)$ and denote it by $\mathbf{G}'\!\left(\mathbf{a}, \mathbf{b}\right)$
        \State Append the connection configuration to $\mathbf{G}'\!\left(\mathbf{a}, \mathbf{b}\right)$
        \State Enqueue the pair $\left(\mathbf{G}'\!\left(\mathbf{a}, \mathbf{b}\right), k+1\right)$ into $\mathcal{Q}_{\mathrm{inner}}$
      \EndFor
    \EndFor
  \EndWhile
\EndWhile

\State \textbf{return} $\mathcal{G} \!\left(\mathbf{a}, \mathbf{b}\right)$
\end{algorithmic}
}
\end{algorithm}

To illustrate the mechanics of the BGE algorithm, we revisit Example~\ref{example1} with $\mathbf{a} = \mathbf{b} = (2,1,1)$. Figure~\ref{fig:enumeration_process_of_b1} depicts the enumeration process for the first node $\left[\mathbf{b}\right]_1$, highlighting the interaction between the outer and inner queues. The algorithm begins by enqueuing the empty bipartite graph $\mathbf{G}(\mathbf{a}, \mathbf{b})$ with node index $j=1$ into the empty outer queue~$\mathcal{Q}_{\mathrm{outer}}$, as shown in panel~1 to panel~2 of Figure~\ref{fig:enumeration_process_of_b1}. When the state $\bigl(\mathbf{G}(\mathbf{a}, \mathbf{b}), j=1\bigr)$ is dequeued from the outer queue, the algorithm starts processing the first node $\left[\mathbf{b}\right]_1$. At this stage, the degree sequence $\mathbf{a}^{[1]}$ is partitioned into two groups ($p=2$): group~1 consists of two nodes of degree 1 ($\alpha_1 = 1$, $m_1 = 2$), and group~2 consists of one node of degree 2 ($\alpha_2 = 2$, $m_2 = 1$). The algorithm then enqueues the empty bipartite graph $\mathbf{G}(\mathbf{a}, \mathbf{b})$ with group index $k=1$ into the empty inner queue~$\mathcal{Q}_{\mathrm{inner}}$, indicating that group~1 should be processed next; see panel~2 to panel~3 in Figure~\ref{fig:enumeration_process_of_b1}. When the state $\bigl(\mathbf{G}(\mathbf{a}, \mathbf{b}), k=1\bigr)$ is dequeued from the inner queue, the algorithm starts processing group~1.
According to Theorem~\ref{thm:interval}, the interval condition for group~1 is $\mathscr{F}(1)\in\left[1,2\right]$. For each feasible integer $\mathscr{F}(1)$, the algorithm generates all $\binom{2}{\mathscr{F}(1)}$ connection configurations by selecting $\mathscr{F}(1)$ nodes from the two nodes in group~1. For example, when $\mathscr{F}(1) = 1$, two distinct connection configurations are generated. For each such connection configuration, a new partial graph $\mathbf{G}'(\mathbf{a}, \mathbf{b})$ is obtained by appending the connection configuration to a copy of the current partial graph $\mathbf{G}(\mathbf{a}, \mathbf{b})$. The resulting state $\bigl(\mathbf{G}'(\mathbf{a}, \mathbf{b}), k=2\bigr)$ is then enqueued into the inner queue~$\mathcal{Q}_{\mathrm{inner}}$, indicating that group~2 should be processed next for node $\left[\mathbf{b}\right]_1$; see panel~3 to panel~4 in Figure~\ref{fig:enumeration_process_of_b1}. The inner breadth-first search proceeds by repeatedly dequeuing states from the inner queue~$\mathcal{Q}_{\mathrm{inner}}$ and processing group~2, as shown in panel~4 to panel~5 of Figure~\ref{fig:enumeration_process_of_b1}. Since $p=2$ in the present example, once group~2 has been processed, the connection step for node $\left[\mathbf{b}\right]_1$ is complete. Each resulting partial graph is therefore associated with the next node index and enqueued into the outer queue~$\mathcal{Q}_{\mathrm{outer}}$ as a state $\bigl(\mathbf{G}'(\mathbf{a},\mathbf{b}), j=2\bigr)$; see panel~5 to panel~6 in Figure~\ref{fig:enumeration_process_of_b1}. Subsequently, each state $\bigl(\mathbf{G}'(\mathbf{a}, \mathbf{b}), j=2\bigr)$ is dequeued from the outer queue~$\mathcal{Q}_{\mathrm{outer}}$, and enqueued into the inner queue~$\mathcal{Q}_{\mathrm{inner}}$, indicating that the algorithm proceeds to process the second node $\left[\mathbf{b}\right]_2$; see panel~6 to panel~7 in Figure~\ref{fig:enumeration_process_of_b1}.

The same enumeration process is then repeated for $j=2$ and $j=3$. For each node $\left[\mathbf{b}\right]_j$, the updated degree sequence $\mathbf{a}^{[j]}$ is first partitioned into degree groups, the interval conditions for each group are recomputed using~\eqref{eq:Fk} in Theorem~\ref{thm:interval}, and all feasible group-level connection configurations are explored through the inner breadth-first search. Once all nodes in $\mathbf{b}$ have been processed, the resulting graph $\mathbf{G}(\mathbf{a}, \mathbf{b})$ is added to the set $\mathcal{G}(\mathbf{a},\mathbf{b})$. In this manner, the BGE algorithm systematically enumerates all bipartite graphs consistent with $\mathbf{a}$ and $\mathbf{b}$ without duplication.

\begin{figure}[htpb!]
    \centering
    \includegraphics[width=\linewidth]{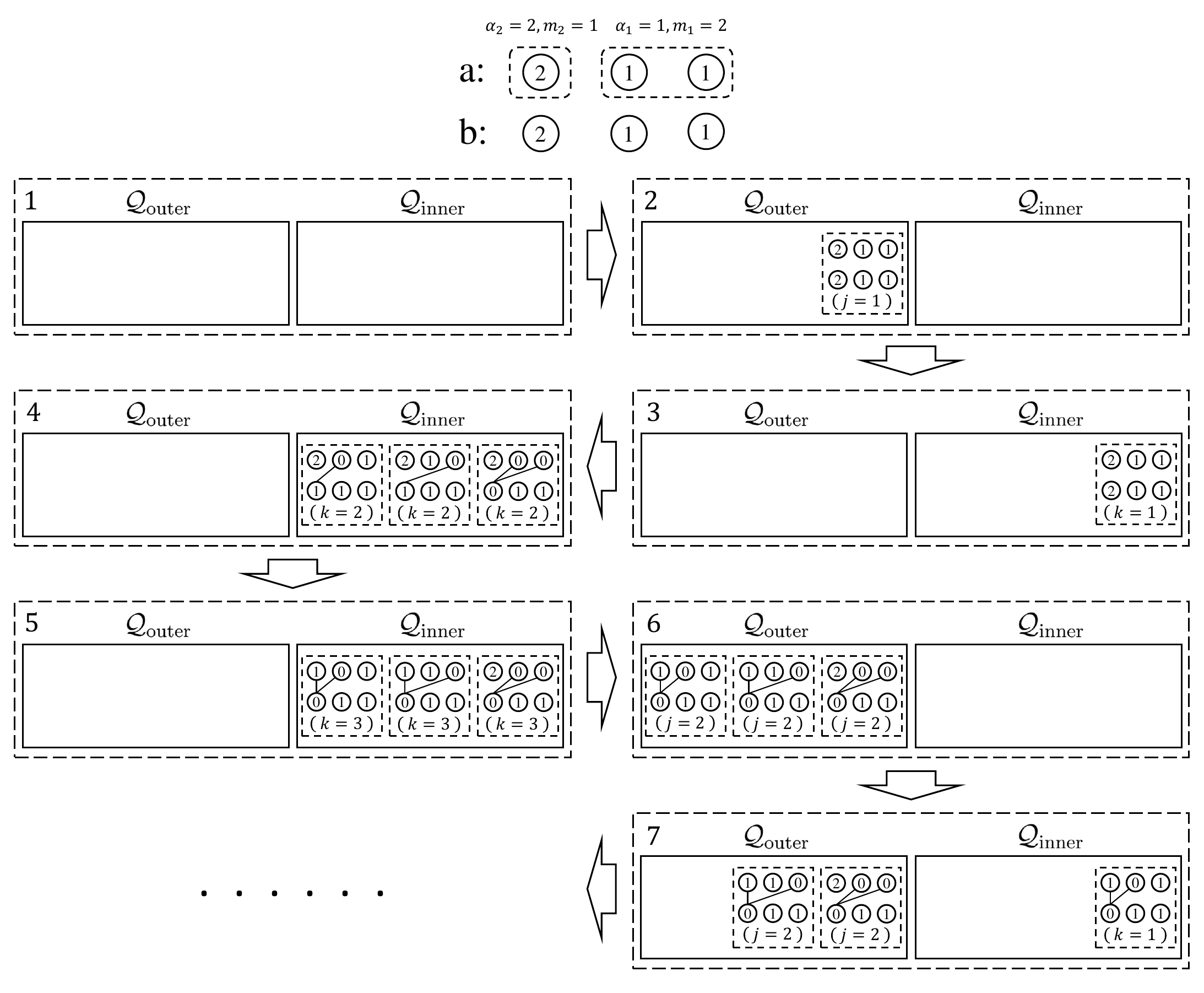}
    \caption{\baselineskip10pt Enumeration process for node $\left[\mathbf{b}\right]_1$ in Example~\ref{example1} with degree sequences $\mathbf{a} = \mathbf{b} = (2,\,1,\,1)$.}
    \label{fig:enumeration_process_of_b1}
\end{figure}



\subsection{Bipartite Graph Sampling Algorithms}\label{ssec:sampling}
The enumeration algorithm is practical for small-scale settings, but for larger-scale degree sequences, the number of bipartite graphs can easily reach hundreds of millions, making exhaustive enumeration impractical. For example, there are more than $10^{20}$ bipartite graphs that satisfy the degree sequences $\mathbf{a} = \mathbf{b} = (6, 6, 6, 5, 5, 5, 5, 4, 4, 3, 2)$. Therefore, sampling algorithms provide a more tractable alternative.

As mentioned in Section \ref{sec:Intro}, in the absence of additional prior information, the goal of a sampling algorithm is to ensure that the sampled graphs are uniformly distributed over the set of feasible bipartite graphs. This is a natural criterion, as uniformity allows the algorithm to explore the entire solution space without bias, ensuring that each graph is equally likely to be generated. The uniformity criteria is also adopted in \cite{Glasserman2023}, where they assumed uniform sampling to guarantee that the sampling process does not inadvertently favor any specific graph structure. Therefore, in this paper, we also use sampling uniformity as one of the criteria to evaluate the performance of the sampling algorithm.

Compared with the BGE algorithm, which must identify all feasible bipartite graphs without duplication, sampling algorithms are logically simpler. Given the interval condition computed from \eqref{eq:Fk}, a value of $\mathscr{F}(k)$ is randomly selected within this interval, and a feasible connection configuration is then chosen at random. 
For example, consider Example~\ref{example1}, where $\mathbf{a} = \mathbf{b} =(2,1,1)$. When connecting $[\mathbf{b}]_1$ to group 1 (i.e., $k=1$), which contains two degree-one nodes, the interval condition is computed as $\mathscr{F}(1)\in[1,2]$ by \eqref{eq:Fk}. Thus, $[\mathbf{b}]_1$ may connect to either one or two nodes in this group. In the sampling algorithm, we randomly select $\mathscr{F}(1)=1$ or $\mathscr{F}(1)=2$. Furthermore, when $\mathscr{F}(1)=1$, an additional choice arises as to which specific node to connect. This is handled by randomly selecting one of the two nodes in group 1.

However, a natural question arises as to whether both the number of connected nodes and the specific nodes should be selected uniformly. 
Returning to Example~\ref{example1}, Figure \ref{fig:weight_ill} presents all five bipartite graphs consistent with the prescribed degree sequences. When $\mathscr{F}(1)=1$ of the node $[\mathbf{b}]_1$, four feasible graphs are possible, while only one graph is possible when $\mathscr{F}(1)=2$. As a result, a uniform choice between $\mathscr{F}(1)=1$ and $\mathscr{F}(1)=2$ leads to a non-uniform distribution over the set of feasible bipartite graphs. To ensure uniform sampling over the five feasible graphs, $\mathscr{F}(1)$ should be set to 1 with probability 0.8 and to 2 with probability 0.2. That is, the selection probability should be proportional to the number of feasible graphs that can be generated from each choice of the number of connected nodes at this step.
Therefore, the key challenge for sampling algorithms is to determine how to assign appropriate weights to the feasible values within the interval so as to achieve uniform sampling over graphs.
\begin{figure}[h!]
    \centering
    
    \begin{minipage}[t]{0.65\textwidth}
        \centering
        $\mathscr{F}(1) = 1$
        \vskip 1em

        \begin{tikzpicture}
            \node[anchor=north, inner sep=0pt] (FonePics) at (0,0) {%
            \begin{minipage}{\textwidth}
                \centering
                
                \begin{minipage}{0.49\textwidth}
                    \centering
                    \begin{tikzpicture}[baseline=(current bounding box.center), scale=0.7]
                        \node (a1) at (0,0) [circle,draw,minimum size=0.6cm] {2};
                        \node (a2) at (2,0) [circle,draw,minimum size=0.6cm] {1};
                        \node (a3) at (4,0) [circle,draw,minimum size=0.6cm] {1};
                        
                        \node (b1) at (0,-2) [circle,draw,minimum size=0.6cm] {2};
                        \node (b2) at (2,-2) [circle,draw,minimum size=0.6cm] {1};
                        \node (b3) at (4,-2) [circle,draw,minimum size=0.6cm] {1};
                        
                        \draw[thick] (b1) -- (a1);
                        \draw[thick] (b1) -- (a2);
                        \draw[thick] (b2) -- (a1);
                        \draw[thick] (b3) -- (a3);
                    \end{tikzpicture}
                \end{minipage}
                \hspace{-1.5em}
                \begin{minipage}{0.49\textwidth}
                    \centering
                    \begin{tikzpicture}[baseline=(current bounding box.center), scale=0.7]
                        \node (a1) at (0,0) [circle,draw,minimum size=0.6cm] {2};
                        \node (a2) at (2,0) [circle,draw,minimum size=0.6cm] {1};
                        \node (a3) at (4,0) [circle,draw,minimum size=0.6cm] {1};
                        
                        \node (b1) at (0,-2) [circle,draw,minimum size=0.6cm] {2};
                        \node (b2) at (2,-2) [circle,draw,minimum size=0.6cm] {1};
                        \node (b3) at (4,-2) [circle,draw,minimum size=0.6cm] {1};
                        
                        \draw[thick] (b1) -- (a1);
                        \draw[thick] (b1) -- (a2);
                        \draw[thick] (b2) -- (a3);
                        \draw[thick] (b3) -- (a1);
                    \end{tikzpicture}
                \end{minipage}
                \vskip 1em
                \begin{minipage}{0.49\textwidth}
                    \centering
                    \begin{tikzpicture}[baseline=(current bounding box.center), scale=0.7]
                        \node (a1) at (0,0) [circle,draw,minimum size=0.6cm] {2};
                        \node (a2) at (2,0) [circle,draw,minimum size=0.6cm] {1};
                        \node (a3) at (4,0) [circle,draw,minimum size=0.6cm] {1};
                        
                        \node (b1) at (0,-2) [circle,draw,minimum size=0.6cm] {2};
                        \node (b2) at (2,-2) [circle,draw,minimum size=0.6cm] {1};
                        \node (b3) at (4,-2) [circle,draw,minimum size=0.6cm] {1};
            
                        \draw[thick] (b1) -- (a1);
                        \draw[thick] (b1) -- (a3);
                        \draw[thick] (b2) -- (a1);
                        \draw[thick] (b3) -- (a2);
                    \end{tikzpicture}
                \end{minipage}
                \hspace{-1.5em}
                \begin{minipage}{0.49\textwidth}
                    \centering
                    \begin{tikzpicture}[baseline=(current bounding box.center), scale=0.7]
                        \node (a1) at (0,0) [circle,draw,minimum size=0.6cm] {2};
                        \node (a2) at (2,0) [circle,draw,minimum size=0.6cm] {1};
                        \node (a3) at (4,0) [circle,draw,minimum size=0.6cm] {1};
                        
                        \node (b1) at (0,-2) [circle,draw,minimum size=0.6cm] {2};
                        \node (b2) at (2,-2) [circle,draw,minimum size=0.6cm] {1};
                        \node (b3) at (4,-2) [circle,draw,minimum size=0.6cm] {1};
                        
                        \draw[thick] (b1) -- (a1);
                        \draw[thick] (b1) -- (a3);
                        \draw[thick] (b2) -- (a2);
                        \draw[thick] (b3) -- (a1);
                    \end{tikzpicture}
                \end{minipage}
        
            \end{minipage}
        };

        \node[draw, dashed, thick, rounded corners, inner ysep=10pt, inner xsep=-6pt, fit=(FonePics)] {};
        \end{tikzpicture}

    \end{minipage}
    \hfill
    \begin{minipage}[t]{0.3\textwidth}
        \centering
        $\mathscr{F}(1) = 2$
        \vskip 1em

        \begin{tikzpicture}
            \node[anchor=north, inner sep=0pt] (FtwoPics) at (0,0) {%
            
            \begin{minipage}{\textwidth}
                \centering
                \begin{tikzpicture}[baseline=(current bounding box.center), scale=0.7]
                    \node (a1) at (0,0) [circle,draw,minimum size=0.6cm] {2};
                    \node (a2) at (2,0) [circle,draw,minimum size=0.6cm] {1};
                    \node (a3) at (4,0) [circle,draw,minimum size=0.6cm] {1};
                    
                    \node (b1) at (0,-2) [circle,draw,minimum size=0.6cm] {2};
                    \node (b2) at (2,-2) [circle,draw,minimum size=0.6cm] {1};
                    \node (b3) at (4,-2) [circle,draw,minimum size=0.6cm] {1};
        
                    \draw[thick] (b1) -- (a2);
                    \draw[thick] (b1) -- (a3);
                    \draw[thick] (b2) -- (a1);
                    \draw[thick] (b3) -- (a1);
                \end{tikzpicture}
            \end{minipage}
        };
        \node[draw, dashed, thick, rounded corners, inner ysep=10pt, inner xsep=-6pt, fit=(FtwoPics)] {};
    \end{tikzpicture}
    \end{minipage}
    
    \vskip 1em
    \caption{\baselineskip10pt All the bipartite graphs in Example~\ref{example1} with degree sequences $\mathbf{a} = \mathbf{b} = (2,\,1,\,1)$.}
\label{fig:weight_ill}
\end{figure}
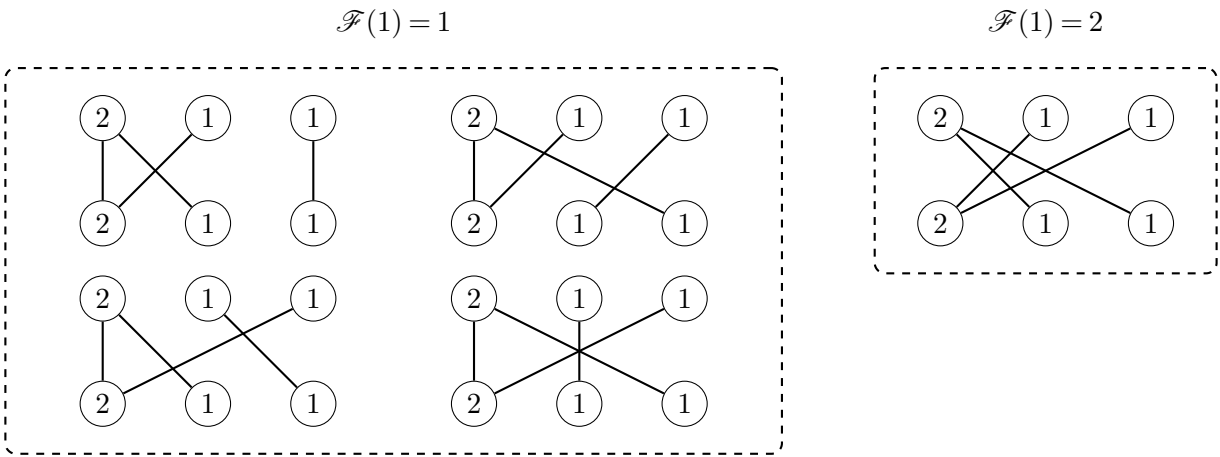

Building on Theorem \ref{thm:interval}, we propose two weight-based sampling algorithms for generating feasible bipartite graphs from the graph space $\mathcal{G}\!\left(\mathbf{a}, \mathbf{b}\right)$. The first algorithm, uniform bipartite graph sampling (UBGS), achieves uniform sampling by assigning weights proportional to the number of bipartite graphs that can be completed from each partial construction. This algorithm, however, requires additional computation and relatively high memory usage to evaluate these quantities at each step. To address this issue, we propose a second algorithm, efficient bipartite graph sampling (EBGS), which replaces the exact counts with approximate weights, reducing computational and memory costs while preserving approximately uniform sampling. 

\subsubsection{Uniform Bipartite Graph Sampling Algorithm.}

In this algorithm, we need to determine, for each pair of degree sequences $(\mathbf{a}^{[j]},\mathbf{b}^{[j]})$ with $j=1,\ldots,n$ and for each $k=1,\ldots,p$, the number of distinct bipartite graphs associated with each feasible $\mathscr{F}_i\!\left(k\right)$, where $\mathscr{F}_i \!\left(k\right)$ denotes the $i$-th feasible value of $\mathscr{F} \!\left(k\right)$ in the interval $[\ell_k,u_k]$ given by Theorem~\ref{thm:interval}. Specifically, let $N_B\!\left(\mathbf{a}^{[j]},\mathbf{b}^{[j]},\mathscr{F} \!\left(k\right)\right)$ denote the number of distinct bipartite graphs that can be generated from $\mathscr{F} \!\left(k\right)$, given that $\left[\mathbf{b}\right]_j$ has already completed its connections for group $1$ through group $k-1$; that is, $\mathscr{F}\!\left(1\right),\ldots,\mathscr{F}(k-\nobreak1)$ have been determined. Likewise, define $N_B\!\left(\mathbf{a}^{[j]},\mathbf{b}^{[j]},\mathscr{F}_i \!\left(k\right)\right)$ as the number of distinct bipartite graphs that can be generated from the $i$-th feasible value $\mathscr{F}_i \!\left(k\right)$ of $\mathscr{F}\!\left(k\right)$. Then, we establish the following proposition, which characterizes a recursive relation between $N_B\!\left(\mathbf{a}^{[j]},\mathbf{b}^{[j]},\mathscr{F}\!\left(k\right)\right)$ and $N_B\!\left(\mathbf{a}^{[j]},\mathbf{b}^{[j]},\mathscr{F}_i\!\left(k\right)\right)$. The proof is included in Section \ref{appx:T2-proof} of the e-companion.


  




\begin{proposition}\label{prop:N(a,b,F(k))}
    For each pair of degree sequences $(\mathbf{a}^{[j]},\mathbf{b}^{[j]})$ with $j=1,\ldots,n$ and for each $k=1,\ldots,p$, the numbers of distinct bipartite graphs that can be generated from $\mathscr{F}\!\left(k\right)$ and from its $i$-th feasible value $\mathscr{F}_i\!\left(k\right)$ are respectively given by
    \begin{equation}\label{eq:N_K-N_Ki}
    N_B\!\left(\mathbf{a}^{[j]},\mathbf{b}^{[j]},\mathscr{F} \!\left(k\right)\right) = \sum_{i=1}^{u_k-\ell_k+1} N_B\!\left(\mathbf{a}^{[j]},\mathbf{b}^{[j]},\mathscr{F}_i \!\left(k\right)\right)
    \end{equation}
    and
    \begin{equation}\label{eq:N(a,b,k)}
    N_B\!\left(\mathbf{a}^{[j]},\mathbf{b}^{[j]},\mathscr{F}_i \!\left(k\right)\right) = 
    \begin{cases}
        \displaystyle
        \binom{m_k}{\mathscr{F}_i \!\left(k\right)} N_B\!\left(\mathbf{a}^{[j]},\mathbf{b}^{[j]},\mathscr{F}\!\left(k+1\right)\right), & \text{if } k < p, \\[10pt]
        \displaystyle
        \binom{m_k}{\mathscr{F}_i \!\left(k\right)} N_B\!\left(\mathbf{a}^{[j+1]},\mathbf{b}^{[j+1]},\mathscr{F}\!\left(1\right)\right), & \text{otherwise.}
    \end{cases}
    \end{equation}
\end{proposition}

By applying \eqref{eq:N_K-N_Ki} and \eqref{eq:N(a,b,k)} recursively, together with the initial condition $N_B\!\left(\mathbf{0},\mathbf{0}, 0\right) = 1$, we can compute all quantities $N_B\!\left(\mathbf{a}^{[j]},\mathbf{b}^{[j]},\mathscr{F}_i\!\left(k\right)\right)$, that is, the numbers of distinct bipartite graphs associated with each feasible $\mathscr{F}_i\!\left(k\right)$. Consequently, the exact weight vector $\mathbf{w}^{\mathrm{UB}}$ for sampling $\mathscr{F}\!\left(k\right)$ can be obtained, where each entry is given by
\begin{equation}\label{eq:weight-for-uniform}
    \left[\mathbf{w}^{\mathrm{UB}}\right]_i = N_B\!\left(\mathbf{a}^{[j]},\mathbf{b}^{[j]},\mathscr{F}_i \!\left(k\right)\right).
\end{equation}

Next, we present the UBGS algorithm, described in Algorithm~\ref{alg:UBGS}, which builds on the sequential method, Theorem~\ref{thm:interval}, and the uniform weight formula \eqref{eq:weight-for-uniform}. Specifically, the algorithm begins with the degree sequences $\mathbf{a}$ and $\mathbf{b}$, and an initially empty bipartite graph $\mathbf{G}\!\left(\mathbf{a}, \mathbf{b}\right)$. For each node $\left[\mathbf{b}\right]_j$, the nodes in $\mathbf{a}^{[j]}$ are partitioned into groups according to their distinct degree values, where group $k$ contains $m_k$ nodes of degree $\alpha_k$. For each group, the interval condition for $\mathscr{F} \!\left(k\right)$ is obtained from Theorem~\ref{thm:interval}, and the exact weight vector $\mathbf{w^{\mathrm{UB}}}$ is computed using \eqref{eq:weight-for-uniform}. Once the value of $\mathscr{F} \!\left(k\right)$ is sampled according to $\mathbf{w}^{\mathrm{UB}}$, $\mathscr{F} \!\left(k\right)$ nodes are uniformly selected from group $k$ to connect with $\left[\mathbf{b}\right]_j$. After processing $\left[\mathbf{b}\right]_j$, the sequences $\mathbf{a}^{[j]}$ and $\mathbf{b}^{[j]}$ are updated by decrementing the degrees of the connected nodes and removing any node whose degree becomes zero. By repeating this procedure for all nodes in $\mathbf{b}$, the algorithm generates a bipartite graph $\mathbf{G}\!\left(\mathbf{a}, \mathbf{b}\right)$ that satisfies the prescribed degree sequences.

\begin{algorithm}[htpb]
\caption{Uniform Bipartite Graph Sampling (UBGS) Algorithm}
\label{alg:UBGS}
{\renewcommand{\baselinestretch}{0.75}\selectfont
\begin{algorithmic}[1]
\Require Degree sequences $\mathbf{a} \in \mathbb{N}_{+}^{m}$ and $\mathbf{b} \in \mathbb{N}_{+}^{n}$, an initially empty bipartite graph $\mathbf{G}\!\left(\mathbf{a}, \mathbf{b}\right)$
\Ensure A bipartite graph $\mathbf{G}\!\left(\mathbf{a}, \mathbf{b}\right)$

\For{$j = 1$ \textbf{to} $n$}
  \State Group $\mathbf{a}^{[j]}$ into $p$ groups $(\alpha_k, m_k)$ with $\alpha_1 < \alpha_2 < \cdots < \alpha_p$

  \For{$k = 1$ \textbf{to} $p$}
    \State Compute the interval condition $[\ell_k,u_k]$ via \eqref{eq:Fk}
    \State Obtain the exact weight vector $\mathbf{w}^{\mathrm{UB}}$ via \eqref{eq:weight-for-uniform}
    \State Sample $\mathscr{F} \!\left(k\right) \in [\ell_k,u_k]$ according to $\mathbf{w}^{\mathrm{UB}}$
    \State Uniformly select $\mathscr{F} \!\left(k\right)$ distinct nodes in group $k$ and connect them to $\left[\mathbf{b}\right]_j$ in $\mathbf{G}\!\left(\mathbf{a}, \mathbf{b}\right)$
  \EndFor
\EndFor

\State \textbf{return} $\mathbf{G}\!\left(\mathbf{a}, \mathbf{b}\right)$
\end{algorithmic}
}
\end{algorithm}

Notice that computing $N_B\!\left(\mathbf{a}^{[j]},\mathbf{b}^{[j]},\mathscr{F}_i\!\left(k\right)\right)$ at every sampling step of the UBGS algorithm is computationally expensive, since these values must be recomputed each time a feasible value $\mathscr{F}_i\!\left(k\right)$ is considered. A natural way to avoid this repeated cost is to precompute all feasible values of $N_B\!\left(\mathbf{a}^{[j]},\mathbf{b}^{[j]},\mathscr{F}_i\!\left(k\right)\right)$ and store them in a lookup table that can be reused across multiple sampling runs. Once constructed, such a lookup table eliminates the need for on-the-fly computation and thus substantially accelerates the UBGS algorithm.

\subsubsection{Efficient Bipartite Graph Sampling Algorithm.}
Constructing the lookup table itself requires additional computation, and storing it incurs relatively high memory usage, both of which grow rapidly as the problem size increases, even though the lookup table can be reused once built. To mitigate these computational and memory burdens, we introduce approximate weights in place of the exact ones, thereby avoiding the construction of the lookup table altogether and reducing both time and memory costs, while still maintaining approximately uniform sampling.

As mentioned above, the weight assigned to the feasible value $\mathscr{F}_i \!\left(k\right)$ for the node $\left[\mathbf{b}\right]_j$  should be proportional to the number of distinct bipartite graphs
that can be generated from it; this quantity is determined by the following three processes:
\begin{enumerate}
    \item Select $\mathscr{F}_i \!\left(k\right)$ nodes from the $m_k$ nodes in group $k$ of $\mathbf{a}$. The number of distinct bipartite graphs generated in this step is denoted by $N_1$, given by
    \begin{equation*}
        N_1 = \binom{m_k}{\mathscr{F}_i \!\left(k\right)}.
    \end{equation*}
    \item Select $\mathscr{F}\!\left(k+1\right), \dots, \mathscr{F}\!\left(p\right)$ nodes from groups $k+1, \dots, p$ of $\mathbf{a}$, respectively. The number of distinct bipartite graphs generated here is denoted by $N_2$, which we approximate as
    \begin{equation*}
        \widehat{N}_2 = \binom{M^{(k)} + m_k - \mathscr{F}_i \!\left(k\right)}{\left[\mathbf{b}\right]_j^{(k)} - \mathscr{F}_i \!\left(k\right)}.
    \end{equation*}
    \item Generate the bipartite graph with the updated degree sequences $\mathbf{a}^{[j+1]}$ and $\mathbf{b}^{[j+1]}$ after processing $\left[\mathbf{b}\right]_j$. For computational efficiency, we assume that the number of distinct bipartite graphs generated in this step remains the same for different feasible values of $\mathscr{F} \!\left(k\right)$.
\end{enumerate}

Consequently, the approximate number of distinct bipartite graphs associated with $\mathscr{F}_i \!\left(k\right)$, denoted by $\widehat{N}_B\!\left(\mathbf{a}^{[j]},\mathbf{b}^{[j]},\mathscr{F}_i\!\left(k\right)\right)$, is $N_1 \times \widehat{N}_2$. This yields the approximate weight vector $\mathbf{w}^{\mathrm{EB}}$ for efficient sampling, where each entry is given by
\begin{equation}\label{eq:weight-for-efficient}
    \left[\mathbf{w}^{\mathrm{EB}}\right]_i = \widehat{N}_B\!\left(\mathbf{a}^{[j]},\mathbf{b}^{[j]},\mathscr{F}_i \!\left(k\right)\right) = \binom{m_k}{\mathscr{F}_i \!\left(k\right)} \cdot \binom{M^{(k)} + m_k - \mathscr{F}_i \!\left(k\right)}{\left[\mathbf{b}\right]_j^{(k)} - \mathscr{F}_i \!\left(k\right)}.
\end{equation}
The EBGS algorithm is then obtained by replacing $\mathbf{w}^{\mathrm{UB}}$ with $\mathbf{w}^{\mathrm{EB}}$ in Step~5 of Algorithm~\ref{alg:UBGS}. 
Pseudocode for the EBGS algorithm is provided as Algorithm~\ref{alg:EBGS} in Section~\ref{appx:alg-EBGS} of the e-companion. 
Notice that the samples generated by the EBGS algorithm are not completely uniform over the bipartite graph space $\mathcal{G}\!\left(\mathbf{a},\mathbf{b}\right)$. Following the idea of \cite{Glasserman2023}, importance sampling can be employed to correct for this nonuniformity when the objective is to compute expectations with respect to the uniform distribution on $\mathcal{G}\!\left(\mathbf{a},\mathbf{b}\right)$, with details included in Section~\ref{appx:IS-EBGS} of the e-companion.

\section{Extended Algorithms}\label{sec:Extend}
In this section, we extend the bipartite graph algorithms to both directed and undirected graphs by introducing additional self-loop and single-degree connection constraints, as well as  feasibility verification and symmetric connection steps. For each graph type, we develop an enumeration algorithm that exhaustively generates all feasible graph instances $\mathbf{G}\!\left(\mathbf{a},\mathbf{b}\right)$ within the corresponding graph space $\mathcal{G} \!\left(\mathbf{a}, \mathbf{b}\right)$ for small-scale settings, and two sampling algorithms that respectively perform uniform and efficient yet slightly biased sampling over the same graph space for larger-scale problems. The algorithms for directed and undirected graphs are presented in Sections~\ref{ssec:directed_graph} and \ref{ssec:undirected_graph}, respectively.

\subsection{Directed Graph Algorithms}\label{ssec:directed_graph}

In this subsection, we first introduce the self-loop and single-degree connection constraints and the feasibility verification step that enable the extension of the bipartite graph algorithms to the directed case. We then develop the corresponding enumeration and sampling algorithms for directed graphs, building on the bipartite graph algorithms and incorporating these additional components.

A directed graph, in which the vectors $\mathbf{a} \in \mathbb{N}_{0}^{m}$ and $\mathbf{b} \in \mathbb{N}_{0}^{n}$ represent the out-degree and in-degree sequences of nodes, respectively, can be regarded as a bipartite graph without self-loops once the nodes of degree zero in $\mathbf{a}$ and $\mathbf{b}$ are ignored. Therefore, to extend the bipartite graph algorithms to the directed case, we introduce additional self-loop and single-degree connection constraints and feasibility verification step that together exclude self-loops and ensure the generation of directed graphs.

To exclude self-loops in the generation of directed graphs, it is necessary to ensure that each $\left[\mathbf{b}\right]_j$ is not connected to its counterpart node in $\mathbf{a}^{[j]}$. Enforcing this constraint directly would require a backtracking procedure to discard configurations that violate the self-loop exclusion constraint, which is computationally inefficient; see Section~\ref{appx:exclude-self-loops} of the e-companion for details. To avoid this inefficiency, we instead adopt an indirect yet structurally convenient transformation. Specifically, we increase the degree of every node in $\mathbf{a}$ and $\mathbf{b}$ by one, including those with zero degree, to obtain augmented degree sequences. We then require each node $\left[\mathbf{b}\right]_j$ to form a self-loop with its counterpart node, which we denote by $\left[\mathbf{a}\right]_j$ for convenience, and let $k_j$ denote the index of the group containing this node. After generating a bipartite graph consistent with the augmented degree sequences, all self-loops are removed, yielding a directed graph without self-loops that satisfies the original degree sequences $\mathbf{a}$ and $\mathbf{b}$. During the generation of the augmented bipartite graph, the requirement that each node forms a self-loop with its counterpart constitutes a self-loop connection constraint.

In addition to the self-loop connection constraint, we impose a single-degree connection constraint to address a special case. When each node $\left[\mathbf{b}\right]_j$ connects to nodes of degree one in group~$1$ of $\mathbf{a}^{[j]}$, it cannot connect to nodes whose only remaining degree must be reserved for forming their self-loops.

Based on the self-loop and single-degree connection constraints, we establish the following theorem, which updates the interval condition in Theorem~\ref{thm:interval} for bipartite graphs to accommodate the augmented degree sequences used in the directed setting. The proof is included in Section \ref{appx:directed-interval-proof} of the e-companion.

\begin{theorem}\label{thm:directed-interval}
    Suppose that a directed graph exists for the prescribed degree sequences $\mathbf{a}$ and $\mathbf{b}$. To enforce the self-loop and single-degree connection constraints in the directed setting, the interval in Theorem~\ref{thm:interval} for bipartite graphs is updated for the augmented degree sequences as follows:
    
    (i) For $k < k_j$:
    \begin{equation}\label{eq:directed-interval-1}
    \mathscr{F} \!\left(k\right) \in
    \begin{cases}
        \left[
        \max \bigl\{ \bigl[\mathbf{b}\bigr]_j^{(k)} - M^{(k)},\, 0 \bigr\},\, \ 
        \min \bigl\{ D_{\mathrm{left}}^{(k)} + D_{\mathrm{now}}^{(k)} + D_{\mathrm{right}}^{(k)},\, m_k - n_{\mathrm{one}},\, u_{\Delta}^{(k)}\bigr\}
        \right], & \textnormal{if } k = 1, \\
        \left[
        \max \bigl\{ \bigl[\mathbf{b}\bigr]_j^{(k)} - M^{(k)},\, 0 \bigr\},\, \ 
        \min \bigl\{ D_{\mathrm{left}}^{(k)} + D_{\mathrm{now}}^{(k)} + D_{\mathrm{right}}^{(k)},\, m_k,\, u_{\Delta}^{(k)}\bigr\}
        \right], & \textnormal{if } k > 1 \\
    \end{cases}
    \end{equation}

    (ii) For $k = k_j$:
    \begin{equation}\label{eq:directed-interval-2}
    \mathscr{F} \!\left(k\right) \in 
    \begin{cases}
        \left[
        \max \bigl\{ \bigl[\mathbf{b}\bigr]_j^{(k)} - M^{(k)},\, 1 \bigr\},\, \ 
        \min \bigl\{ D_{\mathrm{left}}^{(k)} + D_{\mathrm{now}}^{(k)} + D_{\mathrm{right}}^{(k)},\, m_k - n_{\mathrm{one}} + 1\bigr\}
        \right], & \textnormal{if } k = 1, \\
        \left[
        \max \bigl\{ \bigl[\mathbf{b}\bigr]_j^{(k)} - M^{(k)},\, 1 \bigr\},\, \ 
        \min \bigl\{ D_{\mathrm{left}}^{(k)} + D_{\mathrm{now}}^{(k)} + D_{\mathrm{right}}^{(k)},\, m_k\bigr\}
        \right], & \textnormal{if } k > 1 \\
    \end{cases}
    \end{equation}

    (iii) For $k > k_j$:

    \begin{equation}\label{eq:directed-interval-3}
    \mathscr{F} \!\left(k\right) \in \left[
        \max \bigl\{ \bigl[\mathbf{b}\bigr]_j^{(k)} - M^{(k)},\, 0 \bigr\},\, \ 
        \min \bigl\{ D_{\mathrm{left}}^{(k)} + D_{\mathrm{now}}^{(k)} + D_{\mathrm{right}}^{(k)},\, m_k\bigr\}
        \right]
    \end{equation}
    where $n_{\mathrm{one}}$ denotes the number of degree-one nodes in group~$1$ that have not yet formed self-loops with their counterpart; $u_{\Delta}^{(k)} = D_{\mathrm{left}}^{(1,k_j)} + D_{\mathrm{now}}^{(1,k_j)} + D_{\mathrm{right}}^{(1,k_j)} - \sum_{k'=1}^{k-1}\mathscr{F}\!\left(k'\right) - 1$, with
    \begin{equation*}
    D_{\mathrm{left}}^{(1,k_j)} = \sum_{h=1}^{{\alpha_{k_j}}-1} \left[\mathscr{D} \!\left(\mathbf{z}^{(1)}, \mathbf{b}^{[j+1]}\right)\right]_h, \quad D_{\mathrm{now}}^{(1,k_j)}= \left[\mathscr{D} \!\left(\mathbf{z}^{(1)}, \mathbf{b}^{[j+1]}\right)\right]_{\alpha_{k_j}},
    \end{equation*}
    and
    \begin{equation*}
        D_{\mathrm{right}}^{(1,k_j)}= \min\left\{\min_{r\in \{\alpha_{k_j}+1,\, \alpha_{k_j}+2,\, n-j+1\}} \sum_{h=\alpha_{k_j} + 1}^{r} \left[\mathscr{D} \!\left(\mathbf{z}^{(1)}, \mathbf{b}^{[j+1]}\right)\right]_h, 0\right\}.
    \end{equation*}
    
\end{theorem}

However, while the self-loop and single-degree connection constraints substantially restrict the admissible choices of $\mathscr{F} \!\left(k\right)$, they are not sufficient to guarantee the generation of directed graphs. In particular, certain combinations of $\mathscr{F} \!\left(k\right)$ for $\left[\mathbf{b}\right]_j$ may satisfy these local constraints, yet lead to updated degree sequences $\mathbf{a}^{[j+1]}$ and $\mathbf{b}^{[j+1]}$ that are infeasible for generating directed graphs.
To rule out such cases, we further apply a feasibility verification step based on the necessary and sufficient conditions for the existence of directed graphs (see, for instance, Theorem~2 in \citet{Berger2014}). Specifically, for this verification, we temporarily restore the original degree scale by reducing by one the degrees of nodes in the augmented degree sequences that have not yet formed self-loops, and then check whether the resulting degree sequences admit a directed graph. This additional feasibility verification step ensures that the connection configurations of $\left[\mathbf{b}\right]_j$ are globally feasible for directed graphs generation.

In summary, to extend the bipartite graph algorithms to the directed case, we introduce the following two connection constraints and one feasibility verification step:
\begin{itemize}
    \item Self-loop connection constraint. Each node $\left[\mathbf{b}\right]_j$ is required to form a self-loop with its counterpart node $\left[\mathbf{a}\right]_j$ in $\mathbf{a}^{[j]}$ after increasing the degree of all nodes in $\mathbf{a}$ and $\mathbf{b}$ by one, which is equivalent to $\left[\mathbf{b}\right]_j$ not forming a self-loop under the original degree sequences.

    \item Single-degree connection constraint. Each node $\left[\mathbf{b}\right]_j$ cannot connect to nodes of degree one that have not yet formed self-loops, as their only remaining degree must be reserved for forming their own self-loops.

    \item Feasibility verification step. After each $\left[\mathbf{b}\right]_j$ completes its connections, the updated degree sequences $\mathbf{a}^{[j+1]}$ and $\mathbf{b}^{[j+1]}$ are verified using the necessary and sufficient conditions for the existence of a directed graph, to ensure the connection configurations of $\left[\mathbf{b}\right]_j$ are feasible.
\end{itemize}

\subsubsection{Directed Graph Enumeration Algorithm}

Recall that the basic idea of the BGE algorithm is to iterate over all connection configurations corresponding to each feasible value within the interval specified by \eqref{eq:Fk} in Theorem~\ref{thm:interval}. In particular, the BGE algorithm adopts a two-layer breadth-first search, in which the outer search iterates over the nodes of $\mathbf{b}$, while the inner search explores the degree groups associated with each node. To extend the BGE algorithm to the directed case, we incorporate the self-loop and single-degree connection constraints, as well as the feasibility verification step, into the inner layer. Specifically, 
the directed graph enumeration (DGE) algorithm is obtained from Algorithm~\ref{alg:BGE} by the following modifications: (i) augmenting $\mathbf{a}$ and $\mathbf{b}$ by increasing the degree of each node by one at the initialization stage of Algorithm~\ref{alg:BGE}, (ii) updating the interval condition in Step~15 according to \eqref{eq:directed-interval-1}--\eqref{eq:directed-interval-3} in Theorem~\ref{thm:directed-interval}, (iii) connecting each node $\left[\mathbf{b}\right]_j$ to its counterpart node in group~$k_j$ to form the corresponding self-loop in Step~17, (iv) inserting the feasibility verification step after completing the connection process for each node $\left[\mathbf{b}\right]_j$,
and (v) removing all self-loops before adding each directed graph to the output set.
Pseudocode for the DGE algorithm is provided as Algorithm~\ref{alg:DGE} in Section~\ref{appx:alg-DGE} of the e-companion.

\subsubsection{Directed Graph Sampling Algorithms}

To extend the bipartite graph sampling algorithms to the directed case, we begin with the UBGS algorithm. Recall that the key point of the UBGS algorithm is to assign exact weights to each feasible value in the interval, where each weight should be proportional to the number of distinct bipartite graphs that can be generated from selecting that value. To extend the UBGS algorithm to the directed case, we first incorporate the self-loop and single-degree connection constraints, as well as the feasibility verification step, into the computation of the exact weights for directed graphs. Because the single-degree connection constraint and the feasibility verification step depend on which specific nodes are connected rather than only on the number of connections, the weight assignment must be performed at the level of each connection configuration rather than at the level of feasible values. Accordingly, the exact weight assigned to each configuration should be proportional to the number of distinct directed graphs that can be generated from that configuration.

Specifically, let $N_D\!\left(\mathbf{a}^{[j]},\mathbf{b}^{[j]},\mathscr{F} \!\left(k\right)\right)$ denote the number of distinct directed graphs that can be generated from $\mathscr{F} \!\left(k\right)$. Define $N_D\!\left(\mathbf{a}^{[j]},\mathbf{b}^{[j]},\mathscr{F}_{il} \!\left(k\right)\right)$ as the number of distinct directed graphs that can be generated from the $l$-th connection configuration associated with the $i$-th feasible value $\mathscr{F}_i \!\left(k\right)$ of $\mathscr{F}\!\left(k\right)$. We then establish the following proposition, which characterizes a recursive relation between $N_D\!\left(\mathbf{a}^{[j]},\mathbf{b}^{[j]},\mathscr{F}\!\left(k\right)\right)$ and $N_D\!\left(\mathbf{a}^{[j]},\mathbf{b}^{[j]},\mathscr{F}_{il}\!\left(k\right)\right)$. The proof is included in Section \ref{appx:T3-proof} of the e-companion.

\begin{proposition}\label{prop:ND(a,b,F(k))}
    For each pair of degree sequences $(\mathbf{a}^{[j]},\mathbf{b}^{[j]})$ with $j=1,\ldots,n$ and for each $k=1,\ldots,p$, the numbers of distinct directed graphs that can be generated from $\mathscr{F}\!\left(k\right)$ and from the $l$-th connection configuration associated with the $i$-th feasible value $\mathscr{F}_i\!\left(k\right)$ are respectively given by
    \begin{equation}\label{eq:ND_K-N_Ki}
    N_D\!\left(\mathbf{a}^{[j]},\mathbf{b}^{[j]},\mathscr{F} \!\left(k\right)\right) = \sum_{i=1}^{u_k-\ell_k+1} \sum_{l=1}^{\binom{m_k}{\mathscr{F}_i\!\left(k\right)}} N_D\!\left(\mathbf{a}^{[j]},\mathbf{b}^{[j]},\mathscr{F}_{il} \!\left(k\right)\right)
    \end{equation}
    and
    \begin{equation}\label{eq:ND(a,b,k)}
    N_D\!\left(\mathbf{a}^{[j]},\mathbf{b}^{[j]},\mathscr{F}_{il} \!\left(k\right)\right) = 
    \begin{cases}
        \displaystyle
         N_D\!\left(\mathbf{a}^{[j]},\mathbf{b}^{[j]},\mathscr{F}\!\left(k+1\right)\right), & \text{if } k < p, \\[10pt]
        \displaystyle
         N_D\!\left(\mathbf{a}^{[j+1]},\mathbf{b}^{[j+1]},\mathscr{F}\!\left(1\right)\right), & \text{otherwise.}
    \end{cases}
    \end{equation}
\end{proposition}

To incorporate the self-loop and single-degree connection constraints, as well as the feasibility verification step, the lower bound $\ell_k$ and the upper bound $u_k$ in \eqref{eq:ND_K-N_Ki} are updated according to \eqref{eq:directed-interval-1}--\eqref{eq:directed-interval-3}, and the following cases apply:
\begin{equation}\label{eq:directed-conditional-equalities}
\begin{cases}
N_D\!\left(\mathbf{a}^{[j]}, \mathbf{b}^{[j]}, \mathscr{F}_{il}\!\left(k_j\right)\right) = 0, & \text{if } \left[\mathbf{b}\right]_j \text{ does not connect its counterpart node,} \\
N_D\!\left(\mathbf{a}^{[j+1]}, \mathbf{b}^{[j+1]}, \mathscr{F}\!\left(1\right)\right) = 0, & \text{if the feasibility verification for } (\mathbf{a}^{[j+1]}, \mathbf{b}^{[j+1]}) \text{ fails.}
\end{cases}
\end{equation}

By recursively applying \eqref{eq:ND_K-N_Ki} with the updated bounds, together with \eqref{eq:ND(a,b,k)} and \eqref{eq:directed-conditional-equalities}, and using the initial condition $N_D(\mathbf{0},\mathbf{0},0)=~1$, we can compute all quantities $N_D\!\left(\mathbf{a}^{[j]},\mathbf{b}^{[j]},\mathscr{F}_{il}\!\left(k\right)\right)$. 
Consequently, the exact weight vector $\mathbf{w}^{\mathrm{UD}}$ for sampling each connection configuration can be obtained, where each entry is given by
\begin{equation}\label{eq:directed_weight-for-uniform}
    \left[\mathbf{w}^{\mathrm{UD}}\right]_{il} = N_D\!\left(\mathbf{a}^{[j]},\mathbf{b}^{[j]},\mathscr{F}_{il} \!\left(k\right)\right).
\end{equation}

Based on the exact weight vector $\mathbf{w}^{\mathrm{UD}}$ and the UBGS algorithm, we develop the uniform directed graph sampling (UDGS) algorithm. Specifically,
the UDGS algorithm is obtained from Algorithm~\ref{alg:UBGS} by the following modifications: (i) augmenting $\mathbf{a}$ and $\mathbf{b}$ by increasing the degree of each node by one at the initialization stage of Algorithm~\ref{alg:UBGS}, (ii) updating the interval condition in Step~4 according to \eqref{eq:directed-interval-1}--\eqref{eq:directed-interval-3}, (iii) replacing the sampling of feasible values in Steps~5--7 with the sampling of connection configurations, (iv) replacing the weights used in Step~5 with those given in \eqref{eq:directed_weight-for-uniform}, and (v) removing all self-loops before returning the final directed graph. Pseudocode for the UDGS algorithm is provided as Algorithm~\ref{alg:UDGS} in Section~\ref{appx:alg-UDGS} of the e-companion.

Next, we extend the EBGS algorithm to the directed case by incorporating the self-loop and single-degree connection constraints, as well as the feasibility verification step. Specifically, the efficient directed graph sampling (EDGS) algorithm is obtained from Algorithm~\ref{alg:EBGS} by the following modifications: (i) augmenting $\mathbf{a}$ and $\mathbf{b}$ by increasing the degree of each node by one at the initialization stage of Algorithm~\ref{alg:EBGS}, (ii) updating the interval condition in Step~4 according to \eqref{eq:directed-interval-1}--\eqref{eq:directed-interval-3}, (iii) connecting each node $\left[\mathbf{b}\right]_j$ to its counterpart node in group~$k_j$ to form the corresponding self-loop in Step~7, (iv) inserting the feasibility verification step after completing the connection process for each node $\left[\mathbf{b}\right]_j$, and (v) removing all self-loops before returning the final directed graph. Pseudocode for the EDGS algorithm is provided as Algorithm~\ref{alg:EDGS} in Section~\ref{appx:alg-EDGS} of the e-companion. Note that, similar to the EBGS algorithm, importance sampling can also be applied to the EDGS algorithm to correct for its nonuniformity when computing expectations with respect to the uniform distribution on $\mathcal{G}\!\left(\mathbf{a},\mathbf{b}\right)$.

\subsection{Undirected Graph Algorithms}\label{ssec:undirected_graph}

In this subsection, we first introduce the symmetric connection step that enables the extension of the directed graph algorithms to the undirected case. We then develop the corresponding enumeration and sampling algorithms for undirected graphs by incorporating this step into the directed graph algorithms.

An undirected graph, in which the vector $\mathbf{a} = \mathbf{b} \in \mathbb{N}_{+}^{m}$ represents the degree sequence of nodes, can be viewed as a directed graph where the out-degree and the in-degree sequences are identical. To adapt the directed graph algorithms to the undirected case, we introduce a symmetric connection step that enforces $\mathbf{a}^{[j]} = \mathbf{b}^{[j]}$ for all $j = 1,\ldots,n$. Specifically, after each node $\left[\mathbf{b}\right]_j$ completes its connections, its counterpart node $\left[\mathbf{a}\right]_j$ simultaneously establishes a symmetric set of connections, thereby mirroring each edge of $\left[\mathbf{b}\right]_j$ and ensuring that $\mathbf{a}^{[j+1]} = \mathbf{b}^{[j+1]}$. In addition, in the feasibility verification step, the necessary and sufficient conditions for the existence of a directed graph are replaced by the corresponding conditions for an undirected graph (see, for instance, the Erd\H{o}s--Gallai theorem; \citet{Erdos1960}). As a result of the symmetric connection step described above, the degree sequence $\mathbf{b}^{[j+1]}$ is updated, and the next node to be processed should be re-identified as the node with the largest remaining degree.

\subsubsection{Undirected Graph Enumeration Algorithm}

To extend the DGE algorithm to the undirected case, we incorporate the symmetric connection step into the inner layer of the algorithm. Specifically, the undirected graph enumeration (UGE) algorithm is obtained from Algorithm~\ref{alg:DGE} by the following modifications: (i) establishing the corresponding symmetric connections for the node $\left[\mathbf{a}\right]_j$ once each node $\left[\mathbf{b}\right]_j$ completes its connections, (ii) replacing the necessary and sufficient conditions for the existence of a directed graph in the feasibility verification step with the corresponding conditions for an undirected graph, and (iii) re-identifying the node with the largest remaining degree in $\mathbf{b}^{[j+1]}$ as the next node to be processed. Pseudocode for the UGE algorithm is provided as Algorithm~\ref{alg:UGE} in Section~\ref{appx:alg-UGE} of the e-companion.

\subsubsection{Undirected Graph Sampling Algorithms}

To extend the directed graph sampling algorithms to the undirected case, we begin with the UDGS algorithm. Recall that the key point of the UDGS algorithm is to assign exact weights to each connection configuration associated with each feasible value in the interval, where each weight should be proportional to the number of distinct directed graphs that can be generated from that configuration. We first incorporate the symmetric connection step into the computation of the exact weights for undirected graphs. In particular, in the recursive process for computing the exact weights, once each node $\left[\mathbf{b}\right]_j$ completes its connections, its counterpart node $\left[\mathbf{a}\right]_j$ establishes the corresponding symmetric connections. After the symmetric connection step, the feasibility verification step is carried out using the necessary and sufficient conditions for the existence of an undirected graph. Subsequently, for each updated degree sequence $\mathbf{b}^{[j+1]}$, the node with the largest remaining degree is re-identified as the next node to be processed. Therefore, the exact weight vector $\mathbf{w}^{\mathrm{UU}}$ for undirected case can be obtained.

Based on the exact weight vector $\mathbf{w}^{\mathrm{UU}}$ and the UDGS algorithm, we develop the uniform undirected graph sampling (UUGS) algorithm. Specifically, the UUGS algorithm is obtained from Algorithm~\ref{alg:UDGS} by the following modifications: (i) replacing the weights used in Step~7 with $\mathbf{w}^{\mathrm{UU}}$, (ii) establishing the corresponding symmetric connections for the node $\left[\mathbf{a}\right]_j$ once each node $\left[\mathbf{b}\right]_j$ completes its connections, and (iii) re-identifying the node with the largest remaining degree in $\mathbf{b}^{[j+1]}$ as the next node to be processed. Pseudocode for the UUGS algorithm is provided as Algorithm~\ref{alg:UUGS} in Section~\ref{appx:alg-UUGS} of the e-companion.

Next, we extend the EDGS algorithm to the undirected case by incorporating the symmetric connection step into the algorithm. Specifically, the efficient undirected graph sampling (EUGS) algorithm is obtained from Algorithm~\ref{alg:EDGS} by the following modifications: (i) establishing the corresponding symmetric connections for the node $\left[\mathbf{a}\right]_j$ once each node $\left[\mathbf{b}\right]_j$ completes its connections, (ii) replacing the necessary and sufficient conditions for the existence of a directed graph in the feasibility verification step with the corresponding conditions for an undirected graph, and (iii) re-identifying the node with the largest remaining degree in $\mathbf{b}^{[j+1]}$ as the next node to be processed. Pseudocode for the EUGS algorithm is provided as Algorithm~\ref{alg:EUGS} in Section~\ref{appx:alg-EUGS} of the e-companion. Note that, similar to the EBGS algorithm, importance sampling can also be applied to the EUGS algorithm to correct for its nonuniformity when computing expectations with respect to the uniform distribution on $\mathcal{G}\!\left(\mathbf{a},\mathbf{b}\right)$.

\section{Numerical Experiments}\label{sec:numerical-experiemnts}

In this section, we demonstrate the performance of the proposed enumeration and sampling algorithms. For the enumeration algorithms, we compare their runtime performance with that of a na\"{i}ve brute-force algorithm designed for reference.
For the sampling algorithms, we compare them with the {\em sequential algorithm} of \cite{Glasserman2023} and the {\em exact sampling algorithm} of \cite{Miller2013}, focusing on sampling speed, graph distributional uniformity, and graph space coverage ratio. The results for bipartite, directed, and undirected graph algorithms are presented in Sections~\ref{ssec:numerical_bipartite}, \ref{ssec:numerical_directed}, and \ref{ssec:numerical_undirected}, respectively. All numerical experiments are implemented in Python and executed on a desktop computer equipped with a 3.6~GHz Intel Core i7-12700K processor and 16~GB of RAM. 

\subsection{Bipartite Graphs}\label{ssec:numerical_bipartite}

We first evaluate the performance of the BGE algorithm in terms of enumeration runtime, and then evaluate the performance of the UBGS and EBGS algorithms in terms of sampling speed, uniformity, and graph space coverage ratio.

{\bf Enumeration runtime}. We compare the BGE algorithm with the following na\"{i}ve brute-force algorithm: given degree sequences $\mathbf{a}$ and $\mathbf{b}$, generate all graphs that satisfy only one of the two degree sequences, and then filter from them all feasible bipartite graphs that satisfy both degree sequences. Pseudocode for the brute-force algorithm is provided as Algorithm~\ref{alg:brute-force} in Section~\ref{appx:alg-brute-force} of the e-companion.

We implement the brute-force and BGE algorithms on six degree sequences and, for each sequence, record the enumeration runtime. 
We report the results in Table~\ref{tab:bg_enumeration_runtime}, which includes the number of distinct bipartite graphs ($N_B$), as well as the enumeration runtime of the brute-force algorithm ($T_{BF}$) and the BGE algorithm ($T_{BGE}$). The results demonstrate that BGE is several orders of magnitude faster than the brute-force algorithm. For example, BGE enumerates 1,372 bipartite graphs in roughly 0.02 seconds, and 15,732 bipartite graphs in 0.3 seconds. In contrast, the brute-force algorithm requires approximately 14 seconds and $10^{6}$ seconds, respectively.

\begin{table}[htpb!]
    \scriptsize
    \centering
    \caption{\baselineskip10pt Runtime of the deterministic algorithms for bipartite graphs.
    }
    \label{tab:bg_enumeration_runtime}
    \begin{tabular}{llccc}
        \toprule
        \multicolumn{2}{l}{Degree Sequence} 
        & \multirow{2}{*}{$N_B$}
            & \multicolumn{1}{c}{Brute-Force Algorithm} 
            & \multicolumn{1}{c}{BGE Algorithm} \\
        \cmidrule(lr){4-4} \cmidrule(lr){5-5}
        $\mathbf{a}$ & $\mathbf{b}$ 
            &  
            & $T_{BF}$ (sec)
            & $T_{BGE}$ (sec) \\
        \midrule
$(3,3,3,2,1)$ & $(5,3,2,1,1)$ 
    & $27$ &  $0.03$ 
    & $0.0003$ \\
$(3,2,2,2,2)$ & $(4,3,1,1,1,1)$ 
    & $688$ &  $0.4$ 
    & $0.007$ \\
$(5,5,4,4,3,2)$ & $(5,5,5,3,3,2)$ 
    & $1{,}372$  & $14$ 
    & $0.02$ \\
$(7,7,6,5,4,3,2,2)$ & $(7,7,6,5,4,3,3,1)$ 
    & $15{,}732$ & $\approx 10^{6}$ 
    & $0.3$ \\
$(8,7,7,6,5,4,2,2,2)$ & $(8,7,7,6,5,4,3,2,1)$ 
    & $194{,}010$ &  $\gg 10^{6}$ 
    & $2.9$ \\
$(8,7,6,6,5,4,2,2,2)$ & $(8,7,6,6,5,4,3,2,1)$ 
    & $1{,}991{,}820$ &  $\gg 10^{6}$ 
    & $31$ \\
        \bottomrule
    \end{tabular}
    
\end{table}


{\bf Sampling speed}. We evaluate the sampling speed of the UBGS and EBGS algorithms by comparing them with the sequential and exact sampling algorithms. Specifically, we implement all algorithms on the Interbank-1 degree sequence $\mathbf{a} = \mathbf{b} = (6,\,6,\,6,\,5,\,5,\,5,\,5,\,4,\,4,\,3,\,2)$ and the Interbank-2 degree sequence $\mathbf{a} = \mathbf{b} =(9,\,9,\,9,\,9,\,9,\,9,\,8,\,8,\,8,\,7,\,6)$, both of which were used in the numerical experiments of \cite{Glasserman2023}. For each degree sequence, every algorithm is executed 10,000 times, consistent with the setup in their paper. Table~\ref{tab:bg_Interbank} reports the average runtime of the original implementation of the sequential algorithm (as reported in their paper), our reproduced version, the exact sampling algorithm, UBGS, and EBGS. The results show that EBGS and UBGS are slightly faster than the exact sampling algorithm, and all three algorithms are substantially faster than the sequential algorithm across both degree sequences.

\begin{table}[htpb!]
    \scriptsize
    \centering
    \caption{\baselineskip10pt Runtime (mean, standard deviation) of the algorithms for the degree sequences from \cite{Glasserman2023}.}   
    \label{tab:bg_Interbank}
    \begin{tabular}{lccccc}
        \toprule
        \multirow{2}{*}{Degree Sequence} & \multicolumn{2}{c}{Sequential Algorithm} & \multirow{2}{*}{Exact Sampling Algorithm} & \multirow{2}{*}{UBGS Algorithm} & \multirow{2}{*}{EBGS Algorithm} \\
        \cmidrule(lr){2-3}
        & Original & Reproduced \\
        
        \midrule
        Interbank-1 & 0.17 & 0.15 (0.013) & 0.0028 (0.00065) & 0.0017 (0.00055) & 0.0013 (0.00052) \\
        Interbank-2 & 0.15 & 0.14 (0.017) & 0.0017 (0.00053) & 0.0016 (0.00048) & 0.0012 (0.00043) \\
        \bottomrule
    \end{tabular}
\end{table}

To further assess the scalability of these algorithms as the problem size grows, we consider the family of degree sequences $\mathbf{a} = \mathbf{b} = (n-1,\,n-1,\,n-2,\,n-3,\,\ldots,\,2,\,1)$, where $n$ denotes the number of nodes and $n \geq 3$. We evaluate UBGS and EBGS for $n = 3$ to 200, the sequential algorithm for $n = 3$ to 95, since its average runtime exceeds 100 seconds when $n \geq 95$, and the exact sampling algorithm for $n = 3$ to 26 due to the same runtime limitation. For each value of $n$, each experiment is repeated 20 times across all algorithms. Figure~\ref{fig:bg_sampling_speed} reports the average runtime of the sequential, exact sampling, UBGS, and EBGS algorithms as the number of nodes increases. The results show that EBGS and UBGS scale substantially better than the sequential algorithm, and all three algorithms exhibit better scalability than the exact sampling algorithm for this family of degree sequences.

\begin{figure}[htpb!]
    \centering
    \includegraphics[width=0.55\linewidth]{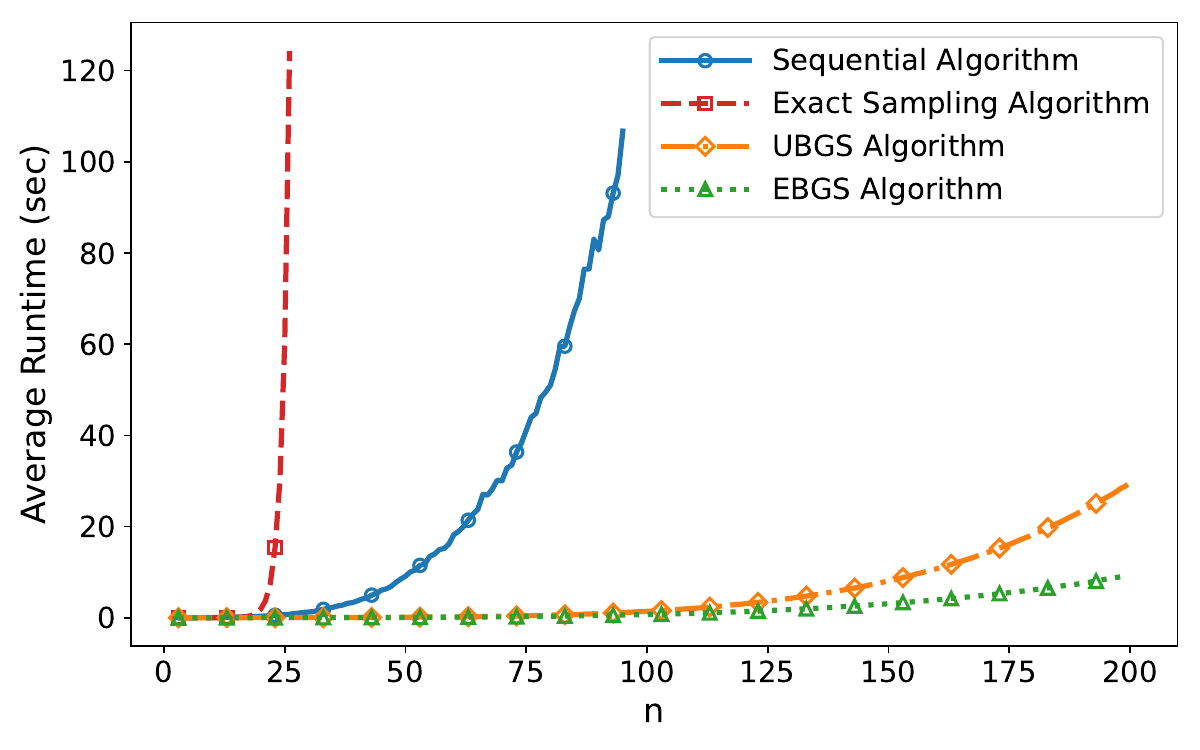}
    \caption{\baselineskip10pt Runtime as a function of the number of nodes for bipartite graphs with degree sequences $\mathbf{a} = \mathbf{b} = (n-1,\,n-1,\,n-2,\,n-3,\,\ldots,\,2,\,1)$ averaged over 20 replications.}
    \label{fig:bg_sampling_speed}
\end{figure}

{\bf Sampling uniformity}. We evaluate the sampling uniformity of all algorithms using two metrics: the coefficient of variation (CV), which measures the dispersion within each sampling distribution, and the Kullback--Leibler (KL) divergence between the sampling distribution and the corresponding uniform distribution, which quantifies the deviation from perfect uniformity. For both metrics, smaller values indicate better uniformity. The comparison is conducted on the four degree sequences used in the enumeration-runtime experiment, and for each sequence, all algorithms generate 5,000 independent samples. Table~\ref{tab:bg_sampling_uniformity} reports the resulting CV and KL values. The results show that, in terms of sampling uniformity, UBGS and the exact sampling algorithm achieve the highest levels of uniformity, followed by the sequential algorithm, while EBGS exhibits the lowest uniformity, although it remains reasonably good.

\begin{table}[htpb!]
    \scriptsize
    \centering
    \caption{\baselineskip10pt Sampling uniformity of the algorithms for bipartite graphs: coefficient of variation (CV) and Kullback--Leibler (KL) divergence; smaller is better.}
    \label{tab:bg_sampling_uniformity}
    \resizebox{\textwidth}{!}{
    \begin{tabular}{ll|cc|cc|cc|cc}
    \toprule
    \multicolumn{2}{l|}{Degree Sequence} 
        & \multicolumn{2}{c|}{Sequential Algorithm} 
        & \multicolumn{2}{c|}{Exact Algorithm} 
        & \multicolumn{2}{c|}{UBGS Algorithm} 
        & \multicolumn{2}{c}{EBGS Algorithm} \\
    \cmidrule(lr){3-4} \cmidrule(lr){5-6} \cmidrule(lr){7-8} \cmidrule(lr){9-10}
    $\mathbf{a}$ & $\mathbf{b}$ 
      & CV & KL 
      & CV & KL
      & CV & KL 
      & CV & KL \\
    \midrule
    $(3,3,3,2,1)$ & $(5,3,2,1,1)$
      & 9.99E-02 & 5.04E-03
      & 6.29E-02 & 1.98E-03
      & 7.31E-02 & 2.72E-03
      & 2.81E-01 & 3.95E-02 \\
    $(3,2,2,2,2)$ & $(4,3,1,1,1,1)$
      & 4.02E-01 & 8.36E-02
      & 3.71E-01 & 7.01E-02
      & 3.72E-01 & 6.98E-02
      & 4.92E-01 & 1.24E-01 \\
    $(5,5,4,4,3,2)$ & $(5,5,5,3,3,2)$
      & 5.41E-01 & 1.45E-01
      & 4.96E-01 & 1.23E-01
      & 4.91E-01 & 1.20E-01
      & 5.99E-01 & 1.69E-01 \\
    $(7,7,6,5,4,3,2,2)$ & $(7,7,6,5,4,3,3,1)$
      & 3.68E-01 & 5.46E-02
      & 3.59E-01 & 5.22E-02
      & 3.54E-01 & 5.12E-02
      & 4.10E-01 & 6.37E-02 \\
    \bottomrule
    \end{tabular}
    }
\end{table}

{\bf Graph space coverage ratio}. We evaluate the graph space coverage ratio of the UBGS and EBGS algorithms, comparing their performance with that of the sequential and exact sampling algorithms. The graph space coverage ratio is defined as the fraction of distinct feasible bipartite graphs generated by the algorithm over the entire graph space. We use one of the degree sequences from the enumeration-runtime experiment, namely $\mathbf{a} = (7,\,7,\,6,\,5,\,4,\,3,\,2,\,2)$ and $\mathbf{b} = (7,\,7,\,6,\,5,\,4,\,3,\,3,\,1)$, which correspond to 15,732 distinct bipartite graphs, as reported in the fourth row of Table~\ref{tab:bg_enumeration_runtime}. For each sampling algorithm, we vary the runtime budget from 10 to 100 seconds (in increments of 10 seconds) and record: (i) the total number of generated samples ($N_T$), (ii) the number of distinct bipartite graphs within $N_T$, i.e., $N_B$, and (iii) the coverage ratio $R_B = N_B / 15{,}732$, i.e., the fraction of distinct feasible bipartite graphs generated over the entire graph space. For each runtime setting, all experiments are repeated 20 times to compute the associated standard deviation.

Figure~\ref{fig:bg_sampling_coverage_ratio} reports the coverage ratio $R_B$ as a function of the runtime budget from 10 to 100 seconds for the four algorithms. Table~\ref{tab:bg_distinct_graph_coverage_ratio} summarizes the detailed numerical results at four representative runtime budgets (10, 20, 40, and 80 seconds), including $N_T$, $N_B$, and $R_B$. The results in Figure~\ref{fig:bg_sampling_coverage_ratio} and Table~\ref{tab:bg_distinct_graph_coverage_ratio} show that UBGS, EBGS, and the exact sampling algorithm consistently outperform the sequential algorithm in terms of {coverage ratio}. For example, within 10 seconds, UBGS, EBGS, and the exact sampling algorithm each generate more than $50\%$ of all distinct bipartite graphs, whereas the sequential algorithm produces less than $3\%$. Moreover, UBGS, EBGS, and the exact sampling algorithm are able to generate nearly all distinct bipartite graphs within 80 seconds, whereas the sequential algorithm produces less than $19\%$ under the same runtime budget. In addition, EBGS is slightly stronger than the exact sampling algorithm during the early stages of sampling but is eventually overtaken in the later stages, because the exact sampling algorithm generates samples according to the uniform distribution, whereas EBGS is only approximately uniform.

\begin{figure}[htpb!]
    \centering
    \includegraphics[width=0.55\linewidth]{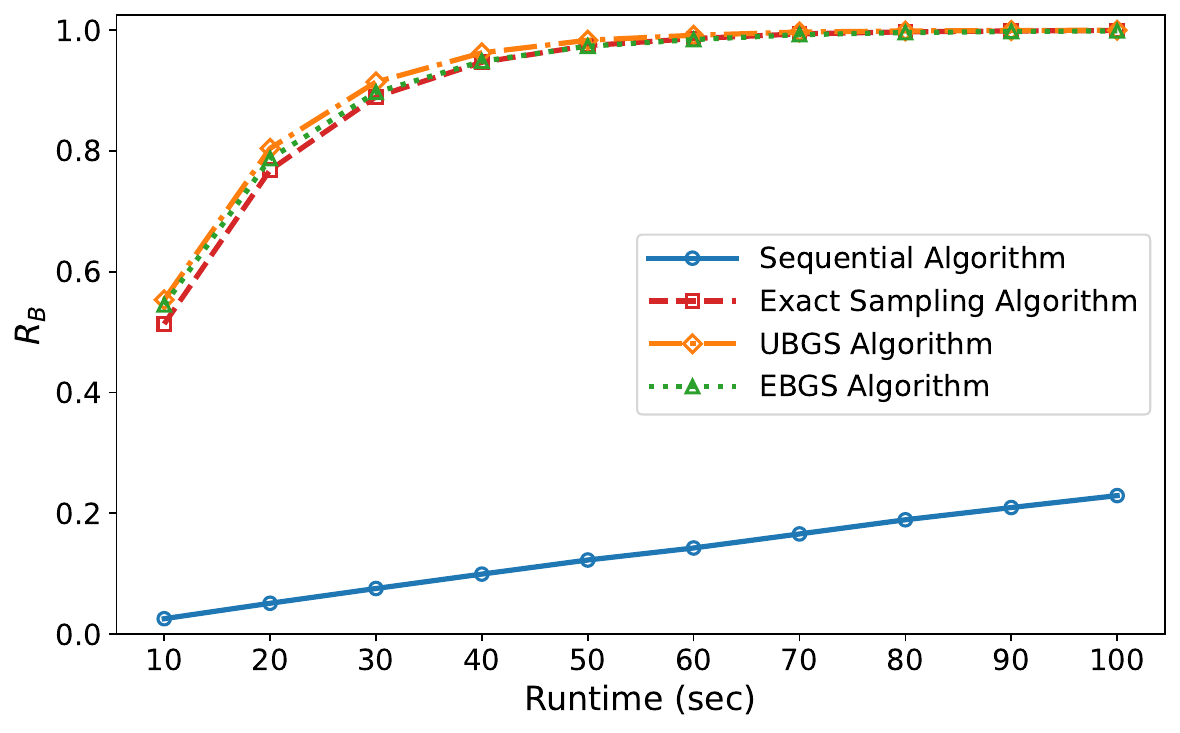}
    \caption{\baselineskip10pt Coverage ratio versus runtime for bipartite graphs with degree sequences $\mathbf{a} = (7,\,7,\,6,\,5,\,4,\,3,\,2,\,2)$ and $\mathbf{b} = (7,\,7,\,6,\,5,\,4,\,3,\,3,\,1)$.}
    \label{fig:bg_sampling_coverage_ratio}
\end{figure}

\begin{table}[htpb!]
    \scriptsize
    \centering
    \caption{\baselineskip10pt Coverage ratio for bipartite graphs with degree sequences $\mathbf{a} = (7,\,7,\,6,\,5,\,4,\,3,\,2,\,2)$ and $\mathbf{b} = (7,\,7,\,6,\,5,\,4,\,3,\,3,\,1)$ (standard deviation).}
    \label{tab:bg_distinct_graph_coverage_ratio}
    \resizebox{\textwidth}{!}{
\begin{tabular}{l|ccc|ccc|ccc|ccc}
    \toprule
    \multirow{2}{*}{Time (sec)} 
      & \multicolumn{3}{c|}{Sequential Algorithm} 
      & \multicolumn{3}{c|}{Exact Sampling Algorithm} 
      & \multicolumn{3}{c|}{UBGS Algorithm} 
      & \multicolumn{3}{c}{EBGS Algorithm}  \\
    \cmidrule(lr){2-4} \cmidrule(lr){5-7} \cmidrule(lr){8-10} \cmidrule(lr){11-13}
     & $N_T$ & $N_B$ & $R_B$
     & $N_T$ & $N_B$ & $R_B$
     & $N_T$ & $N_B$ & $R_B$
     & $N_T$ & $N_B$ & $R_B$ \\
    \midrule
    
    10 
     & \makecell[c]{\scriptsize 405.90\\[-1pt]\scriptsize(6.17)}
     & \makecell[c]{\scriptsize 400.55\\[-1pt]\scriptsize(6.72)}
     & \makecell[c]{\scriptsize\textbf{2.55\%}\\[-1pt]\scriptsize(0.04\%)}
     & \makecell[c]{\scriptsize 11,325.25\\[-1pt]\scriptsize(133.84)}
     & \makecell[c]{\scriptsize 8,078.30\\[-1pt]\scriptsize(63.22)}
     & \makecell[c]{\scriptsize\textbf{51.35\%}\\[-1pt]\scriptsize(0.40\%)}
     & \makecell[c]{\scriptsize 12,662.10\\[-1pt]\scriptsize(191.59)}
     & \makecell[c]{\scriptsize 8,707.70\\[-1pt]\scriptsize(96.45)}
     & \makecell[c]{\scriptsize\textbf{55.35\%}\\[-1pt]\scriptsize(0.61\%)}
     & \makecell[c]{\scriptsize 13,152.55\\[-1pt]\scriptsize(173.79)}
     & \makecell[c]{\scriptsize 8,575.30\\[-1pt]\scriptsize(84.69)}
     & \makecell[c]{\scriptsize\textbf{54.51\%}\\[-1pt]\scriptsize(0.54\%)} \\[6pt]

    20
     & \makecell[c]{\scriptsize 825.75\\[-1pt]\scriptsize(6.66)}
     & \makecell[c]{\scriptsize 802.80\\[-1pt]\scriptsize(6.30)}
     & \makecell[c]{\scriptsize\textbf{5.10\%}\\[-1pt]\scriptsize(0.04\%)}
     & \makecell[c]{\scriptsize 23,033.90\\[-1pt]\scriptsize(135.10)}
     & \makecell[c]{\scriptsize 12,088.40\\[-1pt]\scriptsize(40.15)}
     & \makecell[c]{\scriptsize\textbf{76.84\%}\\[-1pt]\scriptsize(0.26\%)}
     & \makecell[c]{\scriptsize 25,695.35\\[-1pt]\scriptsize(255.34)}
     & \makecell[c]{\scriptsize 12,651.50\\[-1pt]\scriptsize(64.39)}
     & \makecell[c]{\scriptsize\textbf{80.42\%}\\[-1pt]\scriptsize(0.41\%)}
     & \makecell[c]{\scriptsize 26,723.80\\[-1pt]\scriptsize(256.41)}
     & \makecell[c]{\scriptsize 12,384.20\\[-1pt]\scriptsize(62.48)}
     & \makecell[c]{\scriptsize\textbf{78.72\%}\\[-1pt]\scriptsize(0.40\%)} \\[6pt]

    40
     & \makecell[c]{\scriptsize 1,654.40\\[-1pt]\scriptsize(4.99)}
     & \makecell[c]{\scriptsize 1,562.50\\[-1pt]\scriptsize(8.64)}
     & \makecell[c]{\scriptsize\textbf{9.93\%}\\[-1pt]\scriptsize(0.05\%)}
     & \makecell[c]{\scriptsize 46,084.45\\[-1pt]\scriptsize(137.64)}
     & \makecell[c]{\scriptsize 14,894.45\\[-1pt]\scriptsize(24.85)}
     & \makecell[c]{\scriptsize\textbf{94.68\%}\\[-1pt]\scriptsize(0.16\%)}
     & \makecell[c]{\scriptsize 51,393.90\\[-1pt]\scriptsize(456.89)}
     & \makecell[c]{\scriptsize 15,140.35\\[-1pt]\scriptsize(23.80)}
     & \makecell[c]{\scriptsize\textbf{96.24\%}\\[-1pt]\scriptsize(0.15\%)}
     & \makecell[c]{\scriptsize 53,574.75\\[-1pt]\scriptsize(439.99)}
     & \makecell[c]{\scriptsize 14,920.10\\[-1pt]\scriptsize(32.99)}
     & \makecell[c]{\scriptsize\textbf{94.84\%}\\[-1pt]\scriptsize(0.21\%)} \\[6pt]

    80
     & \makecell[c]{\scriptsize 3,313.30\\[-1pt]\scriptsize(10.20)}
     & \makecell[c]{\scriptsize 2,977.45\\[-1pt]\scriptsize(18.71)}
     & \makecell[c]{\scriptsize\textbf{18.93\%}\\[-1pt]\scriptsize(0.12\%)}
     & \makecell[c]{\scriptsize 92,190.10\\[-1pt]\scriptsize(406.69)}
     & \makecell[c]{\scriptsize 15,686.65\\[-1pt]\scriptsize(7.22)}
     & \makecell[c]{\scriptsize\textbf{99.71\%}\\[-1pt]\scriptsize(0.05\%)}
     & \makecell[c]{\scriptsize 103,307.30\\[-1pt]\scriptsize(884.87)}
     & \makecell[c]{\scriptsize 15,709.55\\[-1pt]\scriptsize(6.72)}
     & \makecell[c]{\scriptsize\textbf{99.86\%}\\[-1pt]\scriptsize(0.04\%)}
     & \makecell[c]{\scriptsize 107,535.45\\[-1pt]\scriptsize(947.81)}
     & \makecell[c]{\scriptsize 15,666.55\\[-1pt]\scriptsize(5.80)}
     & \makecell[c]{\scriptsize\textbf{99.58\%}\\[-1pt]\scriptsize(0.04\%)} \\
    \bottomrule
\end{tabular}
}
\end{table}

\subsection{Directed Graphs}\label{ssec:numerical_directed}

We first evaluate the performance of the DGE algorithm in terms of enumeration runtime, and then evaluate the performance of the UDGS and EDGS algorithms in terms of sampling speed, uniformity, and graph space coverage ratio.

{\bf Enumeration runtime}. We evaluate the enumeration runtime of the DGE algorithm by comparing it with the brute-force algorithm. Specifically, we applied both algorithms to six degree sequences and, for each sequence, record the enumeration runtime. We report the results in Table~\ref{tab:dg_enumeration_runtime}, which includes the number of distinct directed graphs ($N_D$), as well as the enumeration runtime of the brute-force algorithm ($T_{BF}$) and the DGE algorithm ($T_{DGE}$). The results demonstrate that DGE is several orders of magnitude faster than the brute-force algorithm. For example, DGE enumerates 4,628 directed graphs in roughly 0.2 seconds, and 48,700 directed graphs in 1.8 seconds. In contrast, the brute-force algorithm requires approximately 380 seconds and $10^4$ seconds, respectively.

\begin{table}[htpb!]
    \scriptsize
    \centering
    \caption{\baselineskip10pt Runtime of the deterministic algorithms for directed graphs.
    }
    \label{tab:dg_enumeration_runtime}
		\begin{tabular}{lccc}
        \toprule
        \multicolumn{1}{l}{Degree Sequence}
        & \multirow{2}{*}{$N_D$}
            & \multicolumn{1}{c}{Brute-Force Algorithm} 
            & \multicolumn{1}{c}{DGE Algorithm} \\
        \cmidrule(lr){3-3} \cmidrule(lr){4-4}
        $\mathbf{a} = \mathbf{b}$ 
            &  
            & $T_{BF}$ (sec)
            & $T_{DGE}$ (sec) \\
        \midrule
$(2, 2, 1, 1, 1)$ & $42$ & $0.1$ & $0.002$ \\
$(2, 2, 2, 1, 1, 1)$ & $686$ & $7.6$ & $0.03$ \\
$(3, 2, 2, 1, 1, 1, 1)$ & $4{,}628$ & $380$ & $0.2$ \\
$(3, 3, 2, 1, 1, 1, 1, 1)$ & $48{,}700$ & $\approx 10^{4}$ & $1.8$ \\
$(3, 2, 2, 2, 1, 1, 1, 1)$ & $191{,}334$ & $\approx 10^{4}$ & $6.8$ \\
$(3, 3, 2, 2, 2, 2, 1, 1)$ & $1{,}977{,}666$ & $\gg 10^{4}$ & $67$ \\
        \bottomrule
        \end{tabular}
    
    \vspace{0.2em}
\end{table}

{\bf Sampling speed.} We evaluate the sampling speed of the UDGS and EDGS algorithms by comparing them with the sequential algorithm. Specifically, we consider the family of degree sequences $\mathbf{a} = (n-2,\,1,\,1,\,2,\,3,\,\dots,\,n-3,\,n-2)$ and $\mathbf{b} = (n-1,\,n-3,\,n-3,\,n-4,\,n-5,\,\dots,\,2,\,1,\,1)$, where $n$ denotes the number of nodes and $n \geq 4$. We evaluate UDGS and EDGS for $n = 4$ to 200, and the sequential algorithm for $n = 4$ to 78, since its average runtime exceeds 100 seconds when $n \geq 78$. For each value of $n$, each experiment is repeated 20 times across all algorithms. Figure~\ref{fig:dg_sampling_speed} reports the average runtime of the sequential, UDGS, and EDGS algorithms as the number of nodes increases. The results show that EDGS and UDGS run faster and scale substantially better than the sequential algorithm for this family of degree sequences.

\begin{figure}[htpb!]
    \centering
    \includegraphics[width=0.55\linewidth]{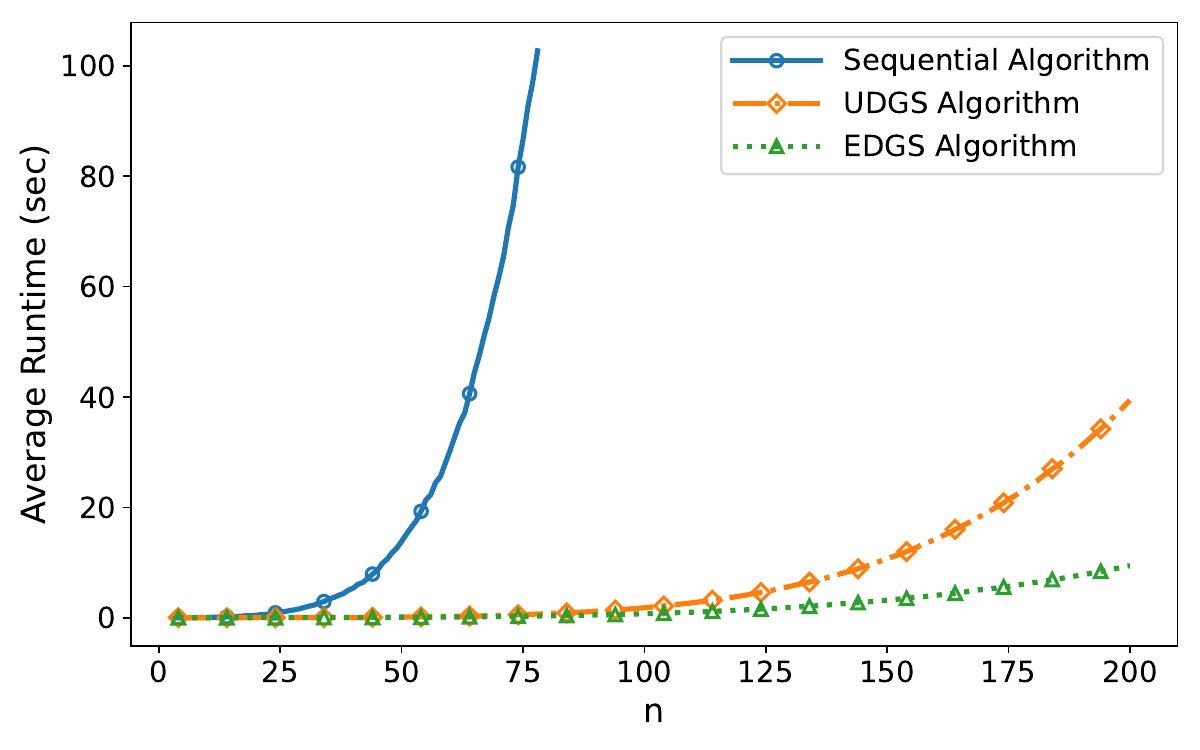}
    \caption{\baselineskip10pt Runtime as a function of the number of nodes for directed graphs with degree sequences $\mathbf{a} = (n-2,\,1,\,1,\,2,\,3,\,\dots,\,n-3,\,n-2)$ and $\mathbf{b} = (n-1,\,n-3,\,n-3,\,n-4,\,n-5,\,\dots,\,2,\,1,\,1)$ averaged over 20 replications.}
    \label{fig:dg_sampling_speed}
\end{figure}

{\bf Sampling uniformity}. We evaluate the sampling uniformity of the three algorithms using the same two metrics as in the bipartite graph experiment, namely CV and KL divergence. The comparison is conducted on the four degree sequences used in the enumeration-runtime experiments, and for each sequence, all algorithms generate 5,000 independent samples. Table~\ref{tab:dg_sampling_uniformity} reports the resulting CV and KL values. The results show that, in terms of sampling uniformity, UDGS achieves the highest level of uniformity, followed by the sequential algorithm, while EDGS exhibits the lowest uniformity, although it remains reasonably good.

\begin{table}[htpb!]
    \scriptsize
    \centering
    \caption{\baselineskip10pt Sampling uniformity of the algorithms for directed graphs.}
    \label{tab:dg_sampling_uniformity}
\begin{tabular}{l|cc|cc|cc}
    \toprule
    \multicolumn{1}{l|}{Degree Sequence} & \multicolumn{2}{c|}{Sequential Algorithm} & \multicolumn{2}{c|}{UDGS Algorithm} & \multicolumn{2}{c}{EDGS Algorithm} \\
    \cmidrule(lr){2-3} \cmidrule(lr){4-5} \cmidrule(lr){6-7} 
    $\mathbf{a} = \mathbf{b}$ & CV & KL & CV & KL & CV & KL \\
    \midrule
    $(2, 2, 1, 1, 1)$ 
      & 1.25E-01 & 7.45E-03
      & 8.85E-02 & 3.88E-03
      & 4.00E-01 & 8.34E-02 \\
    $(2, 2, 2, 1, 1, 1)$
      & 4.04E-01 & 8.25E-02
      & 3.57E-01 & 6.53E-02
      & 7.34E-01 & 2.33E-01 \\
    $(3, 2, 2, 1, 1, 1, 1)$
      & 5.29E-01 & 1.21E-01
      & 5.22E-01 & 1.19E-01
      & 6.25E-01 & 1.64E-01 \\
    $(3, 3, 2, 1, 1, 1, 1, 1)$
      & 2.22E-01 & 1.95E-02
      & 2.13E-01 & 1.78E-02
      & 2.59E-01 & 2.58E-02 \\
    \bottomrule
\end{tabular}
\end{table}

{\bf Graph space coverage ratio}. We evaluate the graph space coverage ratio of the UDGS and EDGS algorithms, comparing their performance with that of the sequential algorithm. We use one of the degree sequences from the enumeration-runtime experiment, namely $\mathbf{a} = \mathbf{b} = (3,\,3,\,2,\,1,\,1,\,1,\,1,\,1)$, which corresponds to 48,700 distinct directed graphs, as reported in the fourth row of Table~\ref{tab:dg_enumeration_runtime}. For each sampling algorithm, we vary the runtime budget from 10 to 100 seconds (in increments of 10 seconds) and record: (i) the total number of generated samples ($N_T$), (ii) the number of distinct directed graphs within $N_T$, i.e., $N_D$, and (iii) the coverage ratio $R_D = N_D / 48{,}700$, i.e., the fraction of distinct feasible directed graphs generated over the entire graph space. For each runtime setting, all experiments are repeated 20 times to compute the associated standard deviation.

Figure~\ref{fig:dg_sampling_coverage_ratio} reports the coverage ratio $R_D$ as a function of the runtime budget from 10 to 100 seconds for the three algorithms. Table~\ref{tab:dg_distinct_graph_coverage_ratio} summarizes the detailed numerical results at four representative runtime budgets (10, 20, 40, and 80 seconds), including $N_T$, $N_D$, and $R_D$. The results in Figure~\ref{fig:dg_sampling_coverage_ratio} and Table~\ref{tab:dg_distinct_graph_coverage_ratio} show that UDGS and EDGS consistently outperform the sequential algorithm in terms of coverage ratio. For example, within 10 seconds, UDGS generates more than $37\%$ of all distinct directed graphs and EDGS more than $29\%$, whereas the sequential algorithm produces less than $1\%$. Moreover, within 80 seconds, UDGS generates more than $97\%$ of all distinct directed graphs and EDGS more than $86\%$, whereas the sequential algorithm produces less than $6\%$.

\begin{figure}[htpb!]
    \centering
    \includegraphics[width=0.55\linewidth]{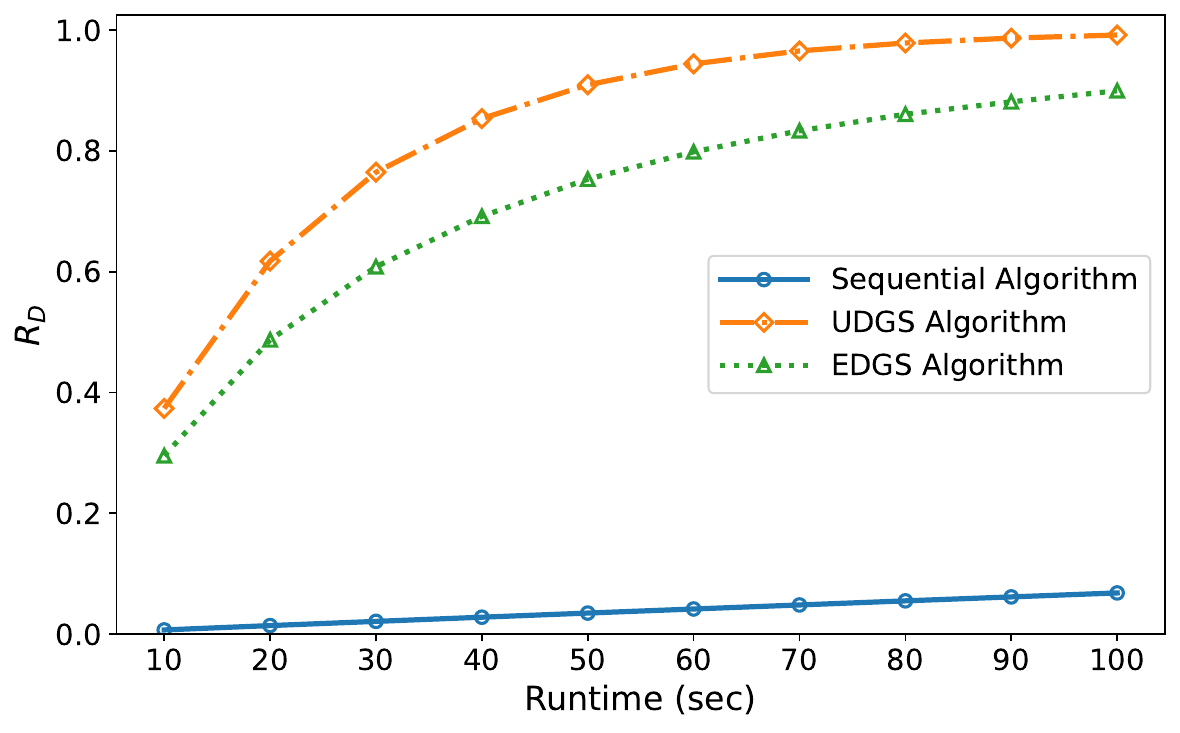}
    \caption{\baselineskip10pt Coverage ratio versus runtime for directed graphs with degree sequences $\mathbf{a} = \mathbf{b} = (3,\,3,\,2,\,1,\,1,\,1,\,1,\,1)$.}
    \label{fig:dg_sampling_coverage_ratio}
\end{figure}

\begin{table}[htpb!]
    \scriptsize
    \centering
    \caption{\baselineskip10pt Coverage ratio for directed graphs with degree sequences $\mathbf{a} = \mathbf{b} = (3,\,3,\,2,\,1,\,1,\,1,\,1,\,1)$.}
    \label{tab:dg_distinct_graph_coverage_ratio}
\begin{tabular}{l|ccc|ccc|ccc}
    \toprule
    \multirow{2}{*}{Time (sec)} 
      & \multicolumn{3}{c|}{Sequential Algorithm} 
      & \multicolumn{3}{c|}{UDGS Algorithm} 
      & \multicolumn{3}{c}{EDGS Algorithm}  \\
    \cmidrule(lr){2-4} \cmidrule(lr){5-7} \cmidrule(lr){8-10}
     & \makebox[3em][c]{$T_D$} & \makebox[3em][c]{$N_D$} & \makebox[3em][c]{$R_D$}
     & \makebox[3em][c]{$T_D$} & \makebox[3em][c]{$N_D$} & \makebox[3em][c]{$R_D$}
     & \makebox[3em][c]{$T_D$} & \makebox[3em][c]{$N_D$} & \makebox[3em][c]{$R_D$} \\
    \midrule

    10 
     & \makecell[c]{\scriptsize 335.15\\[-1pt]\scriptsize(3.86)}
     & \makecell[c]{\scriptsize 333.55\\[-1pt]\scriptsize(3.56)}
     & \makecell[c]{\scriptsize\textbf{0.68\%}\\[-1pt]\scriptsize(0.01\%)}
     & \makecell[c]{\scriptsize 22,767.40\\[-1pt]\scriptsize(162.00)}
     & \makecell[c]{\scriptsize 18,192.80\\[-1pt]\scriptsize(114.18)}
     & \makecell[c]{\scriptsize\textbf{37.36\%}\\[-1pt]\scriptsize(0.23\%)}
     & \makecell[c]{\scriptsize 18,753.40\\[-1pt]\scriptsize(131.41)}
     & \makecell[c]{\scriptsize 14,384.00\\[-1pt]\scriptsize(81.65)}
     & \makecell[c]{\scriptsize\textbf{29.54\%}\\[-1pt]\scriptsize(0.17\%)} \\[6pt]

    20
     & \makecell[c]{\scriptsize 690.15\\[-1pt]\scriptsize(7.62)}
     & \makecell[c]{\scriptsize 685.40\\[-1pt]\scriptsize(6.82)}
     & \makecell[c]{\scriptsize\textbf{1.41\%}\\[-1pt]\scriptsize(0.01\%)}
     & \makecell[c]{\scriptsize 46,783.25\\[-1pt]\scriptsize(281.79)}
     & \makecell[c]{\scriptsize 30,075.55\\[-1pt]\scriptsize(114.31)}
     & \makecell[c]{\scriptsize\textbf{61.76\%}\\[-1pt]\scriptsize(0.23\%)}
     & \makecell[c]{\scriptsize 38,745.10\\[-1pt]\scriptsize(182.53)}
     & \makecell[c]{\scriptsize 23,721.45\\[-1pt]\scriptsize(77.69)}
     & \makecell[c]{\scriptsize\textbf{48.71\%}\\[-1pt]\scriptsize(0.16\%)} \\[6pt]

    40
     & \makecell[c]{\scriptsize 1,380.00\\[-1pt]\scriptsize(8.34)}
     & \makecell[c]{\scriptsize 1,360.90\\[-1pt]\scriptsize(9.73)}
     & \makecell[c]{\scriptsize\textbf{2.79\%}\\[-1pt]\scriptsize(0.02\%)}
     & \makecell[c]{\scriptsize 93,568.65\\[-1pt]\scriptsize(778.12)}
     & \makecell[c]{\scriptsize 41,567.50\\[-1pt]\scriptsize(152.09)}
     & \makecell[c]{\scriptsize\textbf{85.35\%}\\[-1pt]\scriptsize(0.31\%)}
     & \makecell[c]{\scriptsize 77,256.95\\[-1pt]\scriptsize(651.99)}
     & \makecell[c]{\scriptsize 33,687.10\\[-1pt]\scriptsize(136.34)}
     & \makecell[c]{\scriptsize\textbf{69.17\%}\\[-1pt]\scriptsize(0.28\%)} \\[6pt]

    80
     & \makecell[c]{\scriptsize 2,757.40\\[-1pt]\scriptsize(15.82)}
     & \makecell[c]{\scriptsize 2,678.45\\[-1pt]\scriptsize(13.17)}
     & \makecell[c]{\scriptsize\textbf{5.50\%}\\[-1pt]\scriptsize(0.03\%)}
     & \makecell[c]{\scriptsize 187,559.90\\[-1pt]\scriptsize(1,082.68)}
     & \makecell[c]{\scriptsize 47,660.65\\[-1pt]\scriptsize(38.54)}
     & \makecell[c]{\scriptsize\textbf{97.87\%}\\[-1pt]\scriptsize(0.08\%)}
     & \makecell[c]{\scriptsize 155,001.65\\[-1pt]\scriptsize(1,022.14)}
     & \makecell[c]{\scriptsize 41,915.80\\[-1pt]\scriptsize(81.17)}
     & \makecell[c]{\scriptsize\textbf{86.07\%}\\[-1pt]\scriptsize(0.17\%)} \\

    \bottomrule
\end{tabular}
\end{table}

\subsection{Undirected Graphs}\label{ssec:numerical_undirected}

We first evaluate the performance of the UGE algorithm in terms of enumeration runtime, and then evaluate the performance of the UUGS and EUGS algorithms in terms of sampling speed, uniformity, and graph space coverage ratio.

{\bf Enumeration runtime}. We evaluate the enumeration runtime of the UGE algorithm by comparting it with the brute-force algorithm. Specifically, we applied both algorithms to six degree sequences and, for each sequence, record the enumeration runtime. We report the results in Table~\ref{tab:ug_enumeration_runtime}, which includes the number of distinct undirected graphs ($N_U$), as well as the enumeration runtime of the brute-force algorithm ($T_{BF}$) and the UGE algorithm ($T_{UGE}$). The results demonstrate that UGE is several orders of magnitude faster than the brute-force algorithm. For example, UGE enumerates 130 undirected graphs in roughly 0.01 seconds, and 1,074 undirected graphs in 0.1 seconds. In contrast, the brute-force algorithm requires approximately $10^3$ seconds and $10^6$ seconds, respectively.

\begin{table}[htpb!]
    \scriptsize
    \centering
    \caption{\baselineskip10pt Runtime of the deterministic algorithms for undirected graphs.
    }
    \label{tab:ug_enumeration_runtime}
		\begin{tabular}{lccc}
        \toprule
            \multicolumn{1}{l}{Degree Sequence} 
            & \multirow{2}{*}{$N_U$}
            & \multicolumn{1}{c}{Brute-Force Algorithm} 
            & \multicolumn{1}{c}{UGE Algorithm} \\
        \cmidrule(lr){3-3} \cmidrule(lr){4-4}
        $\mathbf{a} = \mathbf{b}$ 
            & 
            & $T_{BF}$ (sec)
            & $T_{UGE}$ (sec) \\
        \midrule
$(2,2,2,2,1,1)$ & $31$ & $19$ & $0.003$ \\
$(3,3,2,2,2,1,1)$ & $130$ & $\approx 10^{3}$ & $0.01$ \\
$(3,3,2,2,2,2,1,1)$ & $1{,}074$ & $\approx 10^{6}$ & $0.1$ \\
$(3,3,3,3,2,2,2,1,1)$ & $13{,}801$ & $\gg 10^{6}$ & $1.8$ \\
$(4,3,3,3,3,2,2,2,1,1)$ & $187{,}426$ & $\gg 10^{6}$ & $27$ \\
$(4,3,3,3,3,3,2,2,1,1,1)$ & $1{,}542{,}445$ & $\gg 10^{6}$ & $263$ \\
        \bottomrule
        \end{tabular}
    
\end{table}

{\bf Sampling speed.} We evaluate the sampling speed of the UUGS and EUGS algorithms by comparing them with the sequential algorithm. Specifically, we consider the following family of degree sequences, whose number of feasible undirected graphs increases with the number of nodes:
\begin{equation}\label{eq:undirected-degree-sequences}
\mathbf{a} = \mathbf{b} =
\bigl(
    c+a,\, c+b,\, 
    \underbrace{c, \dots, c}_{c-1\ \text{times}},\, 
    \underbrace{1, \dots, 1}_{l\ \text{times}}
\bigr),
\end{equation}
where $c = n-l-1$, $a = \lceil l/2 \rceil$, $b = l - a$, $l = \lfloor \log_2 n \rfloor$, $n$ denotes the number of nodes, and $n \geq 8$. We evaluate the sequential, UUGS, and EUGS algorithms for $n = 8$ to 200. For each value of $n$, each experiment is repeated 20 times across all algorithms. Figure~\ref{fig:ug_sampling_speed} reports the average runtime of the three algorithms as the number of nodes increases. The results show that EUGS and UUGS run faster and scale substantially better than the sequential algorithm for this family of degree sequences.

\begin{figure}[htpb!]
    \centering
    \includegraphics[width=0.6\linewidth]{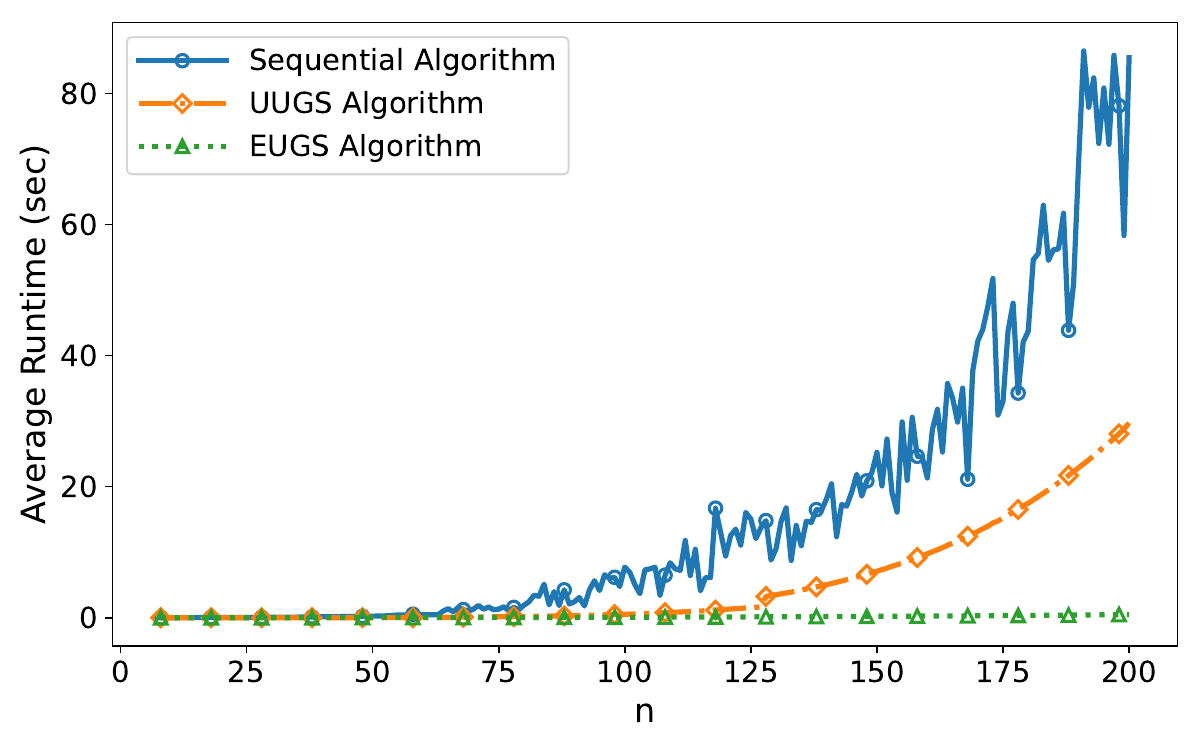}
    \caption{\baselineskip10pt Runtime as a function of the number of nodes for undirected graphs with degree sequences given in \eqref{eq:undirected-degree-sequences}, averaged over 20 replications.}
    \label{fig:ug_sampling_speed}
\end{figure}

{\bf Sampling uniformity}. We evaluate the sampling uniformity of the three algorithms using the same two metrics as in the bipartite graph experiments, namely CV and KL divergence. The comparison is conducted on the four degree sequences used in the enumeration-runtime experiment, and for each sequence, all algorithms generate 5,000 independent samples. Table~\ref{tab:ug_sampling_uniformity} reports the resulting CV and KL values. The results show that, in terms of sampling uniformity, UUGS achieves the highest level of uniformity, followed by the sequential algorithm, while EUGS exhibits the lowest uniformity, although it remains reasonably good.

\begin{table}[htpb!]
    \scriptsize
    \centering
    \caption{\baselineskip10pt Sampling uniformity of algorithms for undirected graphs.}
    \label{tab:ug_sampling_uniformity}
\begin{tabular}{l|cc|cc|cc}
    \toprule
    \multicolumn{1}{l|}{Degree Sequence} 
        & \multicolumn{2}{c|}{Sequential Algorithm} 
        & \multicolumn{2}{c|}{UUGS Algorithm} 
        & \multicolumn{2}{c}{EUGS Algorithm} \\
    \cmidrule(lr){2-3} \cmidrule(lr){4-5} \cmidrule(lr){6-7}
    $\mathbf{a} = \mathbf{b}$ & CV & KL & CV & KL & CV & KL \\
    \midrule
    $(2, 2, 2, 2, 1, 1)$
      & 1.22E-01 & 7.46E-03
      & 8.16E-02 & 3.38E-03
      & 3.70E-01 & 4.68E-02 \\
    $(3, 3, 2, 2, 2, 1, 1)$
      & 1.67E-01 & 1.38E-02
      & 1.58E-01 & 1.24E-02
      & 4.06E-01 & 6.25E-02 \\
    $(3, 3, 2, 2, 2, 2, 1, 1)$
      & 4.55E-01 & 1.07E-01
      & 4.54E-01 & 1.05E-01
      & 6.33E-01 & 1.75E-01 \\
    $(3, 3, 3, 3, 2, 2, 2, 1, 1)$
      & 3.80E-01 & 5.88E-02
      & 3.70E-01 & 5.57E-02
      & 4.77E-01 & 8.73E-02 \\
    \bottomrule
\end{tabular}
\end{table}

{\bf Graph space coverage ratio}. We evaluate the graph space coverage ratio of the UUGS and EUGS algorithms, comparing their performance with that of the sequential algorithm. We use one of the degree sequence from the enumeration-runtime experiment, namely $\mathbf{a} = \mathbf{b} = (3,\,3,\,3,\,3,\,2,\,2,\,2,\,1,\,1)$, which corresponds to 13,801 distinct undirected graphs, as reported in the fourth row of Table~\ref{tab:ug_enumeration_runtime}. For each sampling algorithm, we vary the runtime budget from 10 to 100 seconds (in increments of 10 seconds) and record: (i) the total number of generated samples ($N_T$), (ii) the number of distinct undirected graphs within $N_T$, i.e., $N_U$, and (iii) the coverage ratio $R_U = N_U / 13{,}801$, i.e., the fraction of distinct feasible undirected graphs generated over the entire graph space. For each runtime setting, all experiments are repeated 20 times to compute the associated standard deviation.

Figure~\ref{fig:ug_sampling_coverage_ratio} reports the coverage ratio $R_U$ as a function of the runtime budget from 10 to 100 seconds for the three algorithms. Table~\ref{tab:ug_distinct_graph_coverage_ratio} summarizes the detailed numerical results at three representative runtime budgets (10, 20, and 40 seconds), including $N_T$, $N_U$, and $R_U$. The results in Figure~\ref{fig:ug_sampling_coverage_ratio} and Table~\ref{tab:ug_distinct_graph_coverage_ratio} show that UUGS and EUGS consistently outperform the sequential algorithm in terms of coverage ratio. For example, within 10 seconds, UUGS generates more than $78\%$ of all distinct undirected graphs and EUGS more than $64\%$, whereas the sequential algorithm produces less than $4\%$. Moreover, within 40 seconds, UUGS generates more than $99\%$ of all distinct undirected graphs and EUGS more than $96\%$, whereas the sequential algorithm produces less than $13\%$.

\begin{figure}[htpb!]
    \centering
    \includegraphics[width=0.55\linewidth]{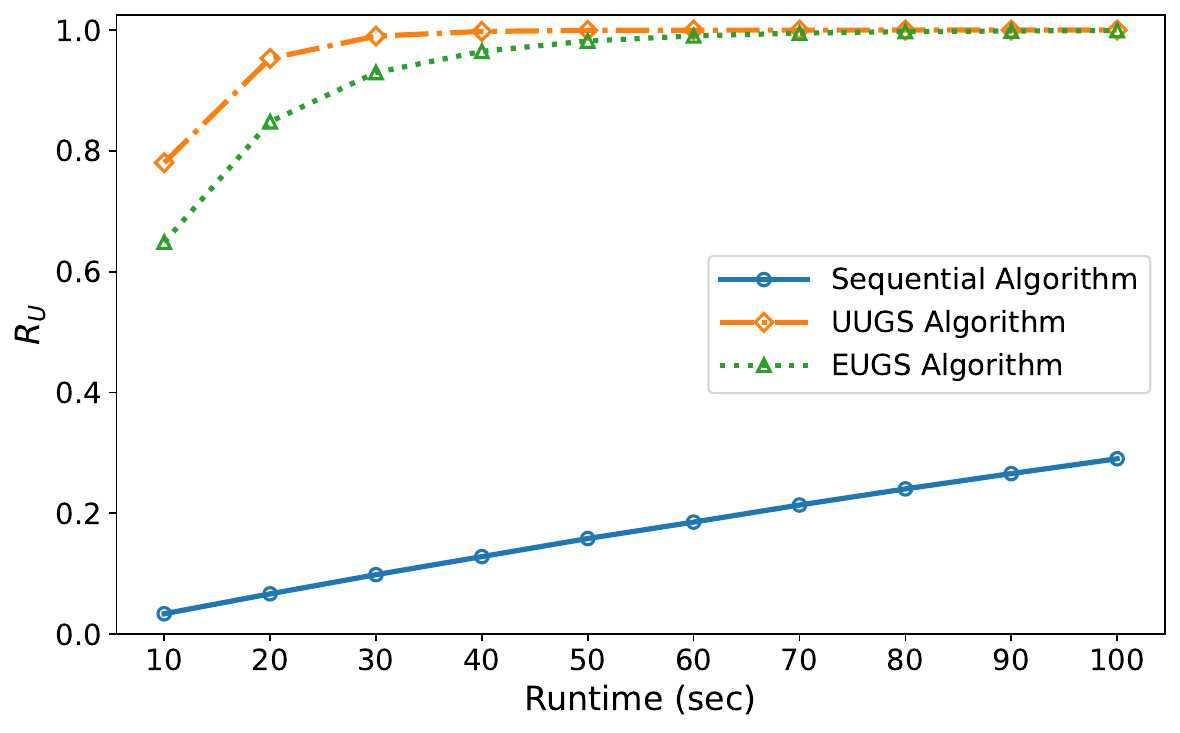}
    \caption{\baselineskip10pt Coverage ratio versus runtime for directed graphs with degree sequences $\mathbf{a} = \mathbf{b} = (3,\,3,\,3,\,3,\,2,\,2,\,2,\,1,\,1)$.}
    \label{fig:ug_sampling_coverage_ratio}
\end{figure}

\begin{table}[htpb!]
    \scriptsize
    \centering
    \caption{\baselineskip10pt Coverage ratio for undirected graphs with degree sequences $\mathbf{a} = \mathbf{b} = (3,\,3,\,3,\,3,\,2,\,2,\,2,\,1,\,1)$.}
    \label{tab:ug_distinct_graph_coverage_ratio}
\begin{tabular}{l|ccc|ccc|ccc}
    \toprule
    \multirow{2}{*}{Time (sec)} 
      & \multicolumn{3}{c|}{Sequential Algorithm} 
      & \multicolumn{3}{c|}{UUGS Algorithm} 
      & \multicolumn{3}{c}{EUGS Algorithm}  \\
    \cmidrule(lr){2-4} \cmidrule(lr){5-7} \cmidrule(lr){8-10}
     & \makebox[3em][c]{$T_U$} & \makebox[3em][c]{$N_U$} & \makebox[3em][c]{$R_U$}
     & \makebox[3em][c]{$T_U$} & \makebox[3em][c]{$N_U$} & \makebox[3em][c]{$R_U$}
     & \makebox[3em][c]{$T_U$} & \makebox[3em][c]{$N_U$} & \makebox[3em][c]{$R_U$} \\
    \midrule
    10
     & \makecell[c]{\scriptsize 473.15\\[-1pt]\scriptsize(6.08)}
     & \makecell[c]{\scriptsize 464.35\\[-1pt]\scriptsize(7.96)}
     & \makecell[c]{\scriptsize\textbf{3.36\%}\\[-1pt]\scriptsize(0.06\%)}
     & \makecell[c]{\scriptsize 20,969.65\\[-1pt]\scriptsize(125.21)}
     & \makecell[c]{\scriptsize 10,770.30\\[-1pt]\scriptsize(45.27)}
     & \makecell[c]{\scriptsize\textbf{78.04\%}\\[-1pt]\scriptsize(0.33\%)}
     & \makecell[c]{\scriptsize 17,820.15\\[-1pt]\scriptsize(78.74)}
     & \makecell[c]{\scriptsize 8,955.40\\[-1pt]\scriptsize(40.52)}
     & \makecell[c]{\scriptsize\textbf{64.89\%}\\[-1pt]\scriptsize(0.29\%)} \\[6pt]

    20
     & \makecell[c]{\scriptsize 952.55\\[-1pt]\scriptsize(7.00)}
     & \makecell[c]{\scriptsize 919.50\\[-1pt]\scriptsize(6.76)}
     & \makecell[c]{\scriptsize\textbf{6.66\%}\\[-1pt]\scriptsize(0.05\%)}
     & \makecell[c]{\scriptsize 42,322.85\\[-1pt]\scriptsize(349.68)}
     & \makecell[c]{\scriptsize 13,154.65\\[-1pt]\scriptsize(36.61)}
     & \makecell[c]{\scriptsize\textbf{95.32\%}\\[-1pt]\scriptsize(0.27\%)}
     & \makecell[c]{\scriptsize 35,707.10\\[-1pt]\scriptsize(206.52)}
     & \makecell[c]{\scriptsize 11,706.40\\[-1pt]\scriptsize(40.39)}
     & \makecell[c]{\scriptsize\textbf{84.82\%}\\[-1pt]\scriptsize(0.29\%)} \\[6pt]

    40
     & \makecell[c]{\scriptsize 1,898.50\\[-1pt]\scriptsize(11.82)}
     & \makecell[c]{\scriptsize 1,770.30\\[-1pt]\scriptsize(13.14)}
     & \makecell[c]{\scriptsize\textbf{12.93\%}\\[-1pt]\scriptsize(0.10\%)}
     & \makecell[c]{\scriptsize 84,394.40\\[-1pt]\scriptsize(610.73)}
     & \makecell[c]{\scriptsize 13,770.60\\[-1pt]\scriptsize(7.56)}
     & \makecell[c]{\scriptsize\textbf{99.78\%}\\[-1pt]\scriptsize(0.05\%)}
     & \makecell[c]{\scriptsize 71,369.40\\[-1pt]\scriptsize(645.63)}
     & \makecell[c]{\scriptsize 13,315.15\\[-1pt]\scriptsize(24.98)}
     & \makecell[c]{\scriptsize\textbf{96.48\%}\\[-1pt]\scriptsize(0.18\%)} \\
    %
        \bottomrule
    \end{tabular}
\end{table}

\section{Concluding Remarks}\label{sec:conclsion}

In this paper, we developed a unified sequential method for generating graphs with prescribed degree sequences for bipartite, directed, and undirected networks. Within this method, we established a necessary and sufficient interval condition for bipartite graph generation that characterizes the admissible number of connections at each
step, thereby guaranteeing global feasibility. Building on this method, we designed enumeration and sampling algorithms for bipartite graphs: the former exhaustively generates all feasible configurations, whereas the latter employs uniform sampling or efficient, albeit slightly biased, sampling. By incorporating additional self-loop and single-degree connection constraints, as well as feasibility verification and symmetric connection steps, the bipartite graph algorithms are extended to the directed and undirected cases while preserving the same algorithmic principles. Numerical experiments demonstrate the effectiveness of the proposed algorithms in terms of enumeration runtime, as well as sampling speed, uniformity, and graph space coverage ratio.

\section*{Acknowledgments}
A preliminary version of this work was published in the 2025 {\em Proceedings of the Winter Simulation Conference} \citep{Sun2025}.

\bibliographystyle{informs2014}
\bibliography{GGM}

\newpage
\ECSwitch
\ECHead{Electronic Companion to ``Graph Simulation Methods under Partial Information"}
\setcounter{equation}{0}
\setcounter{lemma}{0}
\setcounter{algorithm}{0}
\setcounter{subsection}{0}
\setcounter{assumption}{0}
\setcounter{theorem}{1}
\setcounter{proposition}{0}
\renewcommand{\theequation}{A.\arabic{equation}}
\renewcommand{\thelemma}{A.\arabic{lemma}}
\renewcommand{\thesubsection}{\thesection.\arabic{subsection}}
\renewcommand{\thealgorithm}{A.\arabic{algorithm}}
\renewcommand{\theassumption}{A.\arabic{assumption}}
\renewcommand{\thetheorem}{A.\arabic{theorem}}
\renewcommand{\theproposition}{A.\arabic{proposition}}


\section{Proof of Lemma \ref{lemma:necessary and sufficient condition}}\label{appx:L1-proof}
Before proving Lemma \ref{lemma:necessary and sufficient condition}, we need the following lemma.
\begin{lemma}[Gale--Ryser Theorem \citep{gale1957theorem}]\label{lemma:gale-ryser}
    For the prescribed degree sequences $\mathbf{a}$ and $\mathbf{b}$, a bipartite graph exists if and only if $\sum_{i=1}^{m} \left [ \mathbf{a} \right ]_i = \sum_{j=1}^{n} \left [ \mathbf{b} \right ]_j$, and the following condition holds for all $r\in \left\{1, 2, \dots, n\right\}$:
    \begin{equation*}
        \sum_{i=1}^{m} \min \left\{\left [ \mathbf{a} \right ]_i,r\right\} \geq \sum_{j=1}^{r} \left [ \mathbf{b} \right ]_j.
    \end{equation*}
\end{lemma}

Next, we prove Lemma \ref{lemma:necessary and sufficient condition}.

\begin{proof}{Proof.}
By \eqref{eq:Z matrix} and \eqref{eq:z vector}, we can get that $\min \left\{\left [ \mathbf{a} \right ]_i,r\right\} = \sum_{j=1}^{r} \left[\mathscr{Z} \!\left(\mathbf{\mathbf{a}}, |\mathbf{b}|\right)\right]_{ij}$ and $\left[\mathbf{z}\right]_j=\sum_{i=1}^{m} \left[\mathscr{Z} \!\left(\mathbf{\mathbf{a}}, |\mathbf{b}|\right)\right]_{ij}$. Then, we have
\begin{equation*}
    \sum_{i=1}^{m} \min \left\{\left [ \mathbf{a} \right ]_i,r\right\} = \sum_{i=1}^{m} \sum_{j=1}^{r} \left[\mathscr{Z} \!\left(\mathbf{\mathbf{a}}, |\mathbf{b}|\right)\right]_{ij} = \sum_{j=1}^{r} \sum_{i=1}^{m} \left[\mathscr{Z} \!\left(\mathbf{\mathbf{a}}, |\mathbf{b}|\right)\right]_{ij} = \sum_{j=1}^{r} \left[\mathbf{z}\right]_j.
\end{equation*}
Hence, we conclude the lemma by Lemma~\ref{lemma:gale-ryser}.\Halmos
\end{proof}

\section{Proof of Theorem \ref{thm:interval}}\label{appx:T1-proof}

Before proving the theorem, we introduce two greedy algorithms that connect the first node $\left[\mathbf{b}\right]_1$ to nodes in $\mathbf{a}$ and yield the updated degree sequences $\mathbf{a}^{[2]}$ and $\mathbf{b}^{[2]}$. These greedy algorithms differ from the bipartite graph generation mechanism introduced in Section~\ref{sec:BGA} and will serve as auxiliary constructions in the proof.

\subsection{Greedy Algorithms}

We begin with a standard greedy algorithm and then construct its $k$-step variant. The standard greedy algorithm is summarized in Algorithm~\ref{alg:greedy-algorithm}. Given degree sequences $\mathbf{a}$ and $\mathbf{b}$, the nodes in $\mathbf{a}$ are partitioned into $p$ groups according to their distinct degree values: group~$k$ consists of $m_k$ nodes with degree $\alpha_k$, where $\alpha_1 < \alpha_2 < \dots < \alpha_p$ and $\sum_{k=1}^{p} m_k = m$. The standard greedy algorithm connects the node $\left[\mathbf{b}\right]_1$ to nodes in $\mathbf{a}$ in a greedy backward manner. Specifically, it starts from the group with the largest degree (group~$p$), connects $\left[\mathbf{b}\right]_1$ to all  nodes in that group, and then proceeds to group~$p-1$, group~$p-2$, and so on, until the degree of $\left[\mathbf{b}\right]_1$ is exhausted. After establishing these connections, any nodes whose degrees are reduced to zero are removed, yielding the updated degree sequences $\mathbf{a}^{[2]}$ and $\mathbf{b}^{[2]}$. In contrast to the bipartite graph generation mechanism introduced in Section~\ref{sec:BGA}, which connects $\left[\mathbf{b}\right]_1$ by prioritizing lower-degree groups, the greedy algorithm adopts the opposite strategy by prioritizing higher-degree groups, fully exhausting each group before moving on. 

\begin{algorithm}[htpb]
\caption{Standard Greedy Algorithm}
\label{alg:greedy-algorithm}
{\renewcommand{\baselinestretch}{0.75}\selectfont
\begin{algorithmic}[1]
\Require Degree sequences $\mathbf{a} \in \mathbb{N}_{+}^{m}$ and $\mathbf{b} \in \mathbb{N}_{+}^{n}$
\Ensure Updated degree sequences $\mathbf{a}^{[2]}$ and $\mathbf{b}^{[2]}$

\State Group $\mathbf{a}$ into $p$ groups $(\alpha_{k}, m_{k})$ with $\alpha_1 < \alpha_2 < \cdots < \alpha_p$

\For{$k = p$ \textbf{down to} $1$}
\State Let $\mathscr{F}\!\left(k\right) \gets \min\big\{m_{k}, \left[\mathbf{b}\right]_1^{(k)}\big\}$, where $\left[\mathbf{b}\right]_{1}^{(k)}$ denotes the remaining degree of $\left[\mathbf{b}\right]_1$
\State Uniformly select $\mathscr{F} \!\left(k\right)$ distinct nodes in group $k$ and connect them to the node $\left[\mathbf{b}\right]_1$
\EndFor

\State Remove zero-degree nodes from the degree sequences

\State \textbf{return} the updated degree sequences $\mathbf{a}^{[2]}$ and $\mathbf{b}^{[2]}$
\end{algorithmic}
}
\end{algorithm}

Building on this idea, we introduce the $k$-step greedy algorithm, summarized in Algorithm~\ref{alg:k-step-algorithm}. As before, the nodes in $\mathbf{a}$ are partitioned into $p$ groups according to their degrees. The key distinction from the standard greedy algorithm lies in the first $k-1$ steps. In the first $k-1$ steps, the node $\left[\mathbf{b}\right]_1$ connects to groups~$1$ through $k-1$ in the same order as in the bipartite graph generation mechanism described in Section~\ref{sec:BGA}. Moreover, the number of nodes connected to $\left[\mathbf{b}\right]_1$ from each of these groups is specified as follows. Since a bipartite graph exists for the degree sequences $\mathbf{a}$ and $\mathbf{b}$, there exists at least one feasible way for $\left[\mathbf{b}\right]_1$ to connect to nodes in $\mathbf{a}$, which can be characterized by the number of connected nodes from each group. Accordingly, there exists a feasible collection $\{\widetilde{\mathscr{F}}(k')\}_{k'=1}^{p}$, where $\widetilde{\mathscr{F}}(k')$ denotes the number of connected nodes from the $m_{k'}$ nodes in group $k'$, such that the resulting updated degree sequences $\mathbf{a}^{[2]}$ and $\mathbf{b}^{[2]}$ remain feasible for bipartite graph generation. For the $k$-step greedy algorithm, $\widetilde{\mathscr{F}}\!\left(k'\right)$ nodes from each group $k' \in \left\{1,2,\dots,k-1\right\}$ are connected to $\left[\mathbf{b}\right]_1$ during the first $k-1$ steps. From step $k$ onward, the algorithm switches to the greedy backward strategy: the node $\left[\mathbf{b}\right]_1$ connects sequentially to the remaining groups in descending order, starting from group~$p$, then group~$p-1$, group~$p-2$, and so on, until its degree is exhausted. After establishing these connections, any nodes whose degrees are reduced to zero are removed, yielding the updated degree sequences $\mathbf{a}^{[2]}$ and $\mathbf{b}^{[2]}$. For the degree sequences $\mathbf{a}^{[2]}$ and $\mathbf{b}^{[2]}$ produced by the $k$-step greedy algorithm, we establish the following lemma, which guarantees the existence of a bipartite graph.

\begin{algorithm}[htpb]
\caption{$k$-step Greedy Algorithm}
\label{alg:k-step-algorithm}
{\renewcommand{\baselinestretch}{0.75}\selectfont
\begin{algorithmic}[1]
\Require Degree sequences $\mathbf{a} \in \mathbb{N}_{+}^{m}$ and $\mathbf{b} \in \mathbb{N}_{+}^{n}$
\Ensure Updated degree sequences $\mathbf{a}^{[2]}$ and $\mathbf{b}^{[2]}$

\State Group $\mathbf{a}$ into $p$ groups $(\alpha_{k'}, m_{k'})$ with $\alpha_1 < \alpha_2 < \cdots < \alpha_p$

\For{$k' = 1$ \textbf{to} $k-1$}
\State Let $\mathscr{F}\!\left(k'\right) \gets \widetilde{\mathscr{F}}\!\left(k'\right)$, where $\widetilde{\mathscr{F}}\!\left(k'\right)$ is a feasible value
\State Uniformly select $\mathscr{F} \!\left(k'\right)$ distinct nodes in group $k'$ and connect them to the node $\left[\mathbf{b}\right]_1$
\EndFor

\For{$k' = p$ \textbf{down to} $k$}
\State Let $\mathscr{F}\!\left(k'\right) \gets \min\big\{m_{k'}, \left[\mathbf{b}\right]_1^{(k')}\big\}$, where $\left[\mathbf{b}\right]_{1}^{(k')}$ denotes the remaining degree of $\left[\mathbf{b}\right]_1$
\State Uniformly select $\mathscr{F} \!\left(k'\right)$ distinct nodes in group $k'$ and connect them to the node $\left[\mathbf{b}\right]_1$
\EndFor

\State Remove zero-degree nodes from the degree sequences

\State \textbf{return} the updated degree sequences $\mathbf{a}^{[2]}$ and $\mathbf{b}^{[2]}$
\end{algorithmic}
}
\end{algorithm}

\begin{lemma}\label{lemma:k-step_greedy_algorithm}
    Suppose that a bipartite graph exists for the prescribed degree sequences $\mathbf{a}$ and $\mathbf{b}$. Then a bipartite graph also exists for the updated degree sequences $\mathbf{a}^{[2]}$ and $\mathbf{b}^{[2]}$ produced by the $k$-step greedy algorithm. 
\end{lemma}

\begin{proof}{Proof.}
Since a bipartite graph exists for the degree sequences $\mathbf{a}$ and $\mathbf{b}$, there exists a feasible collection $\{\widetilde{\mathscr{F}}(k')\}_{k'=1}^{p}$ such that $[\mathbf{b}]_1$ can connect to $\widetilde{\mathscr{F}}(k')$ nodes from group $k'$ of the grouped $\mathbf{a}$, thereby exactly exhausting the degree of $[\mathbf{b}]_1$. After these connections are made, the resulting updated degree sequences $\mathbf{a}^{[2]}$ and $\mathbf{b}^{[2]}$ remain feasible for bipartite graph generation. 

By Lemma \ref{lemma:necessary and sufficient condition}, as the degree sequences $\mathbf{a}^{[2]}$ and $\mathbf{b}^{[2]}$ induced by $\widetilde{\mathscr{F}}\!\left(k'\right)$ remain feasible for bipartite graph generation, the following inequality holds for all $r\in \left\{1,2,\dots,n-1\right\}$,
\begin{equation}\label{eq:feasible_Fk}
    \sum_{j=1}^r \left[\widetilde{\mathbf{z}}^{[2]}\right]_j \geq \sum_{j=1}^r \left[\mathbf{b}^{[2]}\right]_j,
\end{equation}
where $\widetilde{\mathbf{z}}^{[2]}$ denotes the vector $\mathbf{z}^{[2]}$ obtained from $\mathscr{Z} \!\left(\mathbf{a}^{[2]}, |\mathbf{b}^{[2]}|\right)$ by \eqref{eq:z vector}, with the updated degree sequences $\mathbf{a}^{[2]}$ and $\mathbf{b}^{[2]}$ induced by $\widetilde{\mathscr{F}}\!\left(k'\right)$.

Let $\widehat{\mathscr{F}}\!\left(k'\right)$ denote the number of connected nodes from the $m_{k'}$ nodes in group $k'$ by the $k$-step greedy algorithm. Without loss of generality, 
\begin{equation*}
    \widehat{\mathscr{F}}(k') = \widetilde{\mathscr{F}}(k') \quad \text{for all } k' \in \{1,2,\dots,k-1\}.
\end{equation*}
From step $k$ onward, the $k$-step greedy algorithm follows the greedy backward strategy, allocating the remaining degree of the node $\left[\mathbf{b}\right]_1$ sequentially to the groups $p,p-1,\dots,k$ in descending order of degree and fully exhausting each higher-degree group before moving on to a lower-degree one. Consequently, for any group index not smaller than $k$, the cumulative number of connected nodes from groups $k$ up to that group under the greedy backward strategy is no larger than that under any feasible allocation, i.e.,
\begin{equation*}
    \sum_{k'=k}^{l} \widehat{\mathscr{F}}\!\left(k'\right) \leq \sum_{k'=k}^{l} \widetilde{\mathscr{F}}\!\left(k'\right), \quad \text{for all } l \in \{k,k+1,\dots,p\}
\end{equation*}
Combining this inequality with the equality on the first $k-1$ groups yields
\begin{equation}\label{eq:greedy_Fk}
    \sum_{k'=1}^{l} \widehat{\mathscr{F}}\!\left(k'\right) \leq \sum_{k'=1}^{l} \widetilde{\mathscr{F}}\!\left(k'\right), \quad \text{for all } l \in \{1,2,\dots,p\}
\end{equation}

According to Remark~\ref{F(K)-note}, the quantity $\mathscr{F} \!\left(k'\right)$ also represents a reduction in the $\alpha_{k'}$-th component of $\mathbf{z}$ and hence determines the resulting vector $\mathbf{z}^{[2]}$. A smaller value of $\mathscr{F}(k')$ therefore corresponds to a larger $\alpha_{k'}$-th component of $\mathbf{z}^{[2]}$. It follows from \eqref{eq:greedy_Fk} that, for all $r \in \{1,2,\dots,n-1\}$,
\begin{equation}\label{eq:greedy_z}
    \sum_{j=1}^r \left[\widehat{\mathbf{z}}^{[2]}\right]_j \geq \sum_{j=1}^r \left[\widetilde{\mathbf{z}}^{[2]}\right]_j,
\end{equation}
where $\widehat{\mathbf{z}}^{[2]}$ denotes the vector $\mathbf{z}^{[2]}$ obtained from $\mathscr{Z} \!\left(\mathbf{a}^{[2]}, |\mathbf{b}^{[2]}|\right)$ by \eqref{eq:z vector}, with the updated degree sequences $\mathbf{a}^{[2]}$ and $\mathbf{b}^{[2]}$ produced by the $k$-step greedy algorithm.

Combining \eqref{eq:greedy_z} with \eqref{eq:feasible_Fk}, we obtain that, for all $r \in \{1,2,\dots,n-1\}$,
\begin{equation}\label{eq:lemma 3}
    \sum_{j=1}^r \left[\widehat{\mathbf{z}}^{[2]}\right]_j \geq \sum_{j=1}^r \left[\mathbf{b}^{[2]}\right]_j.
\end{equation}
By Lemma \ref{lemma:necessary and sufficient condition}, there exists a bipartite graph for the updated degree sequences $\mathbf{a}^{[2]}$ and $\mathbf{b}^{[2]}$ produced by the $k$-step greedy algorithm. Hence, we conclude the proof of the lemma.\Halmos
\end{proof}

To illustrate the $k$-step greedy algorithm and Lemma~\ref{lemma:k-step_greedy_algorithm}, we consider an example with $\mathbf{a} = (1,1,2,2,3)$ and $\mathbf{b} = (4,3,2)$. For the case of $2$-step greedy algorithm, we proceed as follows. First, the nodes in $\mathbf{a}$ are partitioned into three groups ($p=3$): group 1 consists of two nodes of degree 1 ($\alpha_1 = 1$, $m_1 = 2$), group 2 consists of two nodes of degree 2 ($\alpha_2 = 2$, $m_2 = 2$), and group 3 consists of one node of degree 3 ($\alpha_3 = 3$, $m_3 = 1$). Next, the node $\left[\mathbf{b}\right]_1$ of degree 4 takes a feasible value of $\mathscr{F}\!\left(1\right)$---for example $\mathscr{F}\!\left(1\right) = 1$---and sequentially connects all nodes in $\mathbf{a}$, proceeding from group 3 to group 2. Consequently, $\left[\mathbf{b}\right]_1$ connects to one node in group 1 (reducing its degree from 1 to 0), one node in group 3 (reducing its degree from 3 to 2), and two nodes in group 2 (reducing its degree from 2 to 1). After these connections, $\left[\mathbf{b}\right]_1$ and the one node of degree 0 in group 1 are removed, yielding the updated degree sequences $\mathbf{a}^{[2]} = (1,1,1,2)$ and $\mathbf{b}^{[2]} = (3, 2)$, for which a bipartite graph exists.

\subsection{Proof of Theorem \ref{thm:interval}}
We prove the theorem by mathematical induction on the group index $k$.

{\em Induction statement.} For each $k \in \{1,2,\dots,p\}$, define the statement $\mathcal{P}(k)$ as follows.

Suppose that a bipartite graph exists for the prescribed degree sequences $\mathbf{a}$ and $\mathbf{b}$,
and that after the node $[\mathbf{b}]_1$ has been connected to nodes in groups~$1$ through $k-1$ of $\mathbf{a}$, the resulting updated degree sequences remain feasible for bipartite graph generation.\footnote{The latter feasibility assumption is equivalent to the existence of a collection $\{\mathscr{F}(k')\}_{k'=k}^{p}$ such that the resulting updated degree sequences $\mathbf{a}^{[2]}$ and $\mathbf{b}^{[2]}$ remain feasible for bipartite graph generation.} Then, after the node $[\mathbf{b}]_1$ is connected to $\mathscr{F}(k)$ nodes in group~$k$, the resulting updated degree sequences remain feasible if and only if the number of connected nodes $\mathscr{F}(k)$ satisfies
\begin{equation*}
    \mathscr{F} \!\left(k\right) \in \left[\max \bigl\{ \bigl[\mathbf{b}\bigr]_1^{(k)} - M^{(k)}, 0 \bigr\},\, \min \bigl\{m_k, D_{\mathrm{left}}^{(k)} + D_{\mathrm{now}}^{(k)} + D_{\mathrm{right}}^{(k)} \bigr\}\right].
\end{equation*}

{\em Sketch of proof.}
We begin by establishing the base case $\mathcal{P}(1)$, whose proof proceeds in two steps. First, we derive an interval of $\mathscr{F}(1)$ for which there exists an integer collection $\{\mathscr{F}(k)\}_{k=2}^{p}$ such that the degree of the node $\left[\mathbf{b}\right]_1$ can be allocated across the $p$ groups of $\mathbf{a}$. However, $\mathscr{F}(1)$ in this interval does not guarantee that the resulting updated degree sequences remain feasible for bipartite graph generation. Second, using Lemma~\ref{lemma:necessary and sufficient condition} and the standard greedy algorithm, we characterize the subset of this interval for which there exists a feasible collection $\{\widetilde{\mathscr{F}}(k)\}_{k=2}^{p}$ such that the resulting updated degree sequences $\mathbf{a}^{[2]}$ and $\mathbf{b}^{[2]}$ remain feasible.

We then consider the induction step. Fix any $k\in \{2,3,\dots,p\}$ and assume that $\mathcal{P}(k')$ holds for all $k'\in \{1,2,\dots,k-1\}$. In particular, there exists a feasible collection $\{\widetilde{\mathscr{F}}(k')\}_{k'=1}^{k-1}$ satisfying the interval conditions prescribed by $\mathcal{P}(k')$ for all $k'\in \{1,2,\dots,k-1\}$, such that the resulting updated degree sequences remain feasible after the node $\left[\mathbf{b}\right]_1$ has been connected to nodes in groups~$1$ through $k-1$ of $\mathbf{a}$. Under this induction hypothesis, the proof of $\mathcal{P}(k)$ follows the same two-step structure as in the base case. We first derive an interval of $\mathscr{F}(k)$ for which there exists an integer collection $\{\mathscr{F}(k')\}_{k'=k+1}^{p}$ such that the remaining degree of the node $\left[\mathbf{b}\right]_1$ can be allocated across the remaining groups of $\mathbf{a}$, and then characterize the subset of this interval for which there exists a feasible collection $\{\widetilde{\mathscr{F}}(k')\}_{k'=k+1}^{p}$ such that the resulting updated degree sequences $\mathbf{a}^{[2]}$ and $\mathbf{b}^{[2]}$ remain feasible, using Lemma~\ref{lemma:necessary and sufficient condition} and the $k$-step greedy algorithm.

\begin{proof}{Proof.}
{\em Base case: $k=1$.} We first establish $\mathcal{P}(1)$ in the following lemma.
\begin{lemma}\label{lemma:k=1-interval}
    Suppose that a bipartite graph exists for the prescribed degree sequences $\mathbf{a}$ and $\mathbf{b}$. Then, after the node $\left[\mathbf{b}\right]_1$ is connected to $\mathscr{F}(1)$ nodes in group~$1$ of $\mathbf{a}$, the resulting updated degree sequences remain feasible for bipartite graph generation if and only if the number of connected nodes $\mathscr{F}(1)$ satisfies
    \begin{equation*}
        \mathscr{F} \!\left(1\right) \in \left[\max \bigl\{ \bigl[\mathbf{b}\bigr]_1^{(1)} - M^{(1)}, 0 \bigr\},\, \min \bigl\{m_1 , D_{\mathrm{left}}^{(1)} + D_{\mathrm{now}}^{(1)} + D_{\mathrm{right}}^{(1)} \bigr\}\right].
    \end{equation*}
\end{lemma}
\begin{proof}{Proof of Lemma~\ref{lemma:k=1-interval}.}
We first derive an interval of $\mathscr{F}(1)$ for which there exists an integer collection $\{\mathscr{F}(k)\}_{k=2}^{p}$ with $0 \leq \mathscr{F}(k) \leq m_k$ such that the degree of the node $\left[\mathbf{b}\right]_1$ can be allocated across the $p$ groups of $\mathbf{a}$, where $m_k$ denotes the number of nodes in group~$k$ of $\mathbf{a}$.

Recall that $\left[\mathbf{b}\right]_1^{(1)}$ denotes the initial degree of node $\left[\mathbf{b}\right]_1$. Since the number of connected nodes in group~$1$, $\mathscr{F} \!\left(1\right)$, cannot exceed either the number of nodes in that group or the degree of the node $\left[\mathbf{b}\right]_1$, we obtain
\begin{equation}\label{eq:0-min1}
    \mathscr{F}\!\left(1\right) \in \left[0,\, \min \bigl\{m_1,\left[\mathbf{b}\right]_1^{(1)}\bigr\}\right].
\end{equation}
Recall that $M^{(1)}$ denotes the total number of nodes in groups $2$ through $p$ of $\mathbf{a}$. If the degree $\left[\mathbf{b}\right]_1^{(1)}$ exceeds $M^{(1)}$, then the node $\left[\mathbf{b}\right]_1$ must connect to at least
$\left[\mathbf{b}\right]_1^{(1)}-M^{(1)}$ nodes from group~$1$ so that its remaining degree can be accommodated by the remaining groups, which implies that $\mathscr{F} \!\left(1\right) \geq \max\bigl\{\left[\mathbf{b}\right]_{1}^{(1)} - M^{(1)}, 0\bigr\}$. Combining this inequality with \eqref{eq:0-min1}, we obtain
\begin{equation}\label{eq:F1_initial_value}
    \mathscr{F} \!\left(1\right) \in \left[\max\bigl\{\left[\mathbf{b}\right]_{1}^{(1)} - M^{(1)}, 0\bigr\},\, \min \bigl\{m_1,\left[\mathbf{b}\right]_1^{(1)}\bigr\}\right].
\end{equation}

Note that for any $\mathscr{F}(1)$ satisfying \eqref{eq:F1_initial_value}, there exists an integer collection $\{\mathscr{F}(k)\}_{k=2}^{p}$ with $0 \leq \mathscr{F}(k) \leq m_k$ such that
\begin{equation}\label{eq:1-collection}
    \mathscr{F}(1) + \sum_{k=2}^{p} \mathscr{F}(k) = \left[\mathbf{b}\right]_1^{(1)},
\end{equation}
that is, the degree of the node $\left[\mathbf{b}\right]_1$ can be allocated across all $p$ groups of $\mathbf{a}$. However, \eqref{eq:F1_initial_value} alone does not guarantee that the updated degree sequences $\mathbf{a}^{[2]}$ and $\mathbf{b}^{[2]}$, induced by a $\mathscr{F}(1)$ satisfying it together with a collection $\{\mathscr{F}(k)\}_{k=2}^{p}$ satisfying \eqref{eq:1-collection}, remain feasible for bipartite graph generation. Indeed, for some values of $\mathscr{F}(1)$ satisfying \eqref{eq:F1_initial_value}, there exists a feasible collection $\{\widetilde{\mathscr{F}}(k)\}_{k=2}^{p}$ such that the resulting updated degree sequences $\mathbf{a}^{[2]}$ and $\mathbf{b}^{[2]}$ remain feasible, whereas for other values no such collection exists.

We then characterize the subset of the interval in \eqref{eq:F1_initial_value} for which there exists a feasible collection $\{\widetilde{\mathscr{F}}(k)\}_{k=2}^{p}$ such that the resulting updated degree sequences $\mathbf{a}^{[2]}$ and $\mathbf{b}^{[2]}$ remain feasible, using Lemma~\ref{lemma:necessary and sufficient condition} and the standard greedy algorithm.

According to Lemma~\ref{lemma:necessary and sufficient condition}, for the updated degree sequences $\mathbf{a}^{[2]}$ and $\mathbf{b}^{[2]}$ induced by a given $\mathscr{F}(1)$ satisfying \eqref{eq:F1_initial_value} together with an integer collection satisfying \eqref{eq:1-collection}, to remain feasible for bipartite graph generation, it is necessary and sufficient that $\sum_{i=1}^{m^{[2]}} \left[\mathbf{a}^{[2]}\right]_i = \sum_{j=1}^{n-1} \left[\mathbf{b}^{[2]}\right]_j$ and
\begin{equation}\label{eq:2latter-condition-lemma1}
    \sum_{j=1}^{r} [\mathbf{z}^{[2]}]_j \geq \sum_{j=1}^{r} [\mathbf{b}^{[2]}]_j \quad \text{for all } r\in \left\{1,2,\dots,n-1\right\},
\end{equation}
where $m^{[2]}$ denotes the number of remaining nodes in $\mathbf{a}^{[2]}$ after removing zero-degree nodes, and the vector $\mathbf{z}^{[2]}$ is obtained from $\mathscr{Z}\!\left(\mathbf{a}^{[2]},|\mathbf{b}^{[2]}|\right)$ by \eqref{eq:z vector}, with the updated degree sequence $\mathbf{a}^{[2]}$ and $\mathbf{b}^{[2]}$ induced by the given $\mathscr{F}(1)$ and the collection $\{\mathscr{F}(k)\}_{k=2}^{p}$. The equality condition of total degrees, $\sum_{i=1}^{m^{[2]}} \left[\mathbf{a}^{[2]}\right]_i = \sum_{j=1}^{n-1} \left[\mathbf{b}^{[2]}\right]_j$, holds automatically, since the initial degree sequences satisfy $\sum_{i=1}^{m} [\mathbf{a}]_i = \sum_{j=1}^{n} [\mathbf{b}]_j$ and each connection of the node $\left[\mathbf{b}\right]_1$ decreases the total degrees of $\mathbf{a}$ and $\mathbf{b}$ by one simultaneously. Hence, it suffices to focus on the cumulative inequality condition \eqref{eq:2latter-condition-lemma1}.

Based on \eqref{eq:greedy_z} in the proof of Lemma~\ref{lemma:k-step_greedy_algorithm}, it follows that for any given $\mathscr{F}(1)$ satisfying \eqref{eq:F1_initial_value}, among all collections $\{\mathscr{F}(k)\}_{k=2}^{p}$ satisfying \eqref{eq:1-collection}, the collection generated under the greedy backward strategy---namely, allocating the remaining degree of the node $\left[\mathbf{b}\right]_1$ sequentially to the groups $p,p-1,\dots,2$ in descending order of degree and fully exhausting each higher-degree group before moving on to a lower-degree one---maximizes the sum $\sum_{j=1}^{r}\left[\mathbf{z}^{[2]}\right]_j$ for all $r\in \{1,2,\dots,n-1\}$. Since condition \eqref{eq:2latter-condition-lemma1} has the same right-hand side for all such collections and is monotone in these sums, it follows that, for the given $\mathscr{F}(1)$, there exists a feasible collection $\{\widetilde{\mathscr{F}}(k)\}_{k=2}^{p}$ such that the resulting $\mathbf{z}^{[2]}$ satisfies \eqref{eq:2latter-condition-lemma1} if and only if the vector $\mathbf{z}^{[2]}$ obtained under the greedy backward strategy satisfies \eqref{eq:2latter-condition-lemma1}. Consequently, the feasible values of $\mathscr{F}(1)$ satisfying \eqref{eq:F1_initial_value} for which such a feasible collection exists are exactly those values for which the corresponding vector $\mathbf{z}^{[2]}$ under the greedy backward strategy satisfies \eqref{eq:2latter-condition-lemma1}.

Therefore, we can derive the feasible values of $\mathscr{F}(1)$ satisfying \eqref{eq:F1_initial_value} by analyzing condition \eqref{eq:2latter-condition-lemma1} for the vector $\mathbf{z}^{[2]}$, where $\mathbf{z}^{[2]}$ is obtained from $\mathscr{Z}\!\left(\mathbf{a}^{[2]},|\mathbf{b}^{[2]}|\right)$ by \eqref{eq:z vector}, with the updated degree sequences $\mathbf{a}^{[2]}$ and $\mathbf{b}^{[2]}$ induced by $\mathscr{F}(1)$ together with the collection $\{\mathscr{F}(k)\}_{k=2}^{p}$ generated under the greedy backward strategy.

By the definition of the $\mathscr{D}$-operation (Definition~\ref{def:D-vector}), the condition \eqref{eq:2latter-condition-lemma1} is equivalent to
\begin{equation}\label{eq:D n&s}
    \sum_{h=1}^{r} \left[\mathscr{D} \!\left(\mathbf{z}^{[2]}, \mathbf{b}^{[2]}\right)\right]_h \geq 0 \quad \text{for all } r\in \left\{1,2,\dots,n-1\right\}.
\end{equation}
To facilitate the derivation of the feasible values of $\mathscr{F}(1)$, we define the following three $\mathscr{D}$-vectors:
\begin{itemize}
    \item $D_1 = \mathscr{D}\!\left(\overline{\mathbf{z}}^{[2]}, \mathbf{b}^{[2]}\right)\in \mathbb{Z}^{n-1}$. Here, the vector $\overline{\mathbf{z}}^{[2]}$ is obtained from $\mathscr{Z}\!\left(\overline{\mathbf{a}}^{[2]},|\mathbf{b}^{[2]}|\right)$ by \eqref{eq:z vector}, with the updated degree sequence $\overline{\mathbf{a}}^{[2]}$ and $\mathbf{b}^{[2]}$ induced by a $\overline{\mathscr{F}}(1)$ satisfying \eqref{eq:F1_initial_value} together with the collection $\{\overline{\mathscr{F}}(k)\}_{k=2}^{p}$ generated under the greedy backward strategy.

    \item $D_2 = \mathscr{D}\!\left(\mathbf{z}^{(1)}, \mathbf{b}^{[2]}\right)\in \mathbb{Z}^{n}$. Here, the vector $\mathbf{z}^{(1)} = \mathbf{z}$ is obtained from $\mathscr{Z}\!\left(\mathbf{a},|\mathbf{b}|\right)$ by \eqref{eq:z vector}, with the initial degree sequences $\mathbf{a}$ and $\mathbf{b}$. The vector $\mathbf{z}^{(1)}$ corresponds to the initial state in which the node $\left[\mathbf{b}\right]_1$ has not yet been connected to any node in group~$1$ of $\mathbf{a}$. This $\mathscr{D}$-vector $D_2$ is used to define the following quantities $D_{\mathrm{left}}^{(1)}$, $D_{\mathrm{now}}^{(1)}$, and $D_{\mathrm{right}}^{(1)}$ according to \eqref{eq:definition of left and now} and \eqref{eq:definition of right}, all of which can be determined in the initial state and are used to characterize the feasible values of $\mathscr{F}(1)$:
    \begin{equation}\label{eq:D2 left and now}
        D_{\mathrm{left}}^{(1)} = \sum_{h=1}^{{\alpha_1}-1} \left[D_2\right]_h, \quad
        D_{\mathrm{now}}^{(1)} = \left[D_2\right]_{\alpha_1},
    \end{equation}
    and
    \begin{equation}\label{eq:D2 right}
        D_{\mathrm{right}}^{(1)} = \min\left\{\min_{r\in \{\alpha_1+1,\, \alpha_1+2,\, \dots,\, n\}} \sum_{h=\alpha_1 + 1}^{r} \left[D_2\right]_h, 0\right\},
    \end{equation}
    where $\alpha_1$ denotes the degree of the nodes in group~$1$ of $\mathbf{a}$.

    \item $D_3 = \mathscr{D}\!\left(\widehat{\mathbf{z}}^{[2]}, \mathbf{b}^{[2]}\right) \in \mathbb{Z}^{n-1}$. Here, the vector $\widehat{\mathbf{z}}^{[2]}$ is obtained from $\mathscr{Z}\!\left(\widehat{\mathbf{a}}^{[2]},|\mathbf{b}^{[2]}|\right)$ by \eqref{eq:z vector}, with the updated degree sequence $\widehat{\mathbf{a}}^{[2]}$ and $\mathbf{b}^{[2]}$ induced by the collection $\{\widehat{\mathscr{F}}(k)\}_{k=1}^{p}$ generated under the standard greedy algorithm, namely allocating the degree of the node $\left[\mathbf{b}\right]_1$ sequentially to the groups $p,p-1,\dots,1$ in descending order of degree and fully exhausting each higher-degree group before moving on to a lower-degree one. According to Lemma~\ref{lemma:k-step_greedy_algorithm}, the updated degree sequences $\widehat{\mathbf{a}}^{[2]}$ and $\mathbf{b}^{[2]}$ produced by the standard greedy algorithm (i.e., the $1$-step greedy algorithm), remain feasible for bipartite graph generation. Since condition \eqref{eq:D n&s} is necessary for the updated degree sequences to remain feasible for bipartite graph generation, we have that
    \begin{equation}\label{eq:1D3_0}
        \sum_{h=1}^{r} \left[D_3\right]_h = \sum_{h=1}^{r} \left[\mathscr{D} \!\left(\widehat{\mathbf{z}}^{[2]}, \mathbf{b}^{[2]}\right)\right]_h \geq 0 \quad \text{for all } r\in \left\{1,2,\dots,n-1\right\}.
    \end{equation}
\end{itemize}

The component-wise relations among the three $\mathscr{D}$-vectors are established as follows. According to Remark~\ref{F(K)-note}, the quantity $\mathscr{F} \!\left(k\right)$ also represents a reduction in the $\alpha_{k}$-th component of $\mathbf{z}$. This reduction determines the $\alpha_{k}$-th component of the updated vector $\mathbf{z}^{[2]}$, which in turn determines the $\alpha_{k}$-th component of the corresponding $\mathscr{D}$-vector $\mathscr{D}\left(\mathbf{z}^{[2]}, \mathbf{b}^{[2]}\right)$. Since the reductions induced by $\{\mathscr{F}(k)\}_{k=1}^{p}$ begin at the $\alpha_1$-th component, the first $\alpha_1-1$ components remain unaffected. Therefore, for the three $\mathscr{D}$-vectors $D_1$, $D_2$, and $D_3$, their first $\alpha_1-1$ components coincide, and we obtain
\begin{equation}\label{eq:1equal D_left}
    \sum_{h=1}^{r} \left[D_1\right]_h = \sum_{h=1}^{r} \left[D_2\right]_h = \sum_{h=1}^{r} \left[D_3\right]_h \quad \text{for all } r\in \left\{1, 2, \dots, \alpha_1-1\right\}.
\end{equation}
Consider the $\alpha_1$-th component of the $\mathscr{D}$-vectors $D_1$ and $D_2$. For $D_1$, we note that the quantity $\mathscr{F}(1) = \overline{\mathscr{F}}(1)$, which induces a reduction of $\overline{\mathscr{F}}(1)$ in the $\alpha_1$-th component, whereas $\mathscr{F}(1)=0$ for $D_2$. Therefore, relative to the $\alpha_1$-th component of $D_2$, that of $D_1$ is reduced by $\overline{\mathscr{F}}(1)$. Combining this with \eqref{eq:D2 left and now}, we obtain
\begin{equation}\label{eq:1equal D_now12}
    \left[D_1\right]_{\alpha_1} = \left[D_2\right]_{\alpha_1} - \overline{\mathscr{F}}(1) = D_{\mathrm{now}}^{(1)} - \overline{\mathscr{F}}(1).
\end{equation}
Consider the remaining components (excluding the first $\alpha_1$ components) of the $\mathscr{D}$-vectors $D_1$ and $D_2$. For $D_1$, the quantity $\overline{\mathscr{F}}(k)$ for $k \in \{2, 3, \dots, p\}$, generated under the greedy backward strategy, induces a reduction of $\overline{\mathscr{F}}(k)$ in the corresponding $\alpha_k$-th component, whereas $\mathscr{F}(k) = 0$ for $D_2$. Therefore, relative to the $\alpha_k$-th component of $D_2$, that of $D_1$ is reduced by $\overline{\mathscr{F}}(k)$ for each $k \in \{2, 3, \dots, p\}$. As a result, for all $h \in \{\alpha_1+1, \alpha_1+2, \dots, n-1\}$, the $h$-th component of $D_1$ is less than or equal to the corresponding component of $D_2$. Consequently, we obtain
\begin{equation}\label{eq:1equal D_right12}
    \sum_{h=\alpha_1 + 1}^{r} \left[D_1\right]_h \leq \sum_{h=\alpha_1 + 1}^{r} \left[D_2\right]_h \quad \text{for all } r\in \left\{\alpha_1+1,\alpha_1+2,\dots,n-1\right\}.
\end{equation}

Based on the three $\mathscr{D}$-vectors defined above and their component-wise relations established in \eqref{eq:1equal D_left}--\eqref{eq:1equal D_right12}, we now derive the feasible values of $\overline{\mathscr{F}}(1)$ satisfying \eqref{eq:F1_initial_value} by analyzing the condition~\eqref{eq:D n&s} for the $\mathscr{D}$-vector $D_1$:
\begin{equation}\label{eq:1D1-ns}
    \sum_{h=1}^{r} \left[D_1\right]_h \geq 0 \quad \text{for all } r\in \left\{1,2,\dots,n-1\right\}.
\end{equation}

Combining \eqref{eq:1D3_0} with \eqref{eq:1equal D_left}, we obtain
\begin{equation*}
    \sum_{h=1}^{r} \left[D_1\right]_h = \sum_{h=1}^{r} \left[D_3\right]_h \geq 0 \quad \text{for all } r\in \left\{1,2,\dots,\alpha_1-1\right\}.
\end{equation*}
Thus, it suffices to focus on the condition \eqref{eq:1D1-ns} for all $r\in \{\alpha_1, \alpha_1+1, \dots, n-1\}$, i.e.,
\begin{equation*}
    \sum_{h=1}^{r} \left[D_1\right]_h \geq 0 \quad \text{for all } r\in \left\{\alpha_1,\alpha_1+1,\dots,n-1\right\}.
\end{equation*}
We divide the above condition into two cases: $r=\alpha_1$ and $r\in\{\alpha_1+1,\alpha_1+2,\dots,n-1\}$, i.e.,
\begin{equation}\label{eq:1two-D1-initial}
    \begin{cases}
        \sum_{h=1}^{\alpha_1-1}\left[D_1\right]_h + \left[D_1\right]_{\alpha_1} \geq 0, & \text{for } r = \alpha_1, \\
        \sum_{h=1}^{\alpha_1-1}\left[D_1\right]_h + \left[D_1\right]_{\alpha_1} + \sum_{h=\alpha_1+1}^{r} \left[D_1\right]_h \geq 0, & \text{for all } r \in \{\alpha_1+1, \alpha_1+2, \dots, n-1\}.
    \end{cases}
\end{equation}
By \eqref{eq:D2 left and now} and \eqref{eq:1equal D_left}, we obtain that $\sum_{h=1}^{\alpha_1-1}\left[D_1\right]_h = \sum_{h=1}^{\alpha_1-1}\left[D_2\right]_h = D_{\mathrm{left}}^{(1)}$. Combining this with \eqref{eq:1equal D_now12}, the condition~\eqref{eq:1two-D1-initial} is equivalent to
\begin{equation}\label{eq:1two-D1}
    \begin{cases}
        D_{\mathrm{left}}^{(1)} + D_{\mathrm{now}}^{(1)} - \overline{\mathscr{F}}\!\left(1\right) \geq 0, & \text{for } r = \alpha_1, \\
        D_{\mathrm{left}}^{(1)} + D_{\mathrm{now}}^{(1)} - \overline{\mathscr{F}}\!\left(1\right) + \sum_{h=\alpha_1+1}^{r} \left[D_1\right]_h \geq 0, & \text{for all } r \in \{\alpha_1+1, \alpha_1+2, \dots, n-1\}.
    \end{cases}
\end{equation}

To further analyze condition~\eqref{eq:1two-D1}, we define $\mu$ as follows. Let $\mu$ be the smallest index at which $\sum_{h=\alpha_1+1}^{r} \left[D_1\right]_h$ attains its minimum over $r\in \{\alpha_1+1, \alpha_1+2, \dots, n-1\}$, i.e.,
\begin{equation*}
    \mu = \min \left\{ r \, \big| \, r \in \mathop{\arg\min}\limits_{r} \sum_{h=\alpha_1+1}^{r} \left[D_1\right]_h \right\}, 
    \quad \text{for all } r \in \left\{\alpha_1+1, \alpha_1+2, \dots, n-1 \right\}.
\end{equation*}
If $\sum_{h=\alpha_1+1}^{\mu} \left[D_1\right]_h > 0$, then $\sum_{h=\alpha_1+1}^{r} \left[D_1\right]_h > 0$ for all $r \in \{\alpha_1+1,\alpha_1+2,\dots,n-1\}$. In this case, the condition \eqref{eq:1two-D1} reduces to
\begin{equation*}
    D_{\mathrm{left}}^{(1)} + D_{\mathrm{now}}^{(1)} - \overline{\mathscr{F}}(1) \ge 0.
\end{equation*}
If, on the other hand, $\sum_{h=\alpha_1+1}^{\mu} \left[D_1\right]_h \le 0$, then the most restrictive inequality in \eqref{eq:1two-D1} is attained at $r=\mu$. In this case, the condition \eqref{eq:1two-D1} reduces to
\begin{equation*}
    D_{\mathrm{left}}^{(1)} + D_{\mathrm{now}}^{(1)} - \overline{\mathscr{F}}(1)
+ \sum_{h=\alpha_1+1}^{\mu} \left[D_1\right]_h \ge 0.
\end{equation*}
Therefore, the condition \eqref{eq:1two-D1} is equivalent to
\begin{equation}\label{eq:1two-D1-mu}
    \begin{cases}
    D_{\mathrm{left}}^{(1)} + D_{\mathrm{now}}^{(1)} - \overline{\mathscr{F}}(1) \ge 0,
    & \text{if } \sum_{h=\alpha_1+1}^{\mu} [D_1]_h > 0, \\[0.5ex]
    D_{\mathrm{left}}^{(1)} + D_{\mathrm{now}}^{(1)} - \overline{\mathscr{F}}(1)
    + \sum_{h=\alpha_1+1}^{\mu} [D_1]_h \ge 0,
    & \text{if } \sum_{h=\alpha_1+1}^{\mu} [D_1]_h \le 0.
    \end{cases}
\end{equation}

We now derive the feasible values of $\overline{\mathscr{F}}(1)$ by analyzing the two cases in \eqref{eq:1two-D1-mu}. To this end, we define $\nu$ as the smallest index at which $\sum_{h=\alpha_1+1}^{r} \left[D_2\right]_h$ attains its minimum over $r\in \{\alpha_1+1, \alpha_1+2, \dots, n-1\}$, i.e.,
\begin{equation*}
    \nu = \min \left\{ r \, \big| \, r \in \mathop{\arg\min}\limits_{r} \sum_{h=\alpha_1+1}^{r} \left[D_2\right]_h \right\}, 
    \quad \text{for all } r \in \left\{\alpha_1+1, \alpha_1+2, \dots, n-1 \right\}.
\end{equation*}
\begin{itemize}
    \item $\sum_{h=\alpha_1+1}^{\mu} \left[D_1\right]_h > 0$.

    By $\eqref{eq:1equal D_right12}$, we have that $\sum_{h=\alpha_1+1}^{\nu}\left[D_2\right]_h \ge \sum_{h=\alpha_1+1}^{\nu}\left[D_1\right]_h$. Combining this with the fact that $\sum_{h=\alpha_1+1}^{\mu} \left[D_1\right]_h > 0$, we obtain
    \begin{equation*}
        \sum_{h=\alpha_1+1}^{\nu} \left[D_2\right]_h \geq \sum_{h=\alpha_1+1}^{\nu} \left[D_1\right]_h \geq \sum_{h=\alpha_1+1}^{\mu} \left[D_1\right]_h > 0.
    \end{equation*}
    Based on the definition of $\nu$ and the fact that $\left[D_2\right]_n = \left[\mathscr{D}(\mathbf{z}^{(1)},\mathbf{b}^{[2]})\right]_n = \left[\mathbf{z}^{(1)}\right]_n \geq 0$, we have that $\min_{r} \sum_{h=\alpha_1 + 1}^{r} \left[D_2\right]_h = \sum_{h=\alpha_1 + 1}^{\nu} \left[D_2\right]_h$. Combining this with \eqref{eq:D2 right} and the fact that $\sum_{h=\alpha_1+1}^{\nu} \left[D_2\right]_h > 0$, we obtain
    \begin{equation*}
        D_{\mathrm{right}}^{(1)} = \min\left\{\min_{r} \sum_{h=\alpha_1 + 1}^{r} \left[D_2\right]_h, 0\right\} = \min\left\{\sum_{h=\alpha_1 + 1}^{\nu} \left[D_2\right]_h, 0\right\} = 0.
    \end{equation*}
    Therefore, the condition \eqref{eq:1two-D1-mu} in this case is equivalent to
    \begin{equation*}
        \overline{\mathscr{F}}(1) \leq D_{\mathrm{left}}^{(1)} + D_{\mathrm{now}}^{(1)} = D_{\mathrm{left}}^{(1)} + D_{\mathrm{now}}^{(1)} + D_{\mathrm{right}}^{(1)}.
    \end{equation*}

    \item $\sum_{h=\alpha_1+1}^{\mu} \left[D_1\right]_h \leq 0$.

    By \eqref{eq:1equal D_right12} and the fact that $\sum_{h=\alpha_1+1}^{\mu} \left[D_1\right]_h \leq \sum_{h=\alpha_1+1}^{\nu} \left[D_1\right]_h$, we obtain
    \begin{equation}\label{eq:1-mu-nu}
        \sum_{h=\alpha_1+1}^{\mu} \left[D_1\right]_h \leq \sum_{h=\alpha_1+1}^{\nu} \left[D_1\right]_h \leq \sum_{h=\alpha_1+1}^{\nu} \left[D_2\right]_h.
    \end{equation}
    By \eqref{eq:D2 right} and the fact that $\min_{r} \sum_{h=\alpha_1 + 1}^{r} \left[D_2\right]_h = \sum_{h=\alpha_1 + 1}^{\nu} \left[D_2\right]_h$, we obtain
    \begin{equation}\label{eq:1D-right}
        D_{\mathrm{right}}^{(1)} = \min\left\{\min_{r} \sum_{h=\alpha_1 + 1}^{r} \left[D_2\right]_h, 0\right\} = \min\left\{\sum_{h=\alpha_1 + 1}^{\nu} \left[D_2\right]_h, 0\right\}.
    \end{equation}
    Based on \eqref{eq:1-mu-nu}, \eqref{eq:1D-right}, and the fact that $\sum_{h=\alpha_1+1}^{\mu} \left[D_1\right]_h \leq 0$, we have that $\sum_{h=\alpha_1+1}^{\mu} \left[D_1\right]_h \leq D_{\mathrm{right}}^{(1)}$. Therefore, the condition \eqref{eq:1two-D1-mu} in this case is equivalent to
    \begin{equation*}
        \overline{\mathscr{F}}\!\left(1\right) \leq D_{\mathrm{left}}^{(1)} + D_{\mathrm{now}}^{(1)} + \sum_{h=\alpha_1+1}^{\mu} \left[D_1\right]_h \leq D_{\mathrm{left}}^{(1)} + D_{\mathrm{now}}^{(1)} + D_{\mathrm{right}}^{(1)}.
    \end{equation*}
\end{itemize}

Based on the analysis of the above two cases, we have that the condition \eqref{eq:1two-D1-mu} is equivalent to
\begin{equation*}
    \overline{\mathscr{F}}\!\left(1\right) \leq D_{\mathrm{left}}^{(1)} + D_{\mathrm{now}}^{(1)} + D_{\mathrm{right}}^{(1)}.
\end{equation*}
Combining this inequality with \eqref{eq:F1_initial_value}, we obtain
\begin{equation}\label{eq:F1_mid_value}
    \overline{\mathscr{F}}\!\left(1\right) \in \left[\max \bigl\{\left[\mathbf{b}\right]_1^{(1)} - M^{(1)},\, 0\bigr\}, \min \bigl\{m_1,\, \left[\mathbf{b}\right]_1^{(1)},\, D_{\mathrm{left}}^{(1)} + D_{\mathrm{now}}^{(1)} + D_{\mathrm{right}}^{(1)}\bigr\}\right].
\end{equation}

The following analysis shows that $D_{\mathrm{left}}^{(1)} + D_{\mathrm{now}}^{(1)} + D_{\mathrm{right}}^{(1)} \leq \left[\mathbf{b}\right]_1^{(1)}$, which simplifies the upper bound in \eqref{eq:F1_mid_value}. By definition of $D_{\mathrm{right}}^{(1)}$ in \eqref{eq:D2 right}, we have that
\begin{equation*}
    D_{\mathrm{right}}^{(1)} = \min\left\{\min_{r} \sum_{h=\alpha_1 + 1}^{r} \left[D_2\right]_h, 0\right\} \leq \sum_{h=\alpha_1+1}^{n} \left[D_2\right]_h.
\end{equation*}
Combining this with the definitions of $D_{\mathrm{left}}^{(1)}$ and $D_{\mathrm{now}}^{(1)}$ in \eqref{eq:D2 left and now}, we obtain
\begin{equation}\label{eq:1-DD}
    D_{\mathrm{left}}^{(1)} + D_{\mathrm{now}}^{(1)} + D_{\mathrm{right}}^{(1)} \leq \sum_{h=1}^{n} \left[D_2\right]_h.
\end{equation}
Recall that $D_2 = \mathscr{D}\!\left(\mathbf{z}^{(1)}, \mathbf{b}^{[2]}\right)$. Combining this with the $\mathscr{D}$-operation (Definition~\ref{def:D-vector}), we have
\begin{equation}\label{eq:1D2-b1}
    \sum_{h=1}^{n}\left[D_2\right]_h = \sum_{h=1}^{n}\mathscr{D}\!\left(\mathbf{z}^{(1)}, \mathbf{b}^{[2]}\right) = \sum_{h=1}^{n} \left[\mathbf{z}^{(1)}\right]_h - \sum_{h=1}^{n-1} \left[\mathbf{b}^{[2]}\right]_h = \sum_{h=1}^{m} \left[\mathbf{a}\right]_h - \sum_{h=1}^{n-1} \left[\mathbf{b}^{[2]}\right]_h = \left[\mathbf{b}\right]_1^{(1)}.
\end{equation}
Based on \eqref{eq:1-DD} and \eqref{eq:1D2-b1}, we obtain
\begin{equation*}
    D_{\mathrm{left}}^{(1)} + D_{\mathrm{now}}^{(1)} + D_{\mathrm{right}}^{(1)} \leq \sum_{h=1}^{n} \left[D_2\right]_h = \left[\mathbf{b}\right]_1^{(1)}.
\end{equation*}
Combining this inequality with \eqref{eq:F1_mid_value}, we obtain
\begin{equation*}
    \overline{\mathscr{F}}\!\left(1\right) \in \left[\max \bigl\{\left[\mathbf{b}\right]_1^{(1)} - M^{(1)},\, 0\bigr\}, \min \bigl\{m_1,\, D_{\mathrm{left}}^{(1)} + D_{\mathrm{now}}^{(1)} + D_{\mathrm{right}}^{(1)}\bigr\}\right].
\end{equation*}

Therefore, under the assumption that a bipartite graph exists for the prescribed degree sequences $\mathbf{a}$ and $\mathbf{b}$, after the node $\left[\mathbf{b}\right]_1$ is connected to $\overline{\mathscr{F}}(1)$ nodes in group~$1$ of $\mathbf{a}$, the resulting updated degree sequences remain feasible for bipartite graph generation if and only if the stated interval condition
holds. Hence, we conclude the proof of the lemma.\Halmos
\end{proof}

{\em Induction step: $k \in \{2,3,\dots,p\}$.} Fix any $k \in \{2,3,\dots,p\}$ and assume that $\mathcal{P}(k')$ holds for all
$k' \in \{1,2,\dots,k-1\}$. We establish $\mathcal{P}(k)$ in the following lemma.
\begin{lemma}\label{lemma:k-interval}
    Suppose that a bipartite graph exists for the prescribed degree sequences $\mathbf{a}$ and $\mathbf{b}$, and that after the node $\left[\mathbf{b}\right]_1$ has been connected to nodes in groups~$1$ through $k-1$ of $\mathbf{a}$, the resulting updated degree sequences remain feasible for bipartite graph generation. Then, after the node $\left[\mathbf{b}\right]_1$ is connected to $\mathscr{F}(k)$ nodes in group~$k$, the resulting updated degree sequences remain feasible if and only if the number of connected nodes $\mathscr{F}(k)$ satisfies
    \begin{equation*}
        \mathscr{F} \!\left(k\right) \in \left[\max \bigl\{ \bigl[\mathbf{b}\bigr]_1^{(k)} - M^{(k)},\, 0 \bigr\},\, \min \bigl\{m_k,\, D_{\mathrm{left}}^{(k)} + D_{\mathrm{now}}^{(k)} + D_{\mathrm{right}}^{(k)} \bigr\}\right].
    \end{equation*}
\end{lemma}
\begin{proof}{Proof of Lemma~\ref{lemma:k-interval}.}
Under the induction hypothesis, there exists a feasible collection $\{\widetilde{\mathscr{F}}(k')\}_{k'=1}^{k-1}$ satisfying the interval conditions prescribed by $\mathcal{P}(k')$ for all
$k' \in \{1,2,\dots,k-1\}$, such that the resulting updated degree sequences remain feasible after the node $\left[\mathbf{b}\right]_1$ has been connected to nodes in groups~$1$ through $k-1$ of $\mathbf{a}$. The proof of $\mathcal{P}(k)$ follows the same two-step structure as in the base case.

We first derive an interval of $\mathscr{F}(k)$ for which there exists an integer collection $\{\mathscr{F}(k')\}_{k'=k+1}^{p}$ with $0 \leq \mathscr{F}(k') \leq m_{k'}$ such that the remaining degree of the node $\left[\mathbf{b}\right]_1$ can be allocated across the remaining groups of $\mathbf{a}$. By the same analysis as in the base case, we obtain
\begin{equation}\label{eq:Fk_initial_value}
    \mathscr{F} \!\left(k\right) \in \left[\max\bigl\{\left[\mathbf{b}\right]_{1}^{(k)} - M^{(k)}, 0\bigr\},\, \min \bigl\{m_k,\left[\mathbf{b}\right]_1^{(k)}\bigr\}\right].
\end{equation}

We then characterize the subset of the interval in \eqref{eq:Fk_initial_value} for which there exists a feasible collection $\{\widetilde{\mathscr{F}}(k')\}_{k'=k+1}^{p}$ such that the resulting updated degree sequences $\mathbf{a}^{[2]}$ and $\mathbf{b}^{[2]}$ remain feasible. Applying Lemma~\ref{lemma:necessary and sufficient condition} together with the $k$-step greedy algorithm, and using the same analysis as in the base case on the updated degree sequences after processing groups~$1$ through $k-1$, we obtain
\begin{equation*}
    \mathscr{F}\!\left(k\right) \leq D_{\mathrm{left}}^{(k)} + D_{\mathrm{now}}^{(k)} + D_{\mathrm{right}}^{(k)}.
\end{equation*}
Combining this inequality with \eqref{eq:Fk_initial_value}, we obtain
\begin{equation}\label{eq:Fk_mid_value}
    \mathscr{F}\!\left(k\right) \in \left[\max \bigl\{\left[\mathbf{b}\right]_1^{(k)} - M^{(k)},\, 0\bigr\}, \min \bigl\{m_k,\, \left[\mathbf{b}\right]_1^{(k)},\, D_{\mathrm{left}}^{(k)} + D_{\mathrm{now}}^{(k)} + D_{\mathrm{right}}^{(k)}\bigr\}\right].
\end{equation}

By the same analysis as in the base case, we further have that $D_{\mathrm{left}}^{(k)} + D_{\mathrm{now}}^{(k)} + D_{\mathrm{right}}^{(k)} \leq \left[\mathbf{b}\right]_1^{(k)}$. Combining this inequality with \eqref{eq:Fk_mid_value}, we obtain
\begin{equation*}
    \mathscr{F}\!\left(k\right) \in \left[\max \bigl\{\left[\mathbf{b}\right]_1^{(k)} - M^{(k)},\, 0\bigr\}, \min \bigl\{m_k,\, D_{\mathrm{left}}^{(k)} + D_{\mathrm{now}}^{(k)} + D_{\mathrm{right}}^{(k)}\bigr\}\right].
\end{equation*}

Therefore, under the induction hypothesis, the resulting updated degree sequences remain feasible after $\left[\mathbf{b}\right]_1$ is connected to $\mathscr{F}(k)$ nodes in group~$k$ if and only if the stated interval condition
holds. Hence, we conclude the proof of the lemma.\Halmos
\end{proof}

{\em Conclusion.} By mathematical induction, $\mathcal{P}(k)$ holds for all $k \in \{1,2,\dots,p\}$. Therefore, after the node $\left[\mathbf{b}\right]_1$ has been connected to nodes in groups~$1$ through $p$ of $\mathbf{a}$, the resulting updated degree sequences $\mathbf{a}^{[2]}$ and $\mathbf{b}^{[2]}$ remain feasible for bipartite graph generation if and only if for each group $k$ of $\mathbf{a}$, the number of connected nodes $\mathscr{F}(k)$ satisfies
\begin{equation*}
    \mathscr{F} \!\left(k\right) \in \left[\max \bigl\{ \bigl[\mathbf{b}\bigr]_1^{(k)} - M^{(k)},\, 0 \bigr\},\, \min \bigl\{m_k,\, D_{\mathrm{left}}^{(k)} + D_{\mathrm{now}}^{(k)} + D_{\mathrm{right}}^{(k)} \bigr\}\right].
\end{equation*}
Hence, we conclude the proof of the theorem.\Halmos
\end{proof}

\section{Proof of Proposition \ref{prop:N(a,b,F(k))}}\label{appx:T2-proof}
\begin{proof}{Proof.}
\eqref{eq:N_K-N_Ki} follows from the fact that $\mathscr{F}\!\left(k\right) \in [\ell_k,u_k]$ and the definition of $\mathscr{F}_i\!\left(k\right)$. \eqref{eq:N(a,b,k)} follows from the fact that the process moves to group $k+1$ in $\mathbf{a}^{[j]}$ when the node $\left[\mathbf{b}\right]_j$ has completed its connections to the nodes in group~$k$ for $k<p$, and moves to group $1$ in $\mathbf{a}^{[j+1]}$ when $\left[\mathbf{b}\right]_j$ has completed its connections to the nodes in group~$p$. Hence, we conclude the proof of the proposition.\Halmos
\end{proof}

\section{Pseudocode of the EBGS algorithm}\label{appx:alg-EBGS}

The pseudocode of the EBGS algorithm is presented in Algorithm~\ref{alg:EBGS}.

\begin{algorithm}[htpb]
\caption{Efficient Bipartite Graph Sampling (EBGS) Algorithm}
\label{alg:EBGS}
{\renewcommand{\baselinestretch}{0.75}\selectfont
\begin{algorithmic}[1]
\Require Degree sequences $\mathbf{a} \in \mathbb{N}_{+}^{m}$ and $\mathbf{b} \in \mathbb{N}_{+}^{n}$, an initially empty bipartite graph $\mathbf{G}\!\left(\mathbf{a}, \mathbf{b}\right)$
\Ensure A bipartite graph $\mathbf{G}\!\left(\mathbf{a}, \mathbf{b}\right)$

\For{$j = 1$ \textbf{to} $n$}
  \State Group $\mathbf{a}^{[j]}$ into $p$ groups $(\alpha_k, m_k)$ with $\alpha_1 < \alpha_2 < \cdots < \alpha_p$

  \For{$k = 1$ \textbf{to} $p$}
    \State Compute the feasible condition $[\ell_k,u_k]$ via \eqref{eq:Fk}
    \State Obtain the approximate weight vector $\mathbf{w}^{\mathrm{EB}}$ via \eqref{eq:weight-for-efficient}
    \State Sample $\mathscr{F} \!\left(k\right) \in [\ell_k,u_k]$ according to $\mathbf{w}^{\mathrm{EB}}$
    \State Uniformly select $\mathscr{F} \!\left(k\right)$ distinct nodes in group $k$ and connect them to $\left[\mathbf{b}\right]_j$ in $\mathbf{G}\!\left(\mathbf{a}, \mathbf{b}\right)$
  \EndFor
\EndFor

\State \textbf{return} $\mathbf{G}\!\left(\mathbf{a}, \mathbf{b}\right)$
\end{algorithmic}
}
\end{algorithm}

\section{Importance Sampling for the EBGS algorithm}\label{appx:IS-EBGS}

We apply importance sampling to correct for the nonuniformity induced by the EBGS algorithm when the objective is to compute expectations with respect to the uniform distribution on $\mathcal{G}\!\left(\mathbf{a},\mathbf{b}\right)$, i.e.,
\begin{equation*}
    \theta \triangleq \mathbb{E}_{\mathbf{w}^{\mathrm{UB}}} \!\Big[f\big(\mathbf{G}\!\left(\mathbf{a},\mathbf{b}\right)\big)\Big],
\end{equation*}
where $f(\cdot)$ is a function of interest. Computing expectations with respect to the uniform distribution has many applications in practice. For instance, as discussed in \cite{Glasserman2016}, in the context of interbank networks one can evaluate various measures of systemic risk $f(\cdot)$ on each sampled network and then average across multiple draws to obtain the expected risk measures over $\mathcal{G}\!\left(\mathbf{a},\mathbf{b}\right)$.

Specifically, the importance weight associated with $\mathbf{G}\!\left(\mathbf{a},\mathbf{b}\right) \in \mathcal{G}\!\left(\mathbf{a},\mathbf{b}\right)$ is given by
\begin{equation*}
    w_{\mathbf{G}} = \frac{1}{\left|\mathcal{G}(\mathbf{a},\mathbf{b})\right| p_{\mathbf{G}}},
\end{equation*}
where $p_{\mathbf{G}}$ is the sampling probability of $\mathbf{G}\!\left(\mathbf{a},\mathbf{b}\right)$ under the EBGS algorithm, expressed as
\begin{equation*}
    p_{\mathbf{G}} = \prod_{j=1}^n \prod_{k=1}^{p} \frac{\left[\mathbf{w}^{\mathrm{EB}}\right]_i}{\left(\sum_{s=1}^{\mu_k-\ell_k+1}\left[\mathbf{w}^{\mathrm{EB}}\right]_s\right)\cdot\binom{m_k}{\mathscr{F}_i \!\left(k\right)}}.
\end{equation*}
This expression for $p_{\mathbf{G}}$ follows from the fact that the normalized weight represents the probability of selecting $\mathscr{F}_i \!\left(k\right)$ within the interval, while the factor $1/\binom{m_k}{\mathscr{F}_i \!\left(k\right)}$ accounts for the uniform selection of a configuration among the $\binom{m_k}{\mathscr{F}_i \!\left(k\right)}$ possible configurations. Consequently, to estimate $\theta$ using the EBGS algorithm, we generate independent samples $\mathbf{G}_1\!\left(\mathbf{a},\mathbf{b}\right),\ldots,\mathbf{G}_N\!\left(\mathbf{a},\mathbf{b}\right)$ and construct the estimator
\begin{equation*}
    \widetilde{\theta} = \frac{1}{N} \sum_{t=1}^{N} w_{\mathbf{G}} f\big(\mathbf{G}_t\!\left(\mathbf{a},\mathbf{b}\right)\big).
\end{equation*}
However, in large-scale problems, computing $\left|\mathcal{G}\!\left(\mathbf{a},\mathbf{b}\right)\right|$ by Proposition~\ref{prop:N(a,b,F(k))} may be infeasible due to resource constraints. In such cases, a modified estimator (as in \citealt{Blitzstein2011}) can be employed:
\begin{equation*}
    \widehat{\theta} = \frac{\sum_{t=1}^{N} w_{\mathbf{G}} f\big(\mathbf{G}_t\!\left(\mathbf{a},\mathbf{b}\right)\big)}{\sum_{t=1}^{N}w_{\mathbf{G}}},
\end{equation*}
where the unknown factor $\left|\mathcal{G}\!\left(\mathbf{a},\mathbf{b}\right)\right|$ cancels out. As a ratio estimator, $\widehat{\theta}$ is biased but consistent as $N \to \infty$ \citep{Glasserman2023}.

\section{Technical details for the directed graph setting}\label{appx:directed-details}

This section collects technical details for the directed graph setting, including the self-loop exclusion constraint, the proofs of the main theoretical results, and the pseudocode of the proposed directed graph algorithms.

\subsection{Self-loop exclusion constraint for directed graphs}\label{appx:exclude-self-loops}

To exclude self-loops, each $\left[\mathbf{b}\right]_j$ must not be connected to its counterpart node in $\mathbf{a}^{[j]}$. Specifically, for notational convenience, let $\left[\mathbf{a}\right]_j$ denote the counterpart node of $\left[\mathbf{b}\right]_j$ in $\mathbf{a}^{[j]}$, and let $k_j$ be the index of the group containing $\left[\mathbf{a}\right]_j$. Since $\left[\mathbf{b}\right]_j$ cannot connect to $\left[\mathbf{a}\right]_j$, the lower bound of $\mathscr{F}\!\left(k_j\right)$, denoted by $\ell_{k_j}$, must satisfy $\ell_{k_j} \leq m_{k_j} - 1$, where $m_{k_j}$ is the number of nodes in group $k_j$. Otherwise, if $\ell_{k_j} = m_{k_j}$, then $\left[\mathbf{b}\right]_j$ would be forced to connect to all nodes in group $k_j$, including $\left[\mathbf{a}\right]_j$, thereby creating a self-loop. Recall from Theorem~\ref{thm:interval} that
\begin{equation*}
    \ell_{k_j} = \max \left\{\left[\mathbf{b}\right]_j^{(k_j)} - M^{(k_j)}, 0\right\}.
\end{equation*}
Therefore, the self-loop exclusion constraint can be expressed as
\begin{equation*}
    \begin{split}
        \left[\mathbf{b}\right]_j^{(k_j)} - M^{(k_j)} &\leq m_{k_j} - 1,\\
       \left[\mathbf{b}\right]_j - \Delta - M^{(k_j)} &\leq m_{k_j} - 1,\\
        \left[\mathbf{b}\right]_j - M^{(k_j)} - m_{k_j} + 1 &\leq \Delta,
    \end{split}
\end{equation*}
where $\Delta  = {\textstyle \sum_{k=1}^{k_j-1}\mathscr{F}\!\left(k\right)}$ denotes the total number of nodes in the first $k_j-1$ groups of $\mathbf{a}^{[j]}$ that have already been connected to $\left[\mathbf{b}\right]_j$. In principle, one could attempt to enforce the above inequality by retrospectively adjusting $\Delta$ whenever it is violated. However, doing so would require backtracking over previously constructed partial connections, leading to significant computational inefficiency.

\subsection{Proof of Theorem \ref{thm:directed-interval}}\label{appx:directed-interval-proof}

\begin{proof}{Proof.}
For the augmented degree sequences, since $\left[\mathbf{b}\right]_j$ must connect to $\left[\mathbf{a}\right]_j$, the upper bound of $\mathscr{F}\!\left(k_j\right)$, denoted by $u_{k_j}$, must satisfy $u_{k_j} \geq 1$. Otherwise, if $u_{k_j} = 0$, then $\left[\mathbf{b}\right]_j$ would be forced not to connect to any node in group $k_j$, including $\left[\mathbf{a}\right]_j$, thereby preventing the creation of a self-loop. Recall from Theorem~\ref{thm:interval} that $u_{k_j} = \min \left\{D_{\mathrm{left}}^{(k_j)} + D_{\mathrm{now}}^{(k_j)} + D_{\mathrm{right}}^{(k_j)}, m_{k_j}\right\}$, thus the self-loop connection constraint can be expressed as
\begin{equation}\label{eq:directed_delta}
    \begin{split}
        D_{\mathrm{left}}^{(k_j)} + D_{\mathrm{now}}^{(k_j)} + D_{\mathrm{right}}^{(k_j)} &\geq 1,\\
        D_{\mathrm{left}}^{(1,k_j)} - \Delta + D_{\mathrm{now}}^{(1,k_j)} + D_{\mathrm{right}}^{(1,k_j)}  &\geq 1,\\
        D_{\mathrm{left}}^{(1,k_j)} + D_{\mathrm{now}}^{(1,k_j)} + D_{\mathrm{right}}^{(1,k_j)} - 1 &\geq \Delta,
    \end{split}
\end{equation}
where $\Delta  = {\textstyle \sum_{k=1}^{k_j-1}\mathscr{F}\!\left(k\right)}$ denotes the total number of nodes in the first $k_j-1$ groups of $\mathbf{a}^{[j]}$ already  connected to $\left[\mathbf{b}\right]_j$,
\begin{equation*}
    D_{\mathrm{left}}^{(1,k_j)} = \sum_{h=1}^{{\alpha_{k_j}}-1} \left[\mathscr{D} \!\left(\mathbf{z}^{(1)}, \mathbf{b}^{[j+1]}\right)\right]_h, \quad D_{\mathrm{now}}^{(1,k_j)}= \left[\mathscr{D} \!\left(\mathbf{z}^{(1)}, \mathbf{b}^{[j+1]}\right)\right]_{\alpha_{k_j}},
\end{equation*}
and
\begin{equation*}
    D_{\mathrm{right}}^{(1,k_j)}= \min\left\{\min_{r} \sum_{h=\alpha_{k_j} + 1}^{r} \left[\mathscr{D} \!\left(\mathbf{z}^{(1)}, \mathbf{b}^{[j+1]}\right)\right]_h, 0\right\}.
\end{equation*}
The second inequality follows from the fact that $D_{\mathrm{left}}^{(k_j)} = D_{\mathrm{left}}^{(1,k_j)} - \Delta$, $D_{\mathrm{now}}^{(k_j)} = D_{\mathrm{now}}^{(1,k_j)}$ and $D_{\mathrm{right}}^{(k_j)} = D_{\mathrm{right}}^{(1,k_j)}$. Both $D_{\mathrm{left}}^{(1,k_j)}$ and $D_{\mathrm{left}}^{(k_j)}$ consider the sum of the first $\alpha_{k_j}-1$ components of $\mathbf{z}^{(1)}$ and $\mathbf{z}^{(k)}$, respectively; hence, they differ by $\Delta$, which represents the total reduction in these components from $\mathbf{z}^{(1)}$ to $\mathbf{z}^{(k)}$. In contrast, $D_{\mathrm{now}}^{(1,k_j)}$ and $D_{\mathrm{now}}^{(k_j)}$ consider the $\alpha_{k_j}$-th component and therefore remain unaffected by the changes in the preceding components, while $D_{\mathrm{right}}^{(1,k_j)}$ and $D_{\mathrm{right}}^{(k_j)}$ are similar.

To enforce the self-loop connection constraint for the augmented degree sequences, we first compute $u_{\Delta} = D_{\mathrm{left}}^{(1,k_j)} + D_{\mathrm{now}}^{(1,k_j)} + D_{\mathrm{right}}^{(1,k_j)} - 1$ before $\left[\mathbf{b}\right]_j$ starts connecting to the nodes in $\mathbf{a}^{[j]}$. This quantity represents the upper limit on the total number of nodes that $\left[\mathbf{b}\right]_j$ can connect to within the first $k_j-1$ groups. Then, we incorporate a dynamically updated upper limit $u_{\Delta}^{(k)}$ into the upper bound of each $\mathscr{F}\!\left(k\right)$ for $k=1,\ldots,k_j-1$, where $u_{\Delta}^{(k)} = u_{\Delta} - \sum_{k'=1}^{k-1}\mathscr{F}\!\left(k'\right)$ and $u_{\Delta}^{(1)} = u_{\Delta}$. Accordingly, the upper bound of $\mathscr{F}\!\left(k\right)$ for $k=1,\ldots,k_j-1$ is modified as
\begin{equation}\label{eq:self-loop}
    u_k = \min \left\{D_{\mathrm{left}}^{(k)} + D_{\mathrm{now}}^{(k)} + D_{\mathrm{right}}^{(k)}, m_{k}, u_{\Delta}^{(k)}\right\}.
\end{equation}
Note that $u_{\Delta}^{(k)}$ is greater than or equal to the lower bound $\ell_k$ of $\mathscr{F}\!\left(k\right)$ for $k=1,\ldots,k_j-1$, that is, $u_{\Delta}^{(k)} \geq \max \left\{\left[\mathbf{b}\right]_j^{(k)} - M^{(k)}, 0\right\}$, because if $\ell_k = \left[\mathbf{b}\right]_j^{(k)} - M^{(k)}$, then $\left[\mathbf{b}\right]_j$ would be forced to connect to all nodes in group $k_j$, including $\left[\mathbf{a}\right]_j$; in addition, $u_{\Delta}^{(k)} \geq 0$ holds trivially. Therefore, the proposed dynamic adjustment guarantees that the upper bound of $\mathscr{F}\!\left(k_j\right)$ satisfies $u_{k_j} \geq 1$. At the same time, the lower bound of $\mathscr{F}\!\left(k_j\right)$ is updated as
\begin{equation}\label{eq:self-loop2}
    \ell_{k_j} = \max \left\{\left[\mathbf{b}\right]_j^{(k_j)} - M^{(k_j)}, 1\right\}.
\end{equation}
Both \eqref{eq:self-loop} and \eqref{eq:self-loop2} ensure that $\left[\mathbf{b}\right]_j$ forms a self-loop with $\left[\mathbf{a}\right]_j$ under the augmented degree sequences, or equivalently, that $\left[\mathbf{b}\right]_j$ does not connect to $\left[\mathbf{a}\right]_j$ under the original degree sequences.

To enforce the single-degree connection constraint, we first count the number of degree-one nodes in group $1$ that have not yet formed a self-loop, denoted by $n_{\mathrm{one}}$, and incorporate this quantity into the upper bound of $\mathscr{F}\!\left(1\right)$, which is then updated as
\begin{equation}\label{eq:single-degree}
    u_1 =
    \begin{cases}
        \min \left\{D_{\mathrm{left}}^{(1)} + D_{\mathrm{now}}^{(1)} + D_{\mathrm{right}}^{(1)},\, m_{1} - n_{\mathrm{one}} + 1\right\}, & \text{if } k_j = 1, \\
        \min \left\{D_{\mathrm{left}}^{(1)} + D_{\mathrm{now}}^{(1)} + D_{\mathrm{right}}^{(1)},\, m_{1} - n_{\mathrm{one}},\, u_{\Delta}^{(1)}
        \right\}, & \text{if } k_j > 1.
    \end{cases}
\end{equation}

Hence, we conclude the theorem by \eqref{eq:self-loop}, \eqref{eq:self-loop2}, and \eqref{eq:single-degree}.\Halmos
\end{proof}

\subsection{Pseudocode of the DGE algorithm}\label{appx:alg-DGE}

The pseudocode of the DGE algorithm is presented in Algorithm~\ref{alg:DGE}.

\begin{algorithm}[htpb]
\caption{Directed Graph Enumeration (DGE) Algorithm}
\label{alg:DGE}
{\renewcommand{\baselinestretch}{0.625}\selectfont
\begin{algorithmic}[1]
\Require Degree sequences $\mathbf{a} \in \mathbb{N}_{+}^{m}$ and $\mathbf{b} \in \mathbb{N}_{+}^{n}$
\Ensure The set $\mathcal{G} \!\left(\mathbf{a}, \mathbf{b}\right)$ of all directed graphs consistent with the degree sequences $\mathbf{a}$ and $\mathbf{b}$

\State Initialize $\mathcal{G}(\mathbf{a},\mathbf{b}) \gets \emptyset$, $\widetilde{\mathbf{a}} \gets \mathbf{a} + \mathbf{1}$, $\widetilde{\mathbf{b}} \gets \mathbf{b} + \mathbf{1}$, an empty bipartite graph $\mathbf{G}(\widetilde{\mathbf{a}}, \widetilde{\mathbf{b}})$, two empty queues $\mathcal{Q}_{\mathrm{outer}}$ and $\mathcal{Q}_{\mathrm{inner}}$
\State Enqueue the pair $(\mathbf{G}(\widetilde{\mathbf{a}}, \widetilde{\mathbf{b}}), 1)$ into $\mathcal{Q}_{\mathrm{outer}}$ \Comment{Node index $j=1$}

\While{$\mathcal{Q}_{\mathrm{outer}} \neq \emptyset$} \Comment{Outer breadth-first search over $n$ nodes}
  \State Dequeue the first pair $(\mathbf{G}(\widetilde{\mathbf{a}}, \widetilde{\mathbf{b}}), j)$ from $\mathcal{Q}_{\mathrm{outer}}$ and record the node index $j$

  \If{$j > n$}
    
    \State Remove all self-loops from $\mathbf{G}(\widetilde{\mathbf{a}}, \widetilde{\mathbf{b}})$ to obtain a directed graph consistent with $\mathbf{a}$ and $\mathbf{b}$
    \State Add the resulting directed graph to $\mathcal{G} \!\left(\mathbf{a}, \mathbf{b}\right)$
    \State \textbf{continue}
  \EndIf
  \State Group $\widetilde{\mathbf{a}}^{[j]}$ into $p$ groups $(\alpha_k, m_k)$ with $\alpha_1 < \cdots < \alpha_p$
  \State Enqueue the pair $(\mathbf{G}(\widetilde{\mathbf{a}}, \widetilde{\mathbf{b}}), 1)$ into $\mathcal{Q}_{\mathrm{inner}}$ \Comment{Group index $k=1$}

  \While{$\mathcal{Q}_{\mathrm{inner}} \neq \emptyset$} \Comment{Inner breadth-first search over $p$ groups for each node}
    \State Dequeue the first pair $(\mathbf{G}(\widetilde{\mathbf{a}}, \widetilde{\mathbf{b}}), k)$ from $\mathcal{Q}_{\mathrm{inner}}$ and record the group index $k$
    \If{$k > p$}
      \If{\textsc{Feasible}$(\mathbf{G}(\widetilde{\mathbf{a}}, \widetilde{\mathbf{b}}), j)$}
      \Comment{The feasibility verification step}
        \State Enqueue the pair $(\mathbf{G}(\widetilde{\mathbf{a}}, \widetilde{\mathbf{b}}), j + 1)$ into $\mathcal{Q}_{\mathrm{outer}}$
      \EndIf
    \State \textbf{continue}
    \EndIf
    
    \State Compute the interval $[\ell_k, u_k]$ via \eqref{eq:directed-interval-1}--\eqref{eq:directed-interval-3}
    \For{$\mathscr{F}\!\left(k\right) = \ell_k,\,\dots,u_k$}
      \If{$k=k_j$}
        \State Generate all $\binom{m_k}{\mathscr{F}\!\left(k\right)}$ connection configurations with self-loop
      \Else
        \State Generate all $\binom{m_k}{\mathscr{F}\!\left(k\right)}$ connection configurations
      \EndIf
      
      \For{each connection configuration}
        \State Copy $\mathbf{G}(\widetilde{\mathbf{a}}, \widetilde{\mathbf{b}})$ and denote it by $\mathbf{G}'(\widetilde{\mathbf{a}}, \widetilde{\mathbf{b}})$
        \State Append the connection configuration to $\mathbf{G}'(\widetilde{\mathbf{a}}, \widetilde{\mathbf{b}})$
        \State Enqueue the pair $(\mathbf{G}'(\widetilde{\mathbf{a}}, \widetilde{\mathbf{b}}), k+1)$ into $\mathcal{Q}_{\mathrm{inner}}$
      \EndFor
    \EndFor
  \EndWhile
\EndWhile

\State \textbf{return} $\mathcal{G} \!\left(\mathbf{a}, \mathbf{b}\right)$
\end{algorithmic}
}
\end{algorithm}

\subsection{Proof of Proposition \ref{prop:ND(a,b,F(k))}}\label{appx:T3-proof}

\begin{proof}{Proof.}
\eqref{eq:ND_K-N_Ki} follows from the fact that $\mathscr{F}\!\left(k\right) \in [\ell_k,u_k]$ and there are $\binom{m_k}{\mathscr{F}_i\!\left(k\right)}$ connection configurations for $\mathscr{F}_i\!\left(k\right)$. \eqref{eq:ND(a,b,k)} follows from the fact that the process moves to group $k+1$ in $\mathbf{a}^{[j]}$ when the node $\left[\mathbf{b}\right]_j$ has completed its connections to the nodes in group~$k$ for $k<p$, and moves to group $1$ in $\mathbf{a}^{[j+1]}$ when $\left[\mathbf{b}\right]_j$ has completed its connections to the nodes in group~$p$. Hence, we conclude the proof of the proposition.\Halmos
\end{proof}

\subsection{Pseudocode of the UDGS algorithm}\label{appx:alg-UDGS}

The pseudocode of the UDGS algorithm is presented in Algorithm~\ref{alg:UDGS}.

\begin{algorithm}[htpb]
\caption{Uniform Directed Graph Sampling (UDGS) Algorithm}
\label{alg:UDGS}
{\renewcommand{\baselinestretch}{0.625}\selectfont
\begin{algorithmic}[1]
\Require Degree sequences $\mathbf{a} \in \mathbb{N}_{+}^{m}$ and $\mathbf{b} \in \mathbb{N}_{+}^{n}$
\Ensure A directed graph consistent with $\mathbf{a}$ and $\mathbf{b}$

\State Initialize $\widetilde{\mathbf{a}} \gets \mathbf{a} + \mathbf{1}$, $\widetilde{\mathbf{b}} \gets \mathbf{b} + \mathbf{1}$, an empty bipartite graph $\mathbf{G}(\widetilde{\mathbf{a}}, \widetilde{\mathbf{b}})$
\For{$j = 1$ \textbf{to} $n$}
  \State Group $\widetilde{\mathbf{a}}^{[j]}$ into $p$ groups $(\alpha_k, m_k)$ with $\alpha_1 < \alpha_2 < \cdots < \alpha_p$

  \For{$k = 1$ \textbf{to} $p$}
    \State Compute the interval $[\ell_k,u_k]$ via \eqref{eq:directed-interval-1}--\eqref{eq:directed-interval-3}
    \State Generate all connection configurations for each $\mathscr{F} \!\left(k\right) \in [\ell_k,u_k]$
    \State Obtain the exact weight vector $\mathbf{w}^{\mathrm{UD}}$ via \eqref{eq:directed_weight-for-uniform}
    \State Sample one connection configuration according to $\mathbf{w}^{\mathrm{UD}}$
    \State Realize the sampled connection configuration for the node $[\widetilde{\mathbf{b}}]_j$ in $\mathbf{G}(\widetilde{\mathbf{a}}, \widetilde{\mathbf{b}})$
  \EndFor
\EndFor

\State Remove all self-loops from $\mathbf{G}(\widetilde{\mathbf{a}}, \widetilde{\mathbf{b}})$ to obtain a directed graph consistent with $\mathbf{a}$ and $\mathbf{b}$

\State \textbf{return} the resulting directed graph
\end{algorithmic}
}
\end{algorithm}

\subsection{Pseudocode of the EDGS algorithm}\label{appx:alg-EDGS}

The pseudocode of the EDGS algorithm is presented in Algorithm~\ref{alg:EDGS}.

\begin{algorithm}[htpb]
\caption{Efficient Directed Graph Sampling (EDGS) Algorithm}
\label{alg:EDGS}
{\renewcommand{\baselinestretch}{0.75}\selectfont
\begin{algorithmic}[1]
\Require Degree sequences $\mathbf{a} \in \mathbb{N}_{+}^{m}$ and $\mathbf{b} \in \mathbb{N}_{+}^{n}$
\Ensure A directed graph consistent with $\mathbf{a}$ and $\mathbf{b}$

\State Initialize $\widetilde{\mathbf{a}} \gets \mathbf{a} + \mathbf{1}$, $\widetilde{\mathbf{b}} \gets \mathbf{b} + \mathbf{1}$, an empty bipartite graph $\mathbf{G}(\widetilde{\mathbf{a}}, \widetilde{\mathbf{b}})$

\For{$j = 1$ \textbf{to} $n$}
  \State Group $\widetilde{\mathbf{a}}^{[j]}$ into $p$ groups $(\alpha_k, m_k)$ with $\alpha_1 < \alpha_2 < \cdots < \alpha_p$

  \For{$k = 1$ \textbf{to} $p$}
    \State Compute the interval $[\ell_k,u_k]$ via \eqref{eq:directed-interval-1}--\eqref{eq:directed-interval-3}
    \State Obtain the approximate weight vector $\mathbf{w}^{\mathrm{ED}}$ via \eqref{eq:weight-for-efficient}
    \State Sample $\mathscr{F} \!\left(k\right) \in [\ell_k,u_k]$ according to $\mathbf{w}^{\mathrm{ED}}$
    \If{$k=k_j$}
      \State Connect node $[\widetilde{\mathbf{b}}]_j$ with its counterpart node to form a self-loop
      \State Uniformly select $\mathscr{F} \!\left(k\right)-1$ distinct remaining nodes in group $k$ and connect them to the node $[\widetilde{\mathbf{b}}]_j$ in $\mathbf{G}(\widetilde{\mathbf{a}}, \widetilde{\mathbf{b}})$
    \Else
      \State Uniformly select $\mathscr{F} \!\left(k\right)$ distinct nodes in group $k$ and connect them to the node $[\widetilde{\mathbf{b}}]_j$ in $\mathbf{G}(\widetilde{\mathbf{a}}, \widetilde{\mathbf{b}})$
  \EndIf
  \EndFor
  \If{\textbf{not} \textsc{Feasible}$(\mathbf{G}(\widetilde{\mathbf{a}}, \widetilde{\mathbf{b}}), j)$}
  \Comment{The feasibility verification step}
    \State 
    Resample from $j=1$
  \EndIf
\EndFor

\State Remove all self-loops from $\mathbf{G}(\widetilde{\mathbf{a}}, \widetilde{\mathbf{b}})$ to obtain a directed graph consistent with $\mathbf{a}$ and $\mathbf{b}$

\State \textbf{return} the resulting directed graph
\end{algorithmic}
}
\end{algorithm}

\section{Pseudocode of the undirected graph algorithms}\label{appx:alg-undirected}

This section presents the pseudocode of the proposed undirected graph algorithms.

\subsection{Pseudocode of the UGE algorithm}\label{appx:alg-UGE}

The pseudocode of the UGE algorithm is presented in Algorithm~\ref{alg:UGE}.

\begin{algorithm}[htpb]
\caption{Undirected Graph Enumeration (UGE) Algorithm}
\label{alg:UGE}
{\renewcommand{\baselinestretch}{0.85}\selectfont
\begin{algorithmic}[1]
\Require Degree sequences $\mathbf{a} = \mathbf{b} \in \mathbb{N}_{+}^{m}$
\Ensure The set $\mathcal{G} \!\left(\mathbf{a}, \mathbf{b}\right)$ of all undirected graphs consistent with the degree sequences $\mathbf{a}$ and $\mathbf{b}$

\State Initialize $\mathcal{G}(\mathbf{a},\mathbf{b}) \gets \emptyset$, $\widetilde{\mathbf{a}} \gets \mathbf{a} + \mathbf{1}$, $\widetilde{\mathbf{b}} \gets \mathbf{b} + \mathbf{1}$, an empty bipartite graph $\mathbf{G}(\widetilde{\mathbf{a}}, \widetilde{\mathbf{b}})$, two empty queues $\mathcal{Q}_{\mathrm{outer}}$ and $\mathcal{Q}_{\mathrm{inner}}$
\State Enqueue the pair $(\mathbf{G}(\widetilde{\mathbf{a}}, \widetilde{\mathbf{b}}), 1)$ into $\mathcal{Q}_{\mathrm{outer}}$ \Comment{Node index $j=1$}

\While{$\mathcal{Q}_{\mathrm{outer}} \neq \emptyset$} \Comment{Outer breadth-first search over $n$ nodes}
  \State Dequeue the first pair $(\mathbf{G}(\widetilde{\mathbf{a}}, \widetilde{\mathbf{b}}), j)$ from $\mathcal{Q}_{\mathrm{outer}}$ and record the node index $j$

  \If{$j > m$}
    
    \State Remove all self-loops from $\mathbf{G}(\widetilde{\mathbf{a}}, \widetilde{\mathbf{b}})$ to obtain an undirected graph satisfying $\mathbf{a}$ and $\mathbf{b}$
    \State Add the resulting undirected graph to $\mathcal{G} \!\left(\mathbf{a}, \mathbf{b}\right)$
    \State \textbf{continue}
  \EndIf
  \State Find an index $j^\star$ of a node with the largest remaining degree in $\widetilde{\mathbf{b}}^{[j]}$
  \State Swap the indices $j$ and $j^\star$ in $\widetilde{\mathbf{b}}^{[j]}$
  \State Group $\widetilde{\mathbf{a}}^{[j]}$ into $p$ groups $(\alpha_k, m_k)$ with $\alpha_1 < \cdots < \alpha_p$
  \State Enqueue the pair $(\mathbf{G}(\widetilde{\mathbf{a}}, \widetilde{\mathbf{b}}), 1)$ into $\mathcal{Q}_{\mathrm{inner}}$ \Comment{Group index $k=1$}

  \While{$\mathcal{Q}_{\mathrm{inner}} \neq \emptyset$} \Comment{Inner breadth-first search over $p$ groups for each node}
    \State Dequeue the first pair $(\mathbf{G}(\widetilde{\mathbf{a}}, \widetilde{\mathbf{b}}), k)$ from $\mathcal{Q}_{\mathrm{inner}}$ and record the group index $k$
    \If{$k > p$}
      \State Establish symmetric connections for the node $\left[\widetilde{\mathbf{a}}\right]_j$ corresponding to those of $[\widetilde{\mathbf{b}}]_j$
      \If{\textsc{Feasible}$(\mathbf{G}(\widetilde{\mathbf{a}}, \widetilde{\mathbf{b}}), j)$}
      \Comment{The feasibility verification step}
        \State Enqueue the pair $(\mathbf{G}(\widetilde{\mathbf{a}}, \widetilde{\mathbf{b}}), j + 1)$ into $\mathcal{Q}_{\mathrm{outer}}$
      \EndIf
    \State \textbf{continue}
    \EndIf
    
    \State Compute the interval $[\ell_k, u_k]$ via \eqref{eq:directed-interval-1}--\eqref{eq:directed-interval-3}
    \For{$\mathscr{F}\!\left(k\right) = \ell_k,\,\dots,u_k$}
      \If{$k=k_j$}
        \State Generate all $\binom{m_k}{\mathscr{F}\!\left(k\right)}$ connection configurations with self-loop
      \Else
        \State Generate all $\binom{m_k}{\mathscr{F}\!\left(k\right)}$ connection configurations
      \EndIf
      
      \For{each connection configuration}
        \State Copy $\mathbf{G}(\widetilde{\mathbf{a}}, \widetilde{\mathbf{b}})$ and denote it by $\mathbf{G}'(\widetilde{\mathbf{a}}, \widetilde{\mathbf{b}})$
        \State Append the connection configuration to $\mathbf{G}'(\widetilde{\mathbf{a}}, \widetilde{\mathbf{b}})$
        \State Enqueue the pair $(\mathbf{G}'(\widetilde{\mathbf{a}}, \widetilde{\mathbf{b}}), k+1)$ into $\mathcal{Q}_{\mathrm{inner}}$
      \EndFor
    \EndFor
  \EndWhile
\EndWhile

\State \textbf{return} $\mathcal{G} \!\left(\mathbf{a}, \mathbf{b}\right)$
\end{algorithmic}
}
\end{algorithm}

\subsection{Pseudocode of the UUGS algorithm}\label{appx:alg-UUGS}

The pseudocode of the UUGS algorithm is presented in Algorithm~\ref{alg:UUGS}.

\begin{algorithm}[htpb]
\caption{Uniform Undirected Graph Sampling (UUGS) Algorithm}
\label{alg:UUGS}
{\renewcommand{\baselinestretch}{0.7}\selectfont
\begin{algorithmic}[1]
\Require Degree sequences $\mathbf{a} = \mathbf{b} \in \mathbb{N}_{+}^{m}$
\Ensure A undirected graph consistent with the degree sequences $\mathbf{a}$ and $\mathbf{b}$

\State Initialize $\widetilde{\mathbf{a}} \gets \mathbf{a} + \mathbf{1}$, $\widetilde{\mathbf{b}} \gets \mathbf{b} + \mathbf{1}$, an empty bipartite graph $\mathbf{G}(\widetilde{\mathbf{a}}, \widetilde{\mathbf{b}})$
\For{$j = 1$ \textbf{to} $m$}
  \State Find an index $j^\star$ of a node with the largest remaining degree in $\widetilde{\mathbf{b}}^{[j]}$
  \State Swap the indices $j$ and $j^\star$ in $\widetilde{\mathbf{b}}^{[j]}$
  \State Group $\widetilde{\mathbf{a}}^{[j]}$ into $p$ groups $(\alpha_k, m_k)$ with $\alpha_1 < \alpha_2 < \cdots < \alpha_p$

  \For{$k = 1$ \textbf{to} $p$}
    \State Compute the interval $[\ell_k,u_k]$ via \eqref{eq:directed-interval-1}--\eqref{eq:directed-interval-3}
    \State Generate all connection configurations for each $\mathscr{F} \!\left(k\right) \in [\ell_k,u_k]$
    \State Obtain the exact weight vector $\mathbf{w}^{\mathrm{UU}}$ via \eqref{eq:directed_weight-for-uniform} and the symmetric connection step
    \State Sample one connection configuration according to $\mathbf{w}^{\mathrm{UU}}$
    \State Realize the sampled connection configuration for the node $[\widetilde{\mathbf{b}}]_j$ in $\mathbf{G}(\widetilde{\mathbf{a}}, \widetilde{\mathbf{b}})$
  \EndFor
  \State Establish symmetric connections for the node $\left[\widetilde{\mathbf{a}}\right]_j$ corresponding to those of $[\widetilde{\mathbf{b}}]_j$
\EndFor

\State Remove all self-loops from $\mathbf{G}(\widetilde{\mathbf{a}}, \widetilde{\mathbf{b}})$ to obtain a directed graph consistent with $\mathbf{a}$ and $\mathbf{b}$

\State \textbf{return} the resulting directed graph
\end{algorithmic}
}
\end{algorithm}

\subsection{Pseudocode of the EUGS algorithm}\label{appx:alg-EUGS}

The pseudocode of the EUGS algorithm is presented in Algorithm~\ref{alg:EUGS}.

\begin{algorithm}[htpb]
\caption{Efficient Undirected Graph Sampling (EUGS) Algorithm}
\label{alg:EUGS}
{\renewcommand{\baselinestretch}{0.7}\selectfont
\begin{algorithmic}[1]
\Require Degree sequences $\mathbf{a} = \mathbf{b} \in \mathbb{N}_{+}^{m}$
\Ensure A undirected graph consistent with the degree sequences $\mathbf{a}$ and $\mathbf{b}$

\State Initialize $\widetilde{\mathbf{a}} \gets \mathbf{a} + \mathbf{1}$, $\widetilde{\mathbf{b}} \gets \mathbf{b} + \mathbf{1}$, an empty bipartite graph $\mathbf{G}(\widetilde{\mathbf{a}}, \widetilde{\mathbf{b}})$

\For{$j = 1$ \textbf{to} $m$}
  \State Find an index $j^\star$ of a node with the largest remaining degree in $\widetilde{\mathbf{b}}^{[j]}$
  \State Swap the indices $j$ and $j^\star$ in $\widetilde{\mathbf{b}}^{[j]}$
  \State Group $\widetilde{\mathbf{a}}^{[j]}$ into $p$ groups $(\alpha_k, m_k)$ with $\alpha_1 < \alpha_2 < \cdots < \alpha_p$

  \For{$k = 1$ \textbf{to} $p$}
    \State Compute the interval $[\ell_k,u_k]$ via \eqref{eq:directed-interval-1}--\eqref{eq:directed-interval-3}
    \State Obtain the approximate weight vector $\mathbf{w}^{\mathrm{EU}}$ via \eqref{eq:weight-for-efficient}
    \State Sample $\mathscr{F} \!\left(k\right) \in [\ell_k,u_k]$ according to $\mathbf{w}^{\mathrm{EU}}$
    \If{$k=k_j$}
      \State Connect node $[\widetilde{\mathbf{b}}]_j$ with its counterpart node to form a self-loop
      \State Uniformly select $\mathscr{F} \!\left(k\right)-1$ distinct remaining nodes in group $k$ and connect them to the node $[\widetilde{\mathbf{b}}]_j$ in $\mathbf{G}(\widetilde{\mathbf{a}}, \widetilde{\mathbf{b}})$
    \Else
      \State Uniformly select $\mathscr{F} \!\left(k\right)$ distinct nodes in group $k$ and connect them to the node $[\widetilde{\mathbf{b}}]_j$ in $\mathbf{G}(\widetilde{\mathbf{a}}, \widetilde{\mathbf{b}})$
  \EndIf
  \EndFor
  \State Establish symmetric connections for the node $\left[\widetilde{\mathbf{a}}\right]_j$ corresponding to those of $[\widetilde{\mathbf{b}}]_j$
  \If{\textbf{not} \textsc{Feasible}$(\mathbf{G}(\widetilde{\mathbf{a}}, \widetilde{\mathbf{b}}), j)$}
  \Comment{The feasibility verification step}
    \State 
    Resample from $j=1$
  \EndIf
\EndFor

\State Remove all self-loops from $\mathbf{G}(\widetilde{\mathbf{a}}, \widetilde{\mathbf{b}})$ to obtain an undirected graph consistent with $\mathbf{a}$ and $\mathbf{b}$

\State \textbf{return} the resulting undirected graph
\end{algorithmic}
}
\end{algorithm}

\section{Pseudocode of the brute-force algorithm}\label{appx:alg-brute-force}

The pseudocode of the brute-force algorithm is presented in Algorithm~\ref{alg:brute-force}.

\begin{algorithm}[htpb]
\caption{Brute-Force Algorithm}
\label{alg:brute-force}
{\renewcommand{\baselinestretch}{0.8}\selectfont
\begin{algorithmic}[1]
\Require Degree sequences $\mathbf{a} \in \mathbb{N}_{+}^{m}$ and $\mathbf{b} \in \mathbb{N}_{+}^{n}$
\Ensure The set $\mathcal{G} \!\left(\mathbf{a}, \mathbf{b}\right)$ of all bipartite graphs consistent with the degree sequences $\mathbf{a}$ and $\mathbf{b}$

\State Initialize $\mathcal{G}(\mathbf{a},\mathbf{b}) \gets \emptyset$

\If{${\textstyle \prod_{i=1}^{m}\binom{n}{[\mathbf{a}]_i}} < {\textstyle \prod_{j=1}^{n}\binom{m}{[\mathbf{b}]_j}}$}
  \State Set $\mathbf{d}^{(1)} \gets \mathbf{a}$, $\mathbf{d}^{(2)} \gets \mathbf{b}$ \Comment{First satisfy $\mathbf{a}$, then filter by $\mathbf{b}$}
\Else
  \State Set $\mathbf{d}^{(1)} \gets \mathbf{b}$, $\mathbf{d}^{(2)} \gets \mathbf{a}$ \Comment{First satisfy $\mathbf{b}$, then filter by $\mathbf{a}$}
\EndIf

\For{$t = 1$ \textbf{to} $\lvert \mathbf{d}^{(1)} \rvert$} \Comment{Construct all connection configurations that satisfy $\mathbf{d}^{(1)}$}

  \State Initialize $\mathcal{C}_t \gets \emptyset$
  \State Generate $\binom{\lvert \mathbf{d}^{(2)} \rvert}{\left[\mathbf{d}^{(1)}\right]_t}$ connection configurations \Comment{Choose $\left[\mathbf{d}^{(1)}\right]_t$ from $\lvert \mathbf{d}^{(2)} \rvert$ nodes in $\mathbf{d}^{(2)}$}
  \State Append these connection configurations to $\mathcal{C}_t$
\EndFor
\State Let $\mathcal{C}_T \gets \mathcal{C}_1 \times \mathcal{C}_2 \times \cdots \times \mathcal{C}_{\lvert \mathbf{d}^{(1)} \rvert}$ \Comment{Cartesian product of $\mathcal{C}_1,\, \mathcal{C}_2,\, \dots,\, \mathcal{C}_{\lvert \mathbf{d}^{(1)} \rvert}$}

\For{$\mathbf{C}_{\mathrm{one}} \in \mathcal{C}_{T}$}
  \State Initialize an empty bipartite graph $\mathbf{G}\!\left(\mathbf{a}, \mathbf{b}\right)$
  \State Append the connection configuration $\mathbf{C}_{\mathrm{one}}$ to $\mathbf{G}\!\left(\mathbf{a}, \mathbf{b}\right)$ 
  \If{$\mathbf{G}\!\left(\mathbf{a}, \mathbf{b}\right)$ \textbf{satisfies} $\mathbf{d}^{(2)}$}
    \State Append $\mathbf{G}\!\left(\mathbf{a}, \mathbf{b}\right)$ to $\mathcal{G}(\mathbf{a},\mathbf{b})$
  \EndIf
\EndFor

\State \textbf{return} $\mathcal{G}(\mathbf{a},\mathbf{b})$
\end{algorithmic}
}
\end{algorithm}

\end{document}